\documentclass[11pt]{article}

\usepackage{multirow}
\usepackage{bigstrut}

\usepackage{algorithm, algorithmic}

\usepackage{float,graphicx,epstopdf}
\usepackage{bbm}  % for the indicator function \mathbbm{1}
\usepackage{xcolor}
\usepackage{hyperref}
\usepackage{dsfont,amssymb,amsmath,subfigure, graphicx,enumitem} %,etaremune,savesym}
\usepackage{amsfonts,dsfont,mathtools, mathrsfs,amsthm,fancyhdr} 
\usepackage[amssymb, thickqspace]{SIunits}
\usepackage{multirow,booktabs,threeparttable,color}
\usepackage{setspace,graphicx,epstopdf,amsmath,amsfonts,amssymb,amsthm}
\usepackage{verbatim}
\usepackage{hyperref}
\usepackage{adjustbox}
\usepackage{float}
\usepackage[longnamesfirst]{natbib}
\usepackage{caption}
\usepackage{appendix}
\usepackage[normalem]{ulem}
\usepackage[margin=1in]{geometry}%

\usepackage[margin=10pt,labelfont=bf]{caption}

\newcommand\tcapfig[1]{\captionsetup{position=top, font=normalsize, labelfont=bf, textfont=normalfont, justification=centering, margin=0mm, aboveskip=2mm, belowskip=0mm, labelsep=colon, singlelinecheck=false}\caption{#1}}
\newcommand\bnotefig[1]{\captionsetup{position=bottom, font=footnotesize,  textfont=normalfont, margin=1mm, skip=2mm, justification=justified, singlelinecheck=false}\caption*{#1}}

\setlist{noitemsep} 
\makeatletter
\def\munderbar#1{\underline{\sbox\tw@{$#1$}\dp\tw@\z@\box\tw@}}
\makeatother

\newtheorem{theorem}{Theorem}[section]
\newtheorem{definition}[theorem]{Definition}
\newtheorem{lemma}[theorem]{Lemma}

\newtheorem{proposition}[theorem]{Proposition}

\newtheorem{assumption}[theorem]{Assumption}
\newtheorem{example}[theorem]{Example}
\newcommand{\be}{\begin{equation}}
\newcommand{\ee}{\end{equation}}
\newcommand{\bea}{\begin{equation*}\begin{aligned}}
\newcommand{\eea}{\end{aligned}\end{equation*}}
\newcommand{\ds}{\displaystyle}

\newcommand{\R}{\mathbb{R}}

\newcommand{\Min}{\min\limits_}
\newcommand{\Sup}{\sup\limits_}

\newcommand{\Tr}[1]{\Trace \big[ #1 \big]}

\newcommand{\wh}{\widehat}
\newcommand{\mc}{\mathcal}
\newcommand{\mbb}{\mathbb}

\newcommand{\PP}{\mbb P}

\DeclareMathOperator{\Trace}{Tr}
\DeclareMathOperator{\diag}{diag}

\DeclareMathOperator{\st}{s.t.}

\newcommand{\PD}{\mathbb{S}_{++}} % the set of positive semi-definite matrices of dimension p
\newcommand{\Let}{\triangleq}
\newcommand{\opt}{^\star}
\newcommand{\eps}{\varepsilon}

\newcommand{\Wass}{\mathds{W}}

\newcommand{\EE}{\mathds{E}}

\newcommand{\half}{\frac{1}{2}}

\newcommand{\cov}{\Sigma}
\newcommand{\covsa}{\wh \cov}
\newcommand{\m}{\mu}
\newcommand{\msa}{\wh \m}
\newcommand{\KL}{\mathrm{KL}}

\newcommand{\N}{\mc N}

\newcommand{\indep}{\perp \!\!\! \perp}

\graphicspath{{graph/}}

\begin{document}

\title{Bayesian Imputation with Optimal Look-Ahead-Bias and Variance Tradeoff\thanks{\scriptsize We gratefully acknowledge the financial support and feedback of MSCI and its data science team. Material in this paper is based upon work supported by the Air Force Office of Scientific Research under award number FA9550-20-1-0397. Additional support is gratefully acknowledged from NSF grants 1915967, 1820942, 1838676, and also from the China Merchant Bank.}}
\date{\today}
\author{Jose Blanchet \thanks{\scriptsize Stanford University, Department of Management Science \& Engineering, Email: jose.blanchet@stanford.edu.}
\and
Fernando Hernandez  \thanks{\scriptsize Stanford University, Department of Management Science \& Engineering, Email: fhernands01@gmail.com}
\and
Viet Anh Nguyen   \thanks{\scriptsize Chinese University of Hong Kong, Email: nguyen@se.cuhk.edu.hk}
\and
Markus Pelger\thanks{\scriptsize Stanford University, Department of Management Science \& Engineering, Email: mpelger@stanford.edu.}
\and
Xuhui Zhang\thanks{\scriptsize Stanford University, Department of Management Science \& Engineering, Email: xuhui.zhang@stanford.edu.}}

\onehalfspacing

\begin{titlepage}
\maketitle
\thispagestyle{empty}

\begin{abstract}
Missing time-series data is a prevalent problem in many prescriptive analytics models in operations management, healthcare and finance. Imputation methods for time-series data are usually applied to the full panel data with the purpose of training a prescriptive model for a downstream out-of-sample task. For example, the imputation of missing asset returns may be applied before estimating an optimal portfolio allocation. However, this practice can result in a look-ahead-bias in the future performance of the downstream task, and there is an inherent trade-off between the look-ahead-bias of using the \textit{entire} data set for imputation and the larger variance of using only the \textit{training} portion of the data set for imputation. By connecting layers of information revealed in time, we propose a Bayesian consensus posterior that fuses an arbitrary number of posteriors to optimize the variance and look-ahead-bias trade-off in the imputation. We derive tractable two-step optimization procedures for finding the optimal consensus posterior, with Kullback-Leibler divergence and Wasserstein distance as the dissimilarity measure between posterior distributions. We demonstrate in simulations and in an empirical study the benefit of our imputation mechanism for portfolio allocation with missing returns.
\end{abstract}

\vspace{1cm}


\end{titlepage}

\section{Introduction}
\label{sec:intro}

Missing time-series data is a prevalent problem in finance, healthcare and operations management. This problem is increasingly critical in the big-data era due to the vast amount of data collected nowadays and the reliance on data-driven analytics. Missing values in financial data may arise for a variety of reasons, including, for example, infrequent trading due to illiquidity, mixed frequency reporting, or improper data collection and storage. Similarly, time-series data in operations management may contain missing values due to sensor failures, human errors or lack of information. A standard and widely-adopted approach is to impute missing data, which has sparked a growing body of literature \citep{ref:little2002statistical,ref:van2012flexible} with over 150 practical implementations available \citep{ref:mayer2019rmisstastic}. How a decision-maker implements a missing value imputation may substantially impact any downstream task, most importantly, in risk management. Nevertheless, these implications of imputation for a downstream out-of-sample task have been neglected so far in the literature. 

To demonstrate the aforementioned implications, we use a leading example of the imputation of missing asset returns prior to solving a portfolio optimization model. Typically, a machine learning or portfolio optimization downstream task splits the given panel data into a train-validation-test set according to the time-ordering of underlying observations. Training corresponds to the left-most block portion of the panel, validation corresponds to a block in the middle of the panel, and testing is conducted on the right-end block of the panel. The ``globally observed data'' is the entire available panel data set representing all the data collected during the entire time horizon. Meanwhile, ``locally observed data'' refers to data available up to a specific time in the past (e.g.,~the timestamp corresponding to the last observation of the training section). Standard imputation procedures in practice usually apply a chosen imputation tool using the globally observed data. While being natural for reducing error in the imputation, ignoring the time order of a future downstream task poses a serious problem for portfolio benchmarking. In fact, a global imputation of the missing values is likely to create a spurious signal due to the temporal correlation between the training and the testing portion of the data. As a consequence, the imputed training portion can (and often will) pick up the spurious signal created by the imputation procedure. As portfolio selection procedures are carefully designed to exploit subtle signals that can lead to favorable future investment outcomes, using an globally imputation method can give a false sense of good performance for the portfolio, simply because it has exploited the unwanted spurious correlation. For example, if the mean return of an asset is large on the testing portion, a global imputation might increase the estimated mean on the training data and bias the stock selection. This is an example of an overfitting phenomenon created by ``looking into the future'' which investors have to mitigate. In summary, even if the portfolio analysis is done ``properly out-of-sample'' on the full panel of globally imputed values, the look-ahead-bias from the imputation can render the results invalid.

Our paper provides a new conceptual perspective by connecting the imputation with a downstream out-of-sample task. There exists an inherent trade-off between the look-ahead-bias of using the full data set for imputation and the larger variance in using only the training portion of the data. As far as we know, this paper is the first to introduce a systematic approach to studying this trade-off. Notice that the look-ahead-bias increases under imputation using globally observed data, and the variance increases under imputation using only locally observed data. To optimize this trade-off, we propose a novel Bayesian solution that fuses the use of globally observed data (i.e., data collected during the entire time horizon) and locally observed data (i.e., data collected up to a fixed time in the past). More specifically, we connect layers of information revealed in time, through a Bayesian consensus posterior that fuses an arbitrary number of posteriors to optimally control the variance and look-ahead-bias trade-off in the imputation.

Statistical and machine learning methods in finance and operations management have usually been very careful in avoiding a look-ahead-bias in the downstream task for a given panel data without missing values. For example, \cite{ref:Bertsimas2022}, \cite{Gu_etal2019ML} and \cite{Chen_etal2019DL} train their machine learning methods on the training data in the first part of the sample to estimate optimal policies and evaluate their performance on the out-of-sample test data at the end of the sample. They follow one common approach in the literature: they use only information from the training data to impute missing values. Their results could be improved by better-imputed values with a small look-ahead-bias. In a financial setting, \cite{Giglio2021} impute a panel of asset returns with matrix completion methods applied to the full panel of asset returns. Hence, their subsequent out-of-sample analysis of abnormal asset returns has an implicit in-sample component as future returns have been used in the imputation. These two cases of conservative imputation (without look-ahead-bias and with large variance) and  full data imputation (with look-ahead-bias and small variance) are both special cases of our Bayesian approach. Each of them corresponds to a different posterior. Our framework allows us to obtain a consensus posterior that is optimal in the sense of this trade-off. Hence, using at least some future information can be optimal if it has a negligible effect on the look-ahead-bias. This becomes particularly relevant if the amount of missing values is large.

Any reasonable formulation to optimize the trade-off between look-ahead-bias and variance should consider at least three ingredients. First, a representative downstream task that captures stylized features of broad interest. Second, a reasonable method to fuse or combine imputation estimators that are obtained using increasing layers of time information. Finally, a link to connect the first and second ingredients to optimally attain the desired bias-variance trade-off. The ultimate objective of this paper is to systematically produce these three ingredients in a meaningful, tractable and flexible way.
%%%%%%%%%%%%%%%%%%%%%5

To this end, we adopt a Bayesian imputation framework: this framework allows us to precisely define look-ahead-bias as deviation from a specific performance related to a reasonable downstream task, which can be considered as optimizing the expected reward, or returns. The expected reward is computed with the baseline posterior distribution obtained using only locally observed information. As it is standard in Bayesian imputation ~\citep{ref:little2002statistical,ref:rubin1987multiple}, we impute by sampling (multiple times) from the posterior distribution given the observed entries by creating multiple versions of the imputed data set which can be used to estimate the optimal posterior return. Our theoretical results focus on the fundamental case of estimating mean returns in order to clearly illustrate the novel conceptual ideas, but can be generalized to more general risk management objectives. 

The second ingredient involves fusing inferences arising from different layers of temporal information. Our Bayesian approach is convenient because we can optimally combine posterior distributions using a variational approach, namely, the weighted barycenter distribution. A barycenter is the result of an infinite dimensional optimization problem that finds a consensus distribution with minimal combined discrepancy among a non-parametric family of distributions. We use the (forward and backward) Kullback-Leibler divergences and the Wasserstein distance as discrepancy criteria to obtain three different ways to fuse posteriors, each posterior reflect a different layer of available information. These optimally fused distributions are indexed by weights corresponding to the relative importance given to each posterior to be combined. It is important to note that different fusing methods could be combined seamlessly with different downstream tasks, and the choice of fusing method is not tied to any particular choice of the downstream problem.

Finally, the third ingredient in our Bayesian method minimizes the variance of the optimal decision policy subjecting to an expected return criterion. The optimization is performed over all possible fusing weights under a given look-ahead-bias constraint. The constraint parameter can be chosen via cross-validation or by minimizing the estimated mean squared error over all consensus posteriors as a function of the constraint parameter.

In sum, the overall procedure gives the optimal utilization of different information layers in time, based on a user-defined bias-variance trade-off relative to the chosen downstream task, and uses an optimal consensus mechanism to fuse the different layers of information.

We illustrate our novel conceptual framework with the data generating distribution $Z_t=\theta+\epsilon_t$, where $\theta$ is a mean parameter and $\epsilon_t$ are the noise terms. This setup allows us to study a core problem in time-series data, namely the estimation of mean vectors for relatively short time horizons.\footnote{In application domain such as finance, assuming time-stationarity (in particular a constant mean) over moderate time horizons is a common practice, see~\citep{ref:tsay2010analysis}.} As all our empirics are pursued on short local windows, the time variation in means and variances are of second order. More generally, with higher frequency data the estimation error of the mean dominates the estimation error of the variance. Hence, the focus is on the accurate mean estimation. The noise $\epsilon_t$ can in general follow an $AR(p)$ process. A suitable joint prior on $\theta$ is imposed, after which the joint posterior conditional on observed entries can in principle be calculated. We aggregate the different posteriors corresponding to different layers of information using the variational approach prescribed by a set of weights. This set of weights is further optimized to find the optimal consensus posterior according to a bias-variance criterion. This additional step of optimization over importance weights differentiates our work from the usual Bayesian fusion setup, as most of the literature only considers a given choice of the weights~\citep{ref:dai2021bayesian,ref:claici2020model}.

 To make the two-step optimization procedure tractable, we assume that the noise $\epsilon_t$ are i.i.d.~Gaussian with a known covariance matrix and use conjugate prior distributions. This assumption makes the posterior of the missing entries a conditional Gaussian with mean sampled from the posterior of $\theta$, and thus we consider only aggregating the different posteriors of $\theta$. 
 
 We provide tractable reformulations of the two-step optimization procedure using the forward Kullback-Leibler (KL) divergence and Wasserstein distance, when imposing an additional simultaneous diagonalizability condition on the covariance matrices of the individual posteriors. While our reformulation using backward KL divergence does not require simultaneous diagonalizability, the resulting optimization problem is non-convex and may not be scalable. We thus advocate using the forward KL divergence and Wasserstein distance over the backward KL divergence, unless the simultaneous diagonalizability condition is seriously violated for the type of data at hand. 

The contributions of this paper can be summarized as follows:
\begin{enumerate}
    \item We identify and formalize a novel problem that has not been studied yet, namely, controlling the look-ahead-bias from imputation for a downstream task. We propose a novel conceptual framework to formalize and quantify the look-ahead-bias in the context of panel data imputation. We combine different layers of additional future information used in the imputation via optimal non-parametric consensus mechanisms among Bayesian posterior distributions to provide an optimal trade-off between the look-ahead-bias and the variance. 
    \item We illustrate the conceptual framework with a data-generating distribution with a constant mean vector over time. We propose a two-step optimization procedure for finding the optimal consensus posterior. The first step of optimization deals with finding the best possible (non-parametric) distribution that summarizes the different posterior models given a set of importance weights. The second step of optimization finds the optimal weights based on a bias-variance criterion that minimizes the variance while explicitly controlling the degree of the look-ahead-bias.
    \item To make the two-step procedures tractable, we further assume that the noise is i.i.d.~multivariate Gaussian with a known covariance matrix and use a conjugate prior distribution. Moreover, we discuss using the (forward and backward) KL-divergence and Wasserstein distance as the measure of dissimilarity between posterior distributions. We advise against the backward KL divergence unless in situations where the simultaneous diagonalizability condition on the covariance matrices of the individual posteriors is drastically violated. We use the simple model setup to provide an analytical and intuitive solution. We show in our empirical analysis that our solution can be easily applied and leads to improvements even if the underlying assumptions are violated. 
\end{enumerate}

 Finally, we empirically validate the benefits of our methodology with the downstream task of portfolio optimization. We show that the use of an optimal consensus posterior results in better performance of the regret of out-of-sample testing returns measured by what we call ``expected conditional mean squared error'', compared to the use of naive extreme posteriors (corresponding to focusing either using only the training part of the data for imputation or the full panel data for imputation). We demonstrate the power of our method on both simulated and real financial data, and under different practically relevant missing patterns including missing at random, dependent block missing, and missing by values.  
 
 The distinction between globally and locally observed data is important in many time-series data, for example in electronic health records and in demand data~\citep{ref:vulcano2012,ref:musalem2010}. More importantly, it generalizes beyond the time-series setting to data sets with a network structure, where each sample corresponds to a vertex in the network. The globally observed data is the entire collection of vertices in the network, while the locally observed data is a subset of the vertices defined by a neighborhood in this network. One may use the locally observed data as the training data set for a relevant downstream task while using the globally observed data for testing. The vertices that are inside of the neighborhood could potentially have a causal effect (with respect to a downstream task) on vertices outside, but the causality in the other direction is not permitted. Such a framework can be used as an approximate model for, e.g, the prediction of crime rates in different parts of a state, of the radiation level of surrounding cities after a damage to a nuclear power plant, or of the earthquake magnitude in a geographic region. If missing data is imputed globally using all information, an out-of-sample performance of the downstream task will inherently suffer from overfitting as the imputation extracts undesirable information. The imputation using only locally observed data again fails to utilize all valuable information to improve the variance.

This paper is organized as follows. In Section~\ref{sec:literature} we discuss the related literature on imputation methods. In Section~\ref{sec:bumi} we extend the bias-variance multiple imputation framework to more than two individual posteriors along with two basic modules of the framework: the posterior generator and the multiple imputation sampler. Section~\ref{sec:consensus} discusses the consensus mechanism module which complements our framework. Subsequently, in Section~\ref{sec:forward} and~\ref{sec:backward}, we use the forward and the backward KL divergence as a measure of dissimilarity between distributions and derive a reformulation of the consensus mechanism, while in Section~\ref{sec:Wass} we use the Wasserstein distance. Section~\ref{sec:numerics} reports the simulation and empirical results. The Appendix collects the proofs and additional empirical results.

%%%%%%%%%%%%%%%%%%%%%%%%%%%%%%%%%%%%
\section{Related Literature}
\label{sec:literature}

Methods for missing data imputation can be classified into two categories: single and multiple imputation. Single imputation refers to the generation and usage of a specific number (i.e., a best guess) in place of each missing value. Standard non-parametric methods for single imputation include $k$-nearest neighbor, mean/median imputation, smoothing, interpolation, and splines \citep{ref:kreindler2006effect}. Alternatively, single imputation can also be achieved using a parametric approach such as the popular EM algorithm of \cite{ref:dempster1977em}, Bayesian networks \citep{ref:lan2020multivariable} or $t$-copulas and transformed $t$-mixture models \citep{ref:craig2012missing}. Single imputation also includes the use of iterative methods in multivariate statistics literature to deal with missing values encountered in classical tasks such
as correspondence analysis (\cite{ref:leeuw1988correspondence}), multiple correspondence analysis, and multivariate analysis of mixed data sets \citep{ref:audigier2016principal}. Recently, there is a surge in the application of recurrent neural networks for imputation, as well as generative adversarial networks.\footnote{Examples of recurrent networks for imputation include \cite{ref:Che2018recurrent,ref:choi2016doctor,ref:zachary2016directly,ref:cao2018brits}, while \cite{ref:Yoon2018GAIN} is an example of the GAN approach.}

A closely related field of single imputation of panel data is matrix completion and matrix factorization.\footnote{This includes the work of \cite{ref:Cai2010singular,ref:candes2009eaxct,ref:candes2010power,ref:mazumder2010spectral,ref:keshavan2010matrix}.} \cite{ref:Chen2019inference} propose a de-biased estimator and an optimal construction of confidence intervals/regions for the missing entries and the low-rank factors. \cite{ref:xiong2019large} impute missing data with a latent factor model by applying principal component analysis to an adjusted covariance matrix estimated from partially observed panel data. They provide an asymptotic inferential theory and allow for general missing patterns, which can also depend on the underlying factor structure. This is important as \cite{Bryzgalova_etal2022} show in their comprehensive empirical analysis, that missing patterns of firm-specific and asset pricing information are impacted by the cross- and time-series dependency structure in the data. \cite{ref:bai2020matrix} suggest an imputation procedure that uses factors estimated from a tall block along with the re-rotated loadings estimated from a wide block. 

%\cite{ref:fithian2018flexible} uses a generalized nuclear norm penalty instead of a rank-constraint to model low-dimensional latent variables. 
% \cite{ref:susan2018matrix} uses matrix completion methods to impute the unobservable counterfactual outcomes and estimate the average treatment effect (ATE) \viet{do we need ATE anywhere else? If not, drop (ATE)} for panel data.

Multiple imputation generalizes the single imputation procedures in that the missing entries are filled-in with multiple guesses instead of one single guess or estimate, accounting for the uncertainty involved in imputation. The use of multiple imputation results in $m$ ($m>2$) completed data sets, each is then processed and analyzed with whatever statistical method can be used as if there is no missing data. Rubin's rule~\citep{ref:rubin1987multiple} is often used to combine the quantity of interest computed from each completed data set, which averages over the point estimates and uses a slightly more involved expression for the standard errors.

%Several advantages of multiple imputation have been recorded in the literature: it results in unbiased estimates, providing more validity than ad hoc approaches to missing data; it uses all available data, preserving sample size and statistical power; it may be used with standard statistical software; and, its results are readily interpreted~\cite{ref:lynnusing}. 

%A disadvantage of multiple imputation is that it assumes the data to be missing at random (MAR). 

%\cite{ref:easton2018inferences} found that earnings data is not missing at random, and hence there is a bias, which they call attrition bias, in the estimates of the earnings properties. \cite{ref:breunig2019information} proposed an approach when the random missing pattern depends on observable variables.~\cite{ref:ibrahim1999missing,ref:morikawa2017incomplete,ref:sportisse2020imputation} use logistic regression models to model missing data distribution.

The central idea underpinning the multiple imputation is the Bayesian framework. This framework samples the missing values from a posterior predictive distribution of the missing entries given the observed data. Multiple imputation is first developed for non-responses in survey sampling~\citep{ref:rubin1987multiple}. Since then, it has been extended to time-series data~\citep{ref:Honaker2010what,ref:brendan2016multiple}, where a key new element is to preserve longitudinal consistency in imputation. More recently,~\cite{ref:Honaker2010what} uses smooth basis functions to increase the smoothness of the imputation values while~\cite{ref:brendan2016multiple} uses a gap-filling sequence of imputation for categorical time-series data and produces smooth patterns of transitions.\footnote{The flexibility and predictive performance of multiple imputation method have been successfully demonstrated in a variety of data sets, such as Industry and Occupation Codes~\citep{ref:clifford1991multiple}, GDP in African nation~\citep{ref:Honaker2010what}, and concentrations of pollutants in the arctic~\citep{ref:hopke2001multiple}. }

This paper differs from standard multiple imputation literature in that we are interested in achieving an optimal trade-off between look-ahead-bias and variance by unifying the posteriors corresponding to different layers of information into a single ``consensus'' posterior, while most of the literature in multiple imputation would advocate using the full panel data. The unification of different posteriors into a single coherent inference, or (Bayesian) ``fusion'', has recently gained considerable interest in the research community~\citep{ref:dai2021bayesian,ref:claici2020model}, but often the different posteriors arise from the context of using paralleled computing machines~\citep{ref:steven2016bayes}, multiple sensors in multitarget tracking~\citep{ref:li2019cardinality,ref:da2020kullback}, or in model selection~\citep{ref:buchholz2020distributed}. In these settings, typically all of the distributions are given equal weight. Our approach contributes to the fusion literature by considering a downstream criterion that can be used to optimize the importance (or weight) of each distribution.

%%%%%%%%%%%%%%%%%%%%%%
\section{The Bias-Variance Multiple Imputation Framework}
\label{sec:bumi}

\textbf{Notations.} For any integers $K\leq L$, we use $[K]$ to denote the set $\{1, \ldots, K\}$ and $[K,L]$ to denote the set $\{K,
\ldots,L\}$. We use $\mc M$ to denote the space of all probability measures supported on $\R^n$. The set of all $n$-dimensional Gaussian measures is denoted by $\mc N_n$, and we use $\PD^n$ to represent the set of positive-definite matrices of dimension $n\times n$.

To facilitate the discussion, we borrow an exemplary application of portfolio allocation with missing asset returns. We assume a universe of $n$ assets over $T$ periods. The joint return of these $n$ assets at any specific time $t$ is represented by a random vector $Z_t \in \R^n$. The random vector $M_t \in \{0, 1\}^n$ indicates the pattern of the missing values at time $t$, that is,
\[
    \forall i \in [n]: \qquad (M_t)_i = \begin{cases}
        1 & \text{if the $i$-th component of $Z_t$ is missing},\\ 
        0 & \text{if the $i$-th component of $Z_t$ is observed}.
    \end{cases}
\]
Coupled with the indicator random variables $M_t$, we define the following projection operator $\mc P_{M_t}$ that projects\footnote{\label{fn:PMt} The projection operator $\mc P_{M_t}$ can be defined using a matrix-vector multiplication as follows. Upon observing $M_t$, let $\mc I = \{i \in [n]: (M_t)_i = 1\}$ be the set of indices of the missing entries. We assume without any loss of generality that $\mc I$ is ordered and has $d$ elements,  and that the $j$-th element of $\mc I$ can be assessed using $\mc I_j$. Define the matrix $P_{M_t} \in \{0, 1\}^{d \times n}$ with its $(j, \mc I_j)$ element being 1 for $j \in [d]$, and the remaining elements are zeros, then $\mc P_{M_t}(Z_t) = P_{M_t} Z_t$. The operator $\mc P_{M_t}^\perp$ can be defined with a matrix $P_{M_t}^\perp$ in a similar manner using the set $\mc I^\perp = \{ i \in [n]: (M_t)_i = 0\}$.} the latent variables $Z_t$ to its \textit{un}observed components $Y_t$
\[
    \mc P_{M_t}(Z_t) = Y_t.
\]
Consequentially, the orthogonal projection $\mc P_{M_t}^\perp$ projects $Z_t$ to the observed components $X_t$
\[
    \mc P_{M_t}^\perp(Z_t) = X_t.
\]

At the fundamental level, we assume the following generative model
\be \label{eq:generative}
    \forall t: \quad 
    \left\{
        \begin{array}{l}
        Z_t = \theta + \eps_t \\
        X_t = \mc P_{M_t}^\perp(Z_t),~ Y_t = \mc P_{M_t}(Z_t),
        \end{array}
    \right.
    \qquad 
    \text{where}
    \qquad 
    \eps_t \in \R^n,~M_t \in \{0, 1\}^n. 
\ee

When the generative model~\eqref{eq:generative} is executed over the $T$ periods, the observed variables $(X_t, M_t)$ are accumulated to produce a data set $\mathcal{D}$ of $n$ rows and $T$ columns. This data set is separated into a \textit{training} set consisting of  $T^{train}$ periods, and the remaining $T^{test} = T - T^{train}$ periods are used as the \textit{testing} set. We will be primarily concerned with the imputation of the training set.

\begin{figure}[h!]
\tcapfig{Illustration of a panel data with missing entries. The missing entries  are colored black. The observed components of $Z_t$ are denoted by $X_t$, and the unobserved components of $Z_t$ are denoted by $Y_t$.}
    \centering
    \includegraphics[width=0.55\columnwidth]{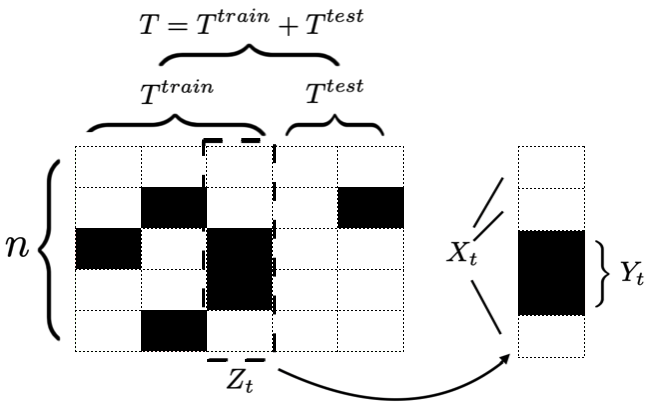}
    \label{fig:panelrep}
\end{figure}

In this paper, we pursue the Bayesian approach: we assume that the vector $\theta \in \R^n$ is unknown and is treated as a Bayesian parameter with a prior distribution $\pi_0$; we also treat $(Y_t)_{t\in[T]}$ as Bayesian parameters with appropriate priors. Based on the data set with missing entries $\mc D$, our strategy relies on computing the distribution of $(Y_t)_{t\in[T]}$ conditional on the data set $\mc D$, and then generating multiple imputations of $(Y_t)_{t\in[T]}$ by sampling from this distribution. To compute the posterior of $(Y_t)_{t\in[T]} | \mc D$, we first need to calculate a posterior distribution of the unknown mean vector $\theta$ conditional on $\mc D$, and then the distribution of $(Y_t)_{t\in[T]} | \mc D$ is obtained by marginalizing out $\theta$ from the distribution of $(Y_t)_{t\in[T]} | (\theta, \mc D)$.

We endeavor in this paper to explore a novel approach to generate multiply-imputed version of the data set. We make the following assumptions. 
\begin{assumption}[Observability]
\label{a:1cell}
We observe at least one value for each row in the training section of $\mc D$.
\end{assumption}
\begin{assumption}[Ignorability]
    \label{a:ignorable}
The missing pattern satisfies the ignorability assumption~\citep{ref:little2002statistical}, namely, the probability of $(M_t)_i=1~\forall i\in[n]~\forall t\in[T]$ does not depend on $(Y_t)_{t\in[T]}$ or $\theta$, conditional on $(X_t)_{t\in[T]}$.
\end{assumption} 

\begin{assumption}[Bayesian conjugacy] \label{a:bayes}
     The noise $\eps_t~\forall t\in[T]$ in the latent generative model~\eqref{eq:generative} is independently and identically distributed as $\mc N_n(0, \Omega)$, where $\Omega \in \PD^n$ is a known covariance matrix. The prior $\pi_0$ is either a non-informative flat prior or Gaussian with mean vector $\m_0 \in \R^n$ and covariance matrix $\cov_0 \in \PD^n$. The priors on $Y_t$ conditional on $\theta$ are independent across $t$, and the conditional distribution is Gaussian $\mc N_{\mathrm{dim}(Y_t)}(\theta_{Y_t},\Omega_{Y_t})$, where  $\theta_{Y_t} = \mc P_{M_t}(\theta)$, and $\Omega_{Y_t}$ is obtained by applying the projection operator $\mc P_{M_t}$ on $\Omega$.\footnote{Using the construction of $P_{M_t}$ in Footnote~\ref{fn:PMt}, we have $\Omega_{Y_t} = P_{M_t} \Omega P_{M_t}^\top$. Likewise for $\Omega_{X_t}$. The latent generative model~\eqref{eq:generative} implies that the likelihood of $X_t$ conditional on $(\theta,Y_t)$ is Gaussian $N_{\mathrm{dim}(X_t)} \big(\theta_{X_t} + \Omega_{X_t Y_t} \Omega_{Y_t}^{-1} (Y_t - \theta_{Y_t}), \Omega_{X_t} -  \Omega_{X_t Y_t} \Omega_{Y_t}^{-1} \Omega_{Y_t X_t} \big)$, where $\theta_{X_t} = \mc P_{M_t}^\perp(\theta)$, the matrix $\Omega_{Y_t X_t} \in \R^{\mathrm{dim}(Y_t) \times \mathrm{dim}(X_t)}$ is the covariance matrix between $Y_t$ and $X_t$ induced by $\Omega$, similarly for $\Omega_{X_tY_t}$.}
\end{assumption}

Following our discussion on the look-ahead-bias, we can observe that the distribution of $\theta|\mc D$ may carry look-ahead-bias originating from the testing portion of $\mc D$, and this look-ahead-bias will be internalized into the look-ahead-bias of $(Y_t)_{t\in[T^{train}]}| \mc D$. To mitigate this negative effect, it is crucial to re-calibrate the distribution of $\theta | \mc D$ to reduce the look-ahead-bias on the posterior distribution of the mean parameter $\theta$, with the hope that this mitigation will be spilled-over to similar mitigation of the look-ahead-bias on the imputation of $(Y_t)_{t\in[T^{train}]}$. Our BVMI framework pictured in Figure~\ref{fig:workflow} is designed to mitigate the look-ahead-bias by implementing three modules:
\begin{itemize}
    \item [(G)] The posterior \textit{G}enerator 
    \item [(C)] The \textit{C}onsensus mechanism
    \item [(S)] The multiple imputation \textit{S}ampler
\end{itemize}

\begin{figure}
\tcapfig{Illustration of workflow}
    \centering
    \includegraphics[width=\columnwidth]{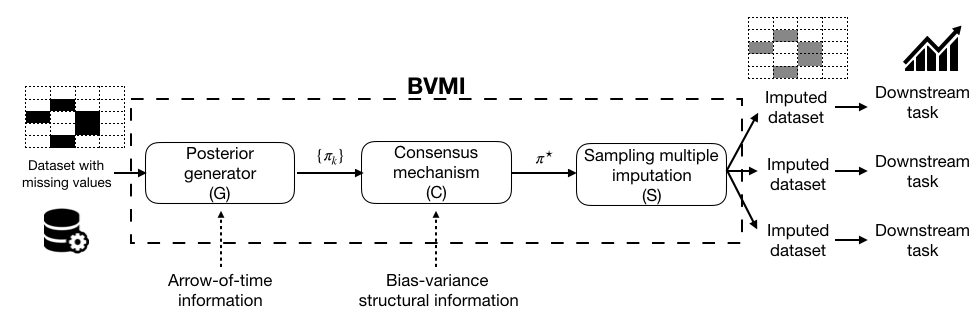}
    \bnotefig{This figure illustrates the workflow of our framework. The Bias-Variance Multiple Imputation (BVMI) framework (contained in the dashed box) receives a data set with missing values as an input, and generates multiple data sets as outputs. The posterior generator (G) can exploit time dependency, while the consensus mechanism (C) is bias-variance targeted.}
    \label{fig:workflow}
\end{figure}

For the remainder of this section, we will elaborate on the generator (G) in Subsection~\ref{sec:generator}, and we detail the sampler (S) in Subsection~\ref{sec:sampler}. The construction of the consensus mechanism (C) is more technical, and thus the full details on (C) will be provided in subsequent sections.

%%%%%%%%%%%%%%%%
\subsection{Posterior Generator}
\label{sec:generator}

We now discuss in more detail how to generate multiple posterior distributions by considering the time-dimension of the data set. To this end, we fix the desired number of posterior distributions to $K$, ($2 \le K \le T^{test}+1)$. Moreover, we assign the integer values $T_1, \ldots, T_K$ satisfying
\[
    T^{train} = T_1 < T_2 < \ldots < T_K = T
\]
to indicate the amount of data that is synthesized in each posterior. The data set truncated to time $T_k$ is denoted as $\mc D_k = \{ (X_t, M_t), t \in [T_k]\}$, note that $\mc D_1$ coincides with the training data set while $\mc D_K$ coincides with the whole data set $\mc D$. By construction, the sets $\mc D_k$ inherit a hierarchy in terms of time: for any $k < k'$, the data set $\mc D_{k}$ is part of $\mc D_{k'}$, and thus $\mc D_{k'}$ contains at least as much of information as $\mc D_{k}$.  The $k$-th posterior for $\theta$ is formulated conditional on $\mc D_k$. Under standard regularity conditions, the posterior $\pi_k$ can be computed from the prior distribution $\pi_0$ using Bayes' theorem~\cite[Theorem~1.31]{ref:schervish1995theory} as
\be \label{eq:bayes1}
    \frac{\mathrm{d} \pi_k}{\mathrm{d} \pi_0}(\theta | \mathcal{D}_k) = \frac{f(\mathcal{D}_k | \theta)}{\int_{\R^n} f(\mathcal{D}_k | \vartheta) \pi_0(\mathrm{d} \vartheta) },
\ee
where $\mathrm{d} \pi_k/ \mathrm{d} \pi_0$ represents the Radon-Nikodym derivative of $\pi_k$ with respect to $\pi_0$, and $f(\mc D_k | \theta)$ is the conditional density representing the likelihood of observing $\mc D_k$ given the parameter $\theta$. Under the generative model~\eqref{eq:generative}, and Assumption~\ref{a:ignorable}, we have
\[
    f(\mc D_k | \theta) = f_X\left( (X_t)_{t \in [T_k]} | \theta\right) \times f_M \left( (M_t)_{t \in [T_k]}|(X_t)_{t\in[T_k]},\theta\right) = f_X\left( (X_t)_{t \in [T_k]} | \theta\right) \times \tilde f_M \left( (M_t)_{t \in [T_k]}|(X_t)_{t\in[T_k]}\right)
\]
for some appropriate (conditional) likelihood functions~$f_X,f_M$ and $\tilde f_M$. In this case, the terms involving $(M_t)_{t \in [T_k]}$ in both the numerator and denominator of equation~\eqref{eq:bayes1} cancel each other out, hence \eqref{eq:bayes1} can be simplified to
\[
    \frac{\mathrm{d} \pi_k}{\mathrm{d} \pi_0}(\theta | \mathcal{D}_k) = \frac{f_X \left((X_t)_{t \in [T_k]} | \theta\right)}{\int_{\R^n} f_X\left((X_t)_{t \in [T_k]} | \vartheta\right) \pi_0(\mathrm{d} \vartheta) } ,
\]

Under Assumption~\ref{a:bayes}, the posterior distribution~$\pi_k$ can be computed more conveniently. Assumption~\ref{a:bayes} implies that conditioning on $\theta$, the latent variables $Z_t$ and $Z_{t'}$ are i.i.d.~following a Gaussian distribution with mean $\theta$ and covariance matrix $\Omega$ for any $t$, $t'$. Because the operator $\mc P_{M_t}^\perp$ is an affine transformation (see Footnote~\ref{fn:PMt}), Theorem~2.16 in~\cite{ref:fang1990symmetric} asserts that $X_t|\theta$ is Gaussian with mean $\theta_{X_t}$ and covariance matrix $\Omega_{X_t}$. In this case, the posterior $\pi_k$ can be computed as
\[
\frac{\mathrm{d} \pi_k}{\mathrm{d} \pi_0}(\theta | \mathcal{D}_k)= \frac{ \prod_{t \in [T_k]}  \mc N_{\mathrm{dim}(X_t)}(X_t | \theta_{X_t}, \Omega_{X_t})}{\int_{\R^n} \prod_{t \in [T_k]}  \mc N_{\mathrm{dim}(X_t)}(X_t | \vartheta_{X_t}, \Omega_{X_t}) \pi_0(\mathrm{d} \vartheta) },
\]
where $\mc N_{\mathrm{dim}(X_t)}(\cdot | \theta_{X_t}, \Omega_{X_t})$ is the probability density function of the Gaussian distribution $\mc N_{\mathrm{dim}(X_t)}(\theta_{X_t}, \Omega_{X_t})$. Because $\pi_0$ is either flat or Gaussian by Assumption~\ref{a:bayes}, Bayes' theorem asserts that the posterior $\pi_k$ also admits a density (with respect to the Lebesgue measure on $\R^n$). This density (we first take the case that $\pi_0$ is Gaussian for illustration) evaluated at any point $\theta \in \R^n$ is proportional to 
\begin{align*}
    &\mc N_{n}(\theta | \m_0, \cov_0) \prod_{t \in [T_k]}  \mc N_{\mathrm{dim}(X_t)}(X_t | \theta_{X_t}, \Omega_{X_t}) \\
    \propto & \exp\left(-\frac{1}{2}(\theta-\mu_0)^\top\Sigma_0^{-1}(\theta-\mu_0)\right) \prod_{t \in [T_k] }\exp\left(-\frac{1}{2}(\theta_{X_t}-X_t)^\top\Omega_{X_t}^{-1}(\theta_{X_t}-X_t)\right).
\end{align*}
By expanding the exponential terms and completing the squares, the posterior distribution $\theta |\mc D_k$ is then governed by $\pi_k \sim \mc N_n(\m_k, \cov_k)$ with 
\begin{equation}\label{eq:infopostcov}
    \cov_k = \left(\cov_0^{-1} +\sum_{t \in [T_k]} (\mc P_{M_t}^\perp)^{-1}(\Omega_{X_t}^{-1})\right)^{-1}, \quad \m_k = \cov_k \left(\cov_0^{-1}\m_0+\sum_{t\in [T_k]} (\mc P_{M_t}^\perp)^{-1}(\Omega_{X_t}^{-1}) \ (\mc P_{M_t}^\perp)^{-1}(X_t) \right),
\end{equation}
where $(\mc P_{M_t}^\perp)^{-1}$ is the inverse map of $\mc P_{M_t}^\perp$ obtained by filling entries with zeros\footnote{This inverse map can be defined as $(\mc P_{M_t}^\perp)^{-1}(\Omega_{X_t}^{-1}) = \arg\min\{\|\tilde \Omega\|_0: \tilde \Omega \in \R^{n \times n},~\mc P_{M_t}^\perp(\tilde \Omega) = \Omega_{X_t}^{-1}\}$, where $\|\cdot\|_0$ is a sparsity inducing norm. Similarly for $(\mc P_{M_t}^\perp)^{-1}(X_t)$.}. An elementary linear algebra argument asserts that $\Omega_{X_t}^{-1}$ is positive definite, and hence $(\mc P_{M_t}^\perp)^{-1}(\Omega_{X_t}^{-1})$ is positive semidefinite. The matrix inversion that defines $\cov_k$ is thus well-defined thanks to the positive definiteness of $\cov_0$ in Assumption~\ref{a:bayes}. Alternatively, in Assumption~\ref{a:bayes} if we impose a flat uninformative prior on $\theta$, then the posterior distribution of $\theta|\mc D_k$ is still Gaussian with $\pi_k \sim \mc N_n(\m_k, \cov_k)$ where 
\begin{equation}\label{eq:flatpostcov}
    \cov_k = \left(\sum_{t \in [T_k]} (\mc P_{M_t}^\perp)^{-1}(\Omega_{X_t}^{-1})\right)^{-1}, \quad \m_k = \cov_k \left(\sum_{t\in [T_k]} (\mc P_{M_t}^\perp)^{-1}(\Omega_{X_t}^{-1}) \ (\mc P_{M_t}^\perp)^{-1}(X_t) \right).
\end{equation}
The following lemma ensures us that the covariance $\Sigma_k$ is positive definite.
\begin{lemma}\label{lemma:pdcov}
    Under Assumption~\ref{a:1cell}, the posterior covariance \eqref{eq:flatpostcov} is positive definite.
\end{lemma}
% \begin{proof}[Proof of Lemma~\ref{lemma:pdcov}]
%     Let $\xi$ by any non-zero vector of dimension $n$. From Assumption~\ref{a:1cell}, there exists $t\in[T_k]$, such that $\xi_{X_t}=\mc P_{M_t}^\perp(\xi)$ is non-zero. Since $\Omega_{X_t}$ is positive definite, we have $\xi_{X_t}^\top\Omega_{X_t}^{-1}\xi_{X_t}>0$. Therefore $$\xi^\top\left((\mc P_{M_t}^\perp)^{-1}(\Omega_{X_t}^{-1})\right)\xi>0.$$ 
%     Thus 
%     $$\xi^\top\left(\sum_{t \in [T_k]} (\mc P_{M_t}^\perp)^{-1}(\Omega_{X_t}^{-1})\right)\xi>0.$$ Since $\xi$ is arbitrary, we have $\sum_{t \in [T_k]} (\mc P_{M_t}^\perp)^{-1}(\Omega_{X_t}^{-1})$ is positive definite. Thus $\Sigma_k$ is positive definite. 
% \end{proof}

The product of the posterior generator module $(G)$ is $K$ elementary posteriors $\{\pi_k\}_{k \in [K]}$, where $\pi_k$ is the posterior of $\theta | \mc D_k$. For any $k < k'$, the definition of the set $\mc D_k$ and $\mc D_{k'}$ implies that the posterior $\pi_{k'}$ carries at least as much of look-ahead-bias as the posterior $\pi_{k}$. The collection of posteriors $\{\pi_k\}$ is transmitted to the consensus mechanism $(C)$ to form an aggregated posterior $\pi\opt$ that strikes a balance between the variance and the look-ahead-bias. The aggregated posterior $\pi\opt$ is injected into the sampler module (S), which we now study.

%%%%%%%%%%%%%%%%%%%%%%%
\subsection{Multiple Imputation Sampler}
\label{sec:sampler}

A natural strategy to impute the missing values depends on recovering the (joint) posterior distribution of $(Y_t)_{t \in [T]}$ conditional on the observed data $\mc D$. Fortunately, if the noise $\epsilon_t$ are mutually independent, and the priors on $Y_t$ are mutually independent given $\theta$, it suffices to consider the posterior distribution of $Y_t$ for any given $t$. Using the law of conditional probability and leveraging the availability of the aggregated posterior $\pi\opt$, we have for any measurable set $\mc Y \subseteq \R^n$
\begin{subequations}
\begin{align}
    \PP( Y_t \in \mc Y | \mc D) &= \int_{\R^n} \PP(Y_t \in \mc Y | \theta, \mc D) \pi\opt(\mathrm{d} \theta) \notag \\
    &= \int_{\R^n} \PP(Y_t \in \mc Y | \theta, X_t, M_t) \pi\opt(\mathrm{d} \theta), \label{eq:Y_posterior0}
\end{align}
where the last equality holds thanks to the independence between $Y_t$ and any $(X_{t'}, M_{t'})$ for $t' \neq t$ conditional on $(\theta,X_t,M_t)$. Evaluating the posterior distribution for $Y_t | \mc D$ thus entails a multi-dimensional integration, which is, in general, computationally intractable.

Under Assumption~\ref{a:bayes}, $Z_t$ conditional on $\theta$ is Gaussian $\mc N_n(\theta, \Omega)$. Because $Z_t$ is decomposed into $X_t$ and $Y_t$, the distribution of $Y_t$ conditional on $\theta$ and $(X_t, M_t)$ is thus Gaussian~\cite[Corollary~5]{ref:cambanis1981theory}. More specifically, 
\be \label{eq:Y_posterior1}
    Y_t | \theta, X_t, M_t \sim \mc N_{\mathrm{dim}(Y_t)} \left(\theta_{Y_t} + \Omega_{Y_t X_t} \Omega_{X_t}^{-1} (X_t - \theta_{X_t}), \Omega_{Y_t} -  \Omega_{Y_t X_t} \Omega_{X_t}^{-1} \Omega_{X_t Y_t} \right),
\ee
\end{subequations}
where $\theta_{X_t} = \mc P_{M_t}(\theta)$, $\theta_{Y_t} = \mc P_{M_t}^\perp(\theta)$, while $\Omega_{X_t}$ and $\Omega_{Y_t}$ are obtained by applying the projection operators $\mc P_{M_t}$ and $\mc P_{M_t}^\perp$ on $\Omega$. The matrix $\Omega_{Y_t X_t} \in \R^{\mathrm{dim}(Y_t) \times \mathrm{dim}(X_t)}$ is the covariance matrix between $Y_t$ and $X_t$ induced by $\Omega$. Notice that by the definitions of the projection operators, $\theta_{X_t}$ and $\theta_{Y_t}$ are affine transformations of $\theta$. Consequently, the mean of $Y|\theta, X_t, M_t$ is an affine function of $\theta$, while its covariance matrix is independent of $\theta$. The posterior formula~\eqref{eq:Y_posterior0} coupled with the structural form~\eqref{eq:Y_posterior1} indicate that the distribution of $Y_t | \mc D$ coincide with the distribution of the random vector $\xi_t \in \R^{\mathrm{dim}(Y_t)}$ dictated by
\be \label{eq:Y_generative}
    \xi_t = A_t \theta + b_t + \eta_t, \qquad \theta \sim \pi\opt, \qquad \eta_t \sim \mc N_{\mathrm{dim}(Y_t)}(0, S_t), \qquad \theta \indep \eta_t,
\ee
for some appropriate parameters $A_t$, $b_t$, and $S_t$ that are designed to match the mean vector and the covariance matrix of the Gaussian distribution in~\eqref{eq:Y_posterior1}. More specifically, we have $A_t\theta= \theta_{Y_t} - \Omega_{Y_tX_t}\Omega_{X_t}^{-1}\theta_{X_t}, b_t = \Omega_{Y_tX_t}\Omega_{X_t}^{-1}X_t,S_t = \Omega_{Y_t} -  \Omega_{Y_t X_t} \Omega_{X_t}^{-1} \Omega_{X_t Y_t}$. Algorithm~\ref{algo:full_Bayes} details the procedure for missing value imputation using the model~\eqref{eq:Y_generative}.

The variability in the sampling of $\xi_t$ using~\eqref{eq:Y_generative} comes from two sources: that of sampling $\theta$ from $\pi\opt$ and that of sampling $\eta_t$ from $\mc N_{\mathrm{dim}(Y_t)}(0,S_t)$. As we have seen from Section~\ref{sec:generator} (e.g., equations~\eqref{eq:infopostcov} and \eqref{eq:flatpostcov}), the posterior covariance of $\theta$ is inverse-proportional to the time-dimension of the data set, and therefore the variability due to $\pi\opt$ is likely to be overwhelmed by the variability due to $\eta_t$ unless $T$ is small. To fully exhibit the potential of the aggregated posterior $\pi\opt$ for bias and variance trade-off, especially when validating our proposed method in numerical experiments, we choose to eliminate the idiosyncratic noise $\eta_t$, namely, the missing values are imputed by its conditional expectation given $\theta$, which is depicted in Algorithm~\ref{algo:conditional_Bayes}.

\begin{table}[th]
	\begin{minipage}{0.83\columnwidth}
		\begin{algorithm}[H]
			\caption{Fully Bayesian Imputation}
			\label{algo:full_Bayes}
			\begin{algorithmic}
				\REQUIRE aggregated posterior $\pi\opt$, data set $\mc D = \{(X_t, M_t): t \in [T]\}$, covariance matrix $\Omega$
				\\[0.5ex]
				\STATE Sample $\theta$ from $\pi\opt$
				\FOR {$t = 1, \ldots, T$}
				\STATE Compute $A_t$, $b_t$, $S_t$ for model~\eqref{eq:Y_generative} dependent on $(X_t, M_t, \Omega)$
				\STATE Sample $\eta_t$ from $\mc N_{\mathrm{dim}(Y_t)}(0, S_t)$
				\STATE Impute $Y_t$ by $A_t \theta + b_t + \eta_t$
				\ENDFOR
				\ENSURE An imputed data set
			\end{algorithmic}
		\end{algorithm}
	\end{minipage}
\end{table}

\begin{table}[th]
	\begin{minipage}{0.83\columnwidth}
		\begin{algorithm}[H]
			\caption{Conditional Expectation Bayesian Imputation}
			\label{algo:conditional_Bayes}
			\begin{algorithmic}
				\REQUIRE aggregated posterior $\pi\opt$, data set $\mc D = \{(X_t, M_t): t \in [T]\}$, covariance matrix $\Omega$
				\\[0.5ex]
				\STATE Sample $\theta$ from $\pi\opt$
				\FOR {$t = 1, \ldots, T$}
				\STATE Compute $A_t$, $b_t$ for model~\eqref{eq:Y_generative} dependent on $(X_t, M_t, \Omega)$
				\STATE Impute $Y_t$ by $A_t \theta + b_t$
				\ENDFOR
				\ENSURE An imputed data set
			\end{algorithmic}
		\end{algorithm}
	\end{minipage}
\end{table}

%%%%%%%%%%%%%%%%%%%%%%%
\section{Bias-Variance Targeted Consensus Mechanism}
\label{sec:consensus}

We devote this section to providing the general technical details on the construction of the consensus mechanism module (C) in our BVMI framework. Conceptually, the consensus mechanism module receives $K$ posteriors $\{\pi_k\}_{k \in [K]}$ from the posterior generator module (G), then synthesizes a unique posterior $\pi\opt$ and transmits $\pi\opt$ to the sampler module (S). We first provide a formal definition of a consensus mechanism.

\begin{definition}[Consensus mechanism] \label{def:consensus}
    Fix an integer $K \ge 2$. A consensus mechanism $\mc C_\lambda: \mc M^K \to \mc M$ parametrized by $\lambda$ outputs a unifying posterior $\pi\opt$ for any collection of $K$ elementary posteriors $\{\pi_k\}$.
\end{definition}

We assume that the parameter $\lambda$ is finite-dimensional, and the set of all feasible parameters is denoted by $\Delta_K$, which is a nonempty finite-dimensional set. The space of all possible consensus mechanisms induced by $\Delta_K$ is represented by $\mbb C \Let \{ \mc C_\lambda: \lambda \in \Delta_K\}$. For a given set of posteriors $\{\pi_k\}$, we need to identify the optimal consensus mechanism $\mc C_{\lambda\opt} \in \mbb C$ in the sense that the aggregated posterior $\pi\opt = \mc C_{\lambda\opt}(\{\pi_k\})$ achieves certain optimality criteria. In this paper, the optimality criteria of $\pi\opt$ are determined based on the notion of bias-variance trade-off, which is ubiquitously used in (Bayesian) statistics. The first step towards this goal entails quantifying the bias of an aggregated posterior~$\pi\opt$. 
% Ideally, this bias quantity is measured by 
% \begin{equation}\label{eq:deflookahead}
%   \| \EE_{\pi\opt}[\theta] - \theta^{true}\|_2,
% \end{equation}
% where $\theta^{true}$ is the true mean of the return $Z$. Unfortunately, the value $\theta^{true}$ is not available in practice, and hence quantifying the true bias of an aggregated posterior is impossible. 
% As an alternative, we resort to a proxy which is obtained through the following more general definition of the bias notion.

% The bias of this aggregated predicative posterior should be ideally measured by
% \[
%     \| \EE_{\pi\opt}[\xi] - \theta\|_2,
% \]
% \viet{where $\theta$ is the true mean vector of the random vector $\xi$?}
% unfortunately, $\theta$ %$\theta_0$
% is unknown, so the true bias is not available to us. We thus need to introduce the following bias criteria.

% \viet{Need some introduction for the anchored mean} We use an anchored mean that is the mean of some anchored posterior which we assume is bias-free. In Section~\ref{sec:bayesianmodel} and~\ref{sec:numerics} we will use the first individual posterior $\pi_1$ since it is free from look-ahead-bias.

% \begin{definition}[$(\ell, \delta)$-bias consensus mechanism] Fix an input  $\{\pi_k\}$ and an anchored mean $\msa$.
%     Given a function $\ell: \R^n \to \R$ and a scalar $\delta \in \R_+$, a mechanism $\mc C_\lambda \in \mbb C$ satisfies the $(\ell,\delta)$-bias requirement if $\ell(\EE_{\mc C_\lambda(\{ \pi_k\})}[\xi] - \msa) \le \delta$.
% \end{definition}

\begin{definition}[$(\ell, \m, \delta)$-bias tolerable posterior] \label{def:bias} Fix a penalization function $\ell: \R^n \to \R$, an anchored mean $ \m \in \R^n$ and a tolerance $\delta \in \R_+$. A posterior $\pi\opt$ is said to be $(\ell,  \m, \delta)$-bias tolerable  if $\ell(\EE_{\pi\opt}[\theta] - \m) \le \delta$.
\end{definition}

If $\ell$ is the Euclidean norm and the anchored mean is set to $\theta^{true}$, then we recover the standard measure of bias in the literature with the (subjective) posterior distribution in place of the (objective) frequentist population distribution, and a posterior $\pi\opt$ is unbiased if it is $(\|\cdot\|_2, \theta^{true}, 0)$-bias tolerable. 

Definition~\ref{def:bias} provides a broad spectrum of options to measure the bias of the aggregated posterior. Most importantly, we can flexibly choose a suitable value for the anchored mean $\mu$ to mitigate the look-ahead-bias.
The tolerance parameter $\delta$ indicates how much look-ahead-bias is accepted in the aggregated posterior $\pi\opt$. If $\delta = 0$, then we aim for an aggregated posterior distribution with the exact mean vector $\mu$. The generality of Definition~\ref{def:bias} also allows us to choose the penalization function. Throughout this paper, we will focus on the penalization function of the form
\be \label{eq:ell-Z}
    \ell_{\mc Z}(\mu) = \Sup{z \in \mc Z}~ z^\top \m
\ee
parametrized by some nonempty, compact set $\mc Z \subseteq \R^n$. The functional form~\eqref{eq:ell-Z} enables us to model diverse bias penalization effects as demonstrated in the following example.

%\newpage

\begin{example}[Forms of $\ell_{\mathcal Z}$] \label{ex:penalty}

    \begin{enumerate}[label=(\roman*),leftmargin=*]
         \item If $\mc Z = \{ z \in \R^n: \| z\|_2 \le 1\}$ is the Euclidean ball, then $\ell_{\mc Z}(\mu) = \| \mu \|_2$. This 2-norm penalization coincides with the canonical bias measurement in the literature. %This penalization function is a proxy for the bias because 
        % \[
        %     \| \mu - \theta_0\|_2 = \Sup{\| z \|_2 = 1} z^\top (\mu - \theta_0) \le \Sup{\| z \|_2 = 1} z^\top \mu + \Sup{\| z \|_2 = 1} - z^\top \theta_0 = \| \mu \|_2 + \|\theta_0\|_2.
        % \]
        % Thus a requirement that $\ell(\mu) \le \delta$ implies that the true bias is upper bounded by $\delta + \|\theta_0\|_2$. 
        \item If $\mc Z$ is a convex polyhedron $\mc Z = \{ z \in \R^n: Az \le b\}$ for some matrix $A$ and vector $b$ of appropriate dimensions, then $\ell_{\mc Z}$ coincides with the support function of $\mc Z$.
        \item If $\mc Z = \{ \mathbbm 1_n\}$, then $\ell_{\mc Z}$ is an upward bias penalization under the $1/n$-uniform portfolio. The $1/n$-uniform portfolio is known to be robust~\citep{ref:DeMigueloptimal}.
    \end{enumerate}
\end{example}

Next, we discuss the variance properties of the aggregated posterior. Let $\mathrm{Cov}_{\mc C_{\lambda}(\{ \pi_k\})}[\theta] \in \R^{n\times n}$ be the covariance matrix of $\theta$ under the posterior distribution $\mc C_{\lambda}(\{ \pi_k\})$. 
% Because there is no natural complete order on $\R^n$, it is impossible, in general, to compare the posterior variance between $\Var_{\mc C_{\lambda}(\{ \pi_k\})}[\theta]$ and $\Var_{\mc C_{\lambda'}(\{ \pi_k\})}[\theta]$ when $\lambda \neq \lambda'$. 
We consider the following minimal variance criterion computed based on taking the sum of the diagonal entries.
\begin{definition}[Minimal variance posterior]
    For a fixed input $\{\pi_k\}$, the aggregated posterior $\pi\opt$ obtained by $\pi\opt = \mc C_{\lambda\opt}(\{\pi_k\})$ for an $\mc C_{\lambda\opt} \in \mbb C$ has minimal variance if
    \[
        \mathrm{Tr}(\mathrm{Cov}_{\pi\opt}[\theta]) \le \mathrm{Tr}(\mathrm{Cov}_{\mc C_{\lambda}(\{ \pi_k\})}[\theta]) \quad \forall \mc C_{\lambda} \in \mbb C.
    \]
    %where $\mathbbm{1}$ is the vector of ones.
\end{definition}

We are now ready to define our notion of optimal consensus mechanism.

\begin{definition}[Optimal consensus mechanism] \label{def:optimal-consensus}
Fix an input $\{\pi_k\}$. The consensus mechanism $\mc C_{\lambda\opt} \in \mbb C$ is optimal for $\{\pi_k\}$ if the aggregated posterior $\pi\opt = \mc C_{\lambda\opt}(\{\pi_k\})$ is $(\ell, \mu, \delta)$-bias tolerable and has minimal variance.
\end{definition}

Intuitively, the optimal consensus mechanism for the set of elementary posteriors $\{\pi_k\}$ is chosen so as to generate an aggregated posterior $\pi\opt$ that has low variance and an acceptable level of bias. Moreover, because the ultimate goal of the consensus mechanism is to mitigate the look-ahead-bias, it is imperative to choose the anchored mean $\mu$ in an appropriate manner. As the posterior $\pi_1$ obtained by conditioning on the training set $\mc D_1 = \{(X_t, M_t): t \in T^{train}\}$ carries the least amount of look-ahead-bias, it is reasonable to consider the discrepancy between the mean of the aggregated posterior $\pi\opt$ and that of $\pi_1$ as a proxy for the true bias. The optimal consensus mechanism can be constructed from $\lambda\opt \in \Delta_K$ that solves the following consensus optimization problem
\be \label{eq:prob}
    \begin{array}{cl}
        \min & \mathrm{Tr}(\mathrm{Cov}_{\mc C_{\lambda}(\{ \pi_k\})}[\theta]) \\
        \st & \lambda \in \Delta_K \\
        &  \ell_{\mc Z}(\EE_{\mc C_\lambda(\{ \pi_k\})}[\theta]- \m_1) \le \delta
    \end{array}
\ee
for some set $\mc Z$. Notice that problem~\eqref{eq:prob} shares a significant resemblance with the Markowitz problem: the objective function of~\eqref{eq:prob} minimizes a variance proxy, while the constraint of~\eqref{eq:prob} involves the expectation parameters.

So far, we have not yet specified how the mechanism $\mc C_{\lambda}$ depends on the parameters $\lambda$. In its most general form, $\mc C_{\lambda}$ may depend nonlinearly on $\lambda$, and in this case, problem~\eqref{eq:prob} can be intractable. It behooves us to simplify this dependence by imposing certain assumptions on $\Delta_K$ and on the relationship between $\lambda$ and $\mc C_\lambda$. Towards this end, we assume for the rest of this paper that $\Delta_K = \{ \lambda \in \R_+^K: \sum_{k = 1}^K \lambda_k = 1 \}$, a $(K-1)$-dimensional simplex. Furthermore, $\mc C_{\lambda}$ depends explicitly on $\lambda$ through the following variational form
\be \label{eq:C-def}
    \mc C_\lambda(\{ \pi_k\}) = \arg \Min{\pi\in\mc M} \sum_{k \in [K]} \lambda_k \varphi(\pi, \pi_k),
\ee
where $\varphi$ represents a measure of dissimilarity in the space of probability measures. If $\lambda$ is a vector of $1/K$, then $\mc C_{1/K \mathbbm{1}}(\{\pi_k\})$ is the (potentially asymmetric) barycenter of $\{\pi_k\}$ under $\varphi$.
The parameter vector $\lambda$ can now be thought of as the vector of weights corresponding to the set of posteriors $\{\pi_k\}$, and as a consequence, the mechanism $\mc C_{\lambda}$ outputs the \textit{weighted} barycenter of $\{\pi_k\}$. 

It is important to emphasize the sharp distinction between our consensus optimization problem~\eqref{eq:prob} and the existing literature on barycenters. While the existing literature typically focuses on studying the barycenter for a fixed parameter $\lambda$, problem~\eqref{eq:prob} searches for the best value of $\lambda$ that meets our optimality criteria. Put differently, one can also view problem~\eqref{eq:prob} as a search problem for the \textit{best reweighting scheme} of $\{\pi_k\}$ so that their barycenter has minimal variance while being bias tolerable. 

The properties of the aggregation mechanism depend fundamentally on the choice of $\varphi$, and in what follows we explore three choices of $\varphi$ that are analytically tractable. In Section~\ref{sec:forward}, we explore the mechanism using the forward Kullback-Leibler divergence, and in Section~\ref{sec:Wass}, we examine the mechanism under the type-2 Wasserstein distance. We will show that in these two cases, the optimal consensus problem~\eqref{eq:prob} is a convex optimization problem under specific assumptions. In Section~\ref{sec:backward}, we study the mechanism under the backward Kullback-Leibler divergence, and reveal that the resulting consensus problem is non-convex.

%%%%%%%%%%%%%%%%%%%%%%%%%%%%%
\section{Optimal Forward Kullback-Leibler Consensus Mechanism}
\label{sec:forward}

In this section, we will measure the dissimilarity between distributions using the Kullback-Leibler (KL) divergence. The KL divergence is formally defined as follows.

\begin{definition}[Kullback-Leibler divergence] \label{def:KL}
	Given two probability measures $\nu_1$ and $\nu_2$ with $\nu_1$ being absolutely continuous with respect to $\nu_2$, the Kullback-Leibler (KL) divergence from $\nu_1$ to $\nu_2$ is defined as $\KL(\nu_1\!\parallel \!\nu_2)\!\Let\!\EE_{\nu_1} \left[ \log (\mathrm{d} \nu_1/\mathrm{d} \nu_2) \right]$,
	where $\mathrm{d} \nu_1/\mathrm{d} \nu_2$ is the Radon-Nikodym derivative of $\nu_1$ with respect to $\nu_2$.
\end{definition}

It is well-known that the KL divergence is non-negative, and it vanishes to zero if and only if $\nu_1$ coincides with $\nu_2$. The KL divergence, on the other hand, is not a distance: it even fails to be symmetric, and in general, we have $\KL(\nu_1 \parallel \nu_2) \neq \KL(\nu_2 \parallel \nu_1)$. The KL divergence is also a fundamental tool in statistics, information theory~\citep{ref:cover2012elements} and information geometry~\citep{ref:amari2016information}. 

Under Assumption~\ref{a:bayes}, the posteriors $\pi_k$ are Gaussian, hence, to compute the aggregated posterior $\mc C_{\lambda}({\pi_k})$, we often need to evaluate the KL divergence from, or to, a Gaussian distribution. Fortunately, when measuring the dissimilarity between two Gaussian distributions, the KL divergence admits a closed-form expression.  
 
\begin{lemma}[KL divergence between Gaussian distributions~\citep{ref:pardo2018statistical}] \label{lemma:KL}
    For any pairs $(\m_i, \cov_i)$, $(\m_j, \cov_j) \in \R^n \times \PD^n$, we have
    \begin{align*}
        \KL \big( \mc N_n(\m_i, \cov_i) \parallel \mc N_n(\m_j, \cov_j) \big) \Let (\m_j - \m_i)^\top\cov_j^{-1} (\m_j - \m_i) +\Tr{\cov_i \cov_j^{-1}} - \log\det (\cov_i \cov_j^{-1}) - n.
    \end{align*}
\end{lemma}

In this section we use $\varphi$ as the forward KL-divergence to aggregate the posteriors; more specifically, for any two probability distributions $\nu_1$ and $\nu_2$, we use $\varphi(\nu_1,\nu_2)=\KL(\nu_1\parallel\nu_2)$. Note that this choice of $\varphi$ places the posteriors $\pi_k$ as the second argument in the divergence. We use $\mc C_\lambda^F$ to denote the consensus mechanism induced by the forward KL-divergence, and $\mc C_\lambda^F$ is formally defined as follows. 
\begin{definition}[Forward KL consensus mechanisms]
    A mechanism $\mc C_\lambda^F$ with weight $\lambda$ is called a forward KL consensus mechanism if for any collection of posterior distributions $\{\pi_k\}$, we have 
\be \label{eq:consensus-F}
    \mc C_\lambda^F(\{\pi_k\}) = \arg \Min{\pi\in\mc M} \sum_{k \in [K]} \lambda_k~\KL (\pi \parallel \pi_k).
\ee
\end{definition} 

The KL barycenter of the form~\eqref{eq:consensus-F} is widely used to fuse distributions and aggregate information, and it has witnessed many successful applications in machine learning~\citep{ref:claici2020model,ref:da2020kullback,ref:li2019cardinality}.
The set of forward KL mechanisms parametrized by the simplex $\Delta_K$ is denoted by $\mbb C^F \Let \{ \mc C_\lambda^F: \lambda \in \Delta_K\}$. For a given input $\{\pi_k\}$ and weight $\lambda \in \Delta_K$, computing $\mc C_{\lambda}^F(\{\pi_k\})$ requires solving an infinite dimensional optimization problem. Fortunately, when the posteriors $\pi_k$ are Gaussian under Assumption~\ref{a:bayes}, we can characterize $\mc C_\lambda^F$ analytically thanks to the following proposition. 

\begin{proposition}[Closed form of the consensus~{\citep{ref:battistelli2013consensus}}] \label{prop:consensus-F}
    Suppose that $\pi_k \sim \N(\m_k, \cov_k)$ with $\cov_k \succ 0$ for any $k \in [K]$. For any $\lambda \in \Delta_K$, we have $\mc C_\lambda^F(\{\pi_k\}) = \mc N(\msa, \covsa)$ with
    \be
        \msa = (\sum_{k \in [K]}\lambda_k\Sigma_k^{-1})^{-1} \big( \sum_{k \in [K]}\lambda_k \cov_k^{-1} \m_k \big) \in \R^n, \quad \covsa = (\sum_{k \in [K]}\lambda_k\Sigma_k^{-1})^{-1} \in \PD^n.
    \ee
\end{proposition}

Next, we discuss the tractability of the bias constraint under the forward KL consensus mechanism. For a given collection of posterior distributions $\{\pi_k\}$, define the following set
\[
    \Lambda_{\ell_{\mc Z}, \delta}^F( \{\pi_k\}) = \left\{
        \lambda \in \Delta_K : ~  \ell_{\mc Z}(\EE_{\mc C_\lambda^F(\{ \pi_k\})}[\theta]- \m_1) \le \delta
    \right\}
\]
that contains all the parameters $\lambda\in\Delta_K$ for which the forward KL mechanism $\mc C_\lambda^F$ satisfies the $(\ell_{\mc Z}, \mu_1,\delta)$-bias requirement.  Unfortunately, despite the analytical expression for the aggregated posterior in Proposition~\ref{prop:consensus-F}, it remains difficult to characterize the set $\Lambda_{\ell_{\mc Z}, \delta}^F( \{\pi_k\})$ in its most general form. This difficulty can be alleviated if we impose additional assumptions on $\{\pi_k\}$ and the set $\mc Z$. 

\begin{proposition}[$(\ell_{\mc Z},\m_1, \delta)$-bias forward KL mechanisms] \label{prop:bias-F}
Suppose that there exists an orthogonal matrix $V \in \R^{n \times n}$ and a collection of diagonal matrices $D_k = \diag(d_k) \in \PD^n$ for $k \in [K]$ such that $\cov_k = V D_k V^\top \succ 0$ and $\pi_k \sim \mc N(\m_k, \cov_k)$ for all $k \in [K]$. 
If $\mc Z = \{ z \in \R^n: \| V^\top z\|_1 \le 1\}$, then 
\begin{equation}\label{eq:biasF}
        \Lambda_{\ell_{\mc Z}, \delta}^F( \{\pi_k\}) = \left\{ \lambda \in \Delta_K: 
       \ds - \delta \le \frac{\sum_{k \in [K]} c_{kj} \lambda_k}{\sum_{k \in [K]}  d_{kj}^{-1} \lambda_k} - v_j^\top \m_1 \le \delta \quad \forall j \in [n]
    \right\},
    \end{equation}
    where $ d_{kj}$ is the $j$-th component of the vector $d_k$, $V = [v_1, \ldots, v_n]$, and
    \[
    c_{kj} = \frac{v_j^\top \m_k}{ d_{kj}} \quad \forall (k,j) \in [K] \times [n].
    \]
\end{proposition}

We are now ready to state the main result of this section, which provides a reformulation of problem~\eqref{eq:prob} under conditions of Proposition~\ref{prop:bias-F}.
\begin{theorem}[Optimal forward KL mechanism] \label{thm:refor}
    Suppose that there exists an orthogonal matrix $V \in \R^{n \times n}$ and a collection of diagonal matrices $D_k = \diag( d_k) \in \PD^n$ for $k \in [K]$ such that $\cov_k = V  D_k V^\top \succ 0$ and $\pi_k \sim \mc N(\m_k, \cov_k)$ for all $k \in [K]$, and suppose that $\mc Z = \{ z \in \R^n: \| V^\top z\|_1 \le 1\}$. Let $\lambda\opt$ be the optimal solution in the variable $\lambda$ of the following second-order cone optimization problem
    \be \label{eq:refor-F}
    \begin{array}{cl}
    \min & \ds\sum_{j\in[n]} \gamma_j \\
    \st & \lambda \in \Delta_K,~\gamma \in \R_+^n \\
        & \!\! \left. \begin{array}{l}
        \ds \sum_{k \in [K]} c_{kj} \lambda_k \le (\delta + v_j^\top \m_1) \sum_{k \in [K]}  d_{kj}^{-1} \lambda_k \\
        \ds \sum_{k \in [K]} c_{kj} \lambda_k \ge (v_j^\top \m_1 - \delta) \sum_{k \in [K]}  d_{kj}^{-1} \lambda_k \\
        \left\| \begin{bmatrix} 
            2 \\ \sum_{k \in [K]} d_{kj}^{-1} \lambda_k - \gamma_j
        \end{bmatrix} \right\|_2 \le \sum_{k \in [K]} d_{kj}^{-1} \lambda_k + \gamma_j
            \end{array} \right\} \quad \forall j \in [n],
        \end{array}
    \ee
    where $d_{kj}$ is the $j$-th component of the vector $ d_k$, $V = [v_1, \ldots, v_n]$, and
    \[
    c_{kj} = \frac{v_j^\top \m_k}{ d_{kj}} \quad \forall (k,j) \in [K] \times [n].
    \]
    Then $\mc C_{\lambda\opt}^F$ is the optimal consensus mechanism in $\mbb C^F$, that is, $\mc C_{\lambda\opt}^F$ is $(\ell_{\mc Z},\mu_1,\delta)$-bias tolerable and has minimal variance.
\end{theorem}

Problem~\eqref{eq:refor-F} is a tractable convex optimization problem, and it can be readily solved using off-the-shelf conic solvers such as MOSEK~\citep{mosek}. Note that the tractability of the reformulation~\eqref{eq:refor-F} depends crucially on the assumption that the covariance matrices $\Sigma_k$ are simultaneously diagonalizable. However, this assumption is rarely satisfied in reality. We now discuss a pre-processing step that projects all the covariance matrices $\Sigma_k$ onto a common eigenbasis. To this end, for a given orthogonal matrix $V \in \R^{n \times n}$, we define the following subspace of Gaussian distributions
\[
    \mc N_V \Let \big\{ \pi \in \mc M: \exists (\m, d) \in \R^n \times \R_{++}^n, \cov = V\diag(d) V^\top, \pi \sim \mc N(\m, \cov) \big\}
\]
that contains all nondegenerate Gaussian distributions whose covariance matrix is diagonalizable with respect to $V$. Given any posterior $\pi_k \sim \N(\m_k, \cov_k)$, we add a pre-processing step to create $\pi_k' \sim \N(\m_k, \cov_k')$ with $\cov_k'$ being diagonalizable with the orthogonal matrix $V$ based on the following lemma.

\begin{lemma}[Projection on $\mc N_V$] \label{lemma:projection-V}
    Let $V$ be an orthogonal matrix. Suppose that $\pi_k \sim \N(\m_k, \cov_k)$ for some $\m_k \in \R^n$ and $\cov_k \in \PD^n$, then the following holds:
    \begin{enumerate}[label=(\roman*),leftmargin=*]
        \item \label{lemma:projection-V1} Let $d_k \in \R_{++}^n$ be such that $(d_k)_j = v_j^\top \cov_k v_j$ for $j \in [n]$. Then the distribution $\pi_k' \sim \N(\m_k, \cov_k')$ with $\cov_k' = V \diag(d_k) V^\top$ is the optimal solution of
        \[
            \Min{\pi_k' \in \mc N_V}~\KL( \pi_k \parallel \pi_k').
        \]
        \item \label{lemma:projection-V2} Let $ d_k \in \R_{++}^n$ be such that $(d_{k})_j = [V^\top (\cov_k)^{-1} V]_{jj}$ for $j \in [n]$. Then the distribution $\pi_k' \sim \N(\m_k, \cov_k')$ with $\cov_k' = V \diag( d_k) V^\top$ is the optimal solution of
        \[
            \Min{\pi_k' \in \mc N_V}~\KL( \pi_k' \parallel \pi_k).
        \] 
        
    \end{enumerate}
\end{lemma}

%%%%%%%%%%%%%%%%%%%%%%%%%%%%%%%
\section{Optimal Wasserstein Consensus Mechanism}
\label{sec:Wass}

In this section, we will measure the dissimilarity between distributions using the Wasserstein type-2 distance.

\begin{definition}[Wasserstein distance]
	\label{definition:wasserstein}	
	Given two probability measures $\nu_1$ and $\nu_2$, the type-2 Wasserstein distance between $\nu_1$ and $\nu_2$ is defined as 
	\be
	\notag
	\Wass_2(\nu_1, \nu_2) \Let \Min{\gamma}  \left( \int_{\R^n \times \R^n} \|\theta_1 - \theta_2\|_2^2\, \gamma({\rm d}\theta_1, {\rm d} \theta_2) \right)^{\half},
	\ee
	where the minimization is taken over all couplings of $\nu_1$ and $\nu_2$. 
\end{definition}

The type-2 Wasserstein distance $\Wass_2$ is non-negative and it vanishes only if $\nu_1 = \nu_2$. Moreover, $\Wass_2$ is symmetric and satisfies the triangle inequality. As a consequence, $\Wass_2$ represents a metric on the space of probability measures~\cite[p.~94]{ref:villani2008optimal}. The type-2 Wasserstein distance $\Wass_2$ is a fundamental tool in the theory of optimal transport, and its squared value can be interpreted as the minimum cost of transporting the distribution $\nu_1$ to $\nu_2$, in which the cost of moving a unit probability mass is given by the squared Euclidean distance. The application of the Wasserstein distance has also been extended to economics~\citep{ref:galichon2016optimal}, machine learning~\citep{ref:kuhn2019wasserstein} and statistics~\citep{ref:panaretos2019statistical}.

As the posteriors $\pi_k$ are Gaussian under Assumption~\ref{a:bayes}, we often need to evaluate the Wasserstein distance from a Gaussian distribution. Conveniently, the Wasserstein distance between two Gaussian distributions can be written in an analytical form~\citep{ref:givens1984class}.   

\begin{definition}[Wasserstein type-2 distance between Gaussian distributions] \label{def:Wass}
	The Wasserstein $\Wass_2$ distance between $\mc N(\m, \cov)$ and $\mc N(\m', \cov')$ is given by
\begin{equation*}
\Wass_2\big( \mc N(\m, \cov) , \mc N(\m', \cov') \big) = \sqrt{ \| \m - \m' \|_2^2 + \Tr{\cov + \cov' - 2 \big( (\cov')^{\half} \cov (\cov')^{\half} \big)^{\half} } }.
\end{equation*}
\end{definition}

In this section, we use $\varphi$ as the Wasserstein distance $\Wass_2$ to aggregate the posteriors. In more details, for any two probability distributions $\nu_1$ and $\nu_2$, we use $\varphi(\nu_1,\nu_2)=\Wass_2(\nu_1, \nu_2)$. We use $\mc C_\lambda^W$ to denote the consensus mechanism induced by the Wasserstein metric, and $\mc C_\lambda^W$ is formally defined as follows. 

\begin{definition}[Wasserstein consensus mechanisms] \label{def:consensus-W}
    A mechanism $\mc C_\lambda^W$ with weight $\lambda$ is called a Wasserstein consensus mechanism if for any collection of posterior distributions $\{\pi_k\} \in \mc N^K$, we have
\be \label{eq:consensus-W}
    \mc C_\lambda^W(\{\pi_k\})  = \arg \min_{\pi \in \N}  \sum_{k \in [K]}\lambda_k~\Wass_2 (\pi , \pi_k)^2.
\ee
\end{definition} 

Notice that by Definition~\ref{def:consensus-W}, we will optimize only over the space of Gaussian distributions in problem~\eqref{eq:consensus-W}. The Wasserstein barycenter of the form~\eqref{eq:consensus-W} is widely used to fuse distributions and aggregate information. It has many successful applications in machine learning, for example, distribution clustering~\citep{ref:ye2017fast}, shape and texture interpolation~\citep{ref:rabin2012wasserstein,ref:solomon2015con}, and multi-target tracking~\citep{ref:baum2015on}. The set of Wasserstein mechanisms parametrized by the simplex $\Delta_K$ is denoted by $\mbb C^W \Let \{ \mc C_\lambda^W: \lambda \in \Delta_K\}$. When the posteriors $\pi_k$ are Gaussian under Assumption~\ref{a:bayes}, we can characterize $\mc C_\lambda^W$ semi-analytically thanks to the following proposition.

\begin{proposition}[{Closed form of the consensus~\citep{ref:agueh2011barycenter}}] \label{prop:consensus-W}
Suppose that $\pi_k \sim \N(\m_k, \cov_k)$ with $\cov_k \succ 0$ for any $k \in [K]$. For any $\lambda \in \Delta_K$, we have $\mc C_\lambda^W(\{\pi_k\}) = \mc N(\msa, \covsa)$ with $\msa = \sum_{k \in [K]} \lambda_k \mu_k \in \R^n$ and $\covsa \in \PD^n$ is the unique solution of the matrix equation
    \be
        \covsa = \sum_{k \in [K]} \lambda_k (\covsa^\half \cov_k \covsa^\half)^\half.
    \ee
\end{proposition}

Proposition~\ref{prop:consensus-W} asserts that the aggregated mean $\msa$ is simply the weighted average of the posterior means $\mu_k$. The aggregated covariance matrix $\covsa$ is generally not available in closed form. Nevertheless, there exists a fast numerical routine to compute~$\covsa$ for any set of input $\{\cov_k\}$ using fixed point iterations~\citep{ref:alvarez2016fixed}. However, for the special case of $K=2$~\cite{ref:McCann1997convexity} show that $\covsa$ admits a closed-form expression as
\[
\covsa = \left(\lambda_1 I_n+\lambda_2\Phi\right)\Sigma_1\left(\lambda_1I_n+\lambda_2\Phi\right) \in \PD^n,
\]
where $\Phi =\Sigma_2^{1/2}\left(\Sigma_2^{1/2}\Sigma_1\Sigma_2^{1/2}\right)^{-1/2}\Sigma_2^{1/2}$.

In the next step, we leverage the expression of $\msa$ to reformulate the bias tolerance constraint. For the collection of posterior distributions $\{\pi_k\}$, define the following set
\[
    \Lambda_{\ell_{\mc Z}, \delta}^W( \{\pi_k\}) = \left\{
        \lambda \in \Delta_K : ~  \ell_{\mc Z}(\EE_{\mc C_\lambda^W(\{ \pi_k\})}[\theta]- \m_1) \le \delta
    \right\}
\]
that contains the subset of $\lambda\in\Delta_K$ for which the Wasserstein mechanism $\mc C_\lambda^W$ satisfies the $(\ell_{\mc Z},\m_1, \delta)$-bias requirement. The next proposition provides the expression of~$\Lambda_{\ell_{\mc Z}, \delta}^W( \{\pi_k\})$ corresponding to the examples of $\mc Z$ in Example~\ref{ex:penalty}.
\begin{proposition}[$(\ell_{\mc Z}, \m_1,\delta)$-bias Wasserstein mechanisms] \label{prop:bias-W} 
    The following assertions hold.
     \begin{enumerate}[label=(\roman*),leftmargin=*]
         \item If $\mc Z = \{ z \in \R^n: \| z\|_2 \le 1\}$ is the Euclidean ball, then \begin{equation}\label{eq:euclideanball}
        \Lambda_{\ell_{\mc Z}, \delta}^W( \{\pi_k\}) = \left\{ \lambda \in \Delta_K: 
      \left\|\sum_{k \in [K]}\lambda_k\mu_k - \m_1\right\|_2\leq \delta
    \right\}.
    \end{equation}%This penalization function is a proxy for the bias because 
        % \[
        %     \| \mu - \theta_0\|_2 = \Sup{\| z \|_2 = 1} z^\top (\mu - \theta_0) \le \Sup{\| z \|_2 = 1} z^\top \mu + \Sup{\| z \|_2 = 1} - z^\top \theta_0 = \| \mu \|_2 + \|\theta_0\|_2.
        % \]
        % Thus a requirement that $\ell(\mu) \le \delta$ implies that the true bias is upper bounded by $\delta + \|\theta_0\|_2$. 
        % \item If $\mc Z = \{ z \in \R^n: Az \le b\}$ for some matrix $A$ and vector $b$ of appropriate dimensions. \note{then?}
        \item If $\mc Z$ is a convex polyhedron $\mc Z = \{ z \in \R^n: Az \le b\}$ for some matrix $A$ and vector $b$ of appropriate dimensions, and if the linear programming
        \[
        \begin{array}{cl}
        \inf & b^\top w\\
        \st & A^\top w = \sum_{k\in[K]}\lambda_k\mu_k-\mu_1,~w\geq0
        \end{array}
        \]
        has an optimal solution for any $\lambda\in\Delta_K$, then
        \[
        \Lambda_{\ell_{\mc Z}, \delta}^W( \{\pi_k\}) = \left\{ \lambda \in \Delta_K: \exists w\geq0,~ b^\top w\leq \delta,~ A^\top w = \left(\sum_{k \in [K]}\lambda_k\m_k-\m_1\right)\right\}.
        \]
    
    \item If $\mc Z = \{ \mathbbm 1_n\}$, then
        \[
        \Lambda_{\ell_{\mc Z}, \delta}^W( \{\pi_k\}) = \left\{ \lambda \in \Delta_K: 
      \sum_{k \in [K]}\lambda_k\mathbbm1_n^\top\mu_k - \mathbbm1_n^\top\m_1\leq \delta
    \right\}.
    \]
    \end{enumerate}
\end{proposition}
% \begin{proof}[Proof of Proposition~\ref{prop:bias-W}]
%  Note that by Proposition~\ref{prop:consensus-W}, the first moment of the solution $\pi\opt$ to~\eqref{eq:consensus-W} satisfy
% \[
% \mathbb{E}_{\pi\opt}[\theta] = \msa =\sum_{k \in [K]}\lambda_k\m_k.
% \]  
% Now $(i)$ and $(iii)$ follow from definition. For $(ii)$, note that
% \[
%         \Lambda_{\ell_{\mc Z}, \delta}^W( \{\pi_k\}) = \left\{ \lambda \in \Delta_K: 
%       \sup_{Az\leq b} z^\top\left(\sum_{k \in [K]}\lambda_k\m_k-\m_1\right)\leq \delta
%     \right\}.
%     \] 
%     The dual linear programming of $\sup_{Az\leq b} z^\top\left(\sum_{k \in [K]}\lambda_k\m_k-\m_1\right)$ is
%       \[
%         \begin{array}{cl}
%         \inf & b^\top w\\
%         \st & A^\top w = \sum_{k\in[K]}\lambda_k\mu_k-\mu_1\\
%         & w\geq0,
%         \end{array}
%         \]
%     and by strong duality, the primal and dual optimal objective values are equal. Therefore we can reformulate $\sup_{Az\leq b} z^\top\left(\sum_{k \in [K]}\lambda_k\m_k-\m_1\right)\leq \delta$ as the existence of $w\geq0$ such that $b^\top w\leq\delta$ and $A^\top w = \sum_{k\in[K]}\lambda_k\mu_k-\mu_1$.
% \end{proof}

Because the aggregated covariance matrix $\covsa$ is not available in closed form  for general $K$, it is intractable to devise the optimal consensus mechanism under the Wasserstein distance in the general setting. Fortunately, if we impose the simultaneous diagonalizability assumption similar to Theorem~\ref{thm:refor}, then the optimal mechanism can be identified by solving a convex optimization problem, as highlighted in the next theorem.

\begin{theorem}[Optimal Wasserstein consensus mechanism] \label{thm:consensus-W}
    Suppose that there exists an orthogonal matrix $V \in \R^{n \times n}$ and a collection of diagonal matrices $D_k = \diag( d_k) \in \PD^n$ for $k \in [K]$ such that $\cov_k = V  D_k V^\top \succ 0$. Let $\lambda\opt$ be the optimal solution to the following problem
    \be \label{eq:consensus-W-refor}
        \begin{array}{cl}
        \min & \sum_{i, j \in [K]} \lambda_i \lambda_j \Tr{\cov_i^\half \cov_j^\half}\\
        \st & \lambda \in \Delta_K \\
        &  \ell_{\mc Z}\left(\sum_{k \in [K]}\lambda_k\m_k - \m_1\right) \le \delta.
    \end{array}
\ee
 Then $\mc C_{\lambda\opt}^W$ is the optimal consensus mechanism in $\mbb C^W$, that is, $\mc C_{\lambda\opt}^W$ is $(\ell_{\mc Z},\mu_1,\delta)$-bias tolerable and has minimal variance.
\end{theorem}
    Notice that the function $ \lambda \mapsto \sum_{i, j \in [K]} \lambda_i \lambda_j \Tr{\cov_i^\half \cov_j^\half} = \lambda^\top G \lambda$
    is convex in $\lambda$ because the matrix $G = (\Tr{\cov_i^\half \cov_j^\half})_{i, j \in [K]}$ is the Gram matrix of the linear reproducing kernel\footnote{More details on the reproducing kernel Hilbert space can be found in~\cite{ref:berlinet2004reproducing}.} on $\PD^n$ applied to $(\cov_1^\half, \ldots, \cov_K^\half)$, and thus $G$ is a positive semidefinite matrix. The result of Proposition~\ref{prop:bias-W} can now be incorporated into the reformulation~\eqref{eq:consensus-W-refor}, and lead to a convex finite-dimensional optimization problem, which can be solved efficiently using off-the-shelf solvers such as MOSEK~\citep{mosek}. 
    
   For the special case of $K=2$, we can leverage McCann's interpolation~\citep{ref:McCann1997convexity} to identify the optimal consensus mechanism without the simultaneous diagonalizability assumption.
    \begin{proposition} \label{prop:reformKis2}
    Suppose that $\pi_k \sim \N(\m_k, \cov_k)$ with $\cov_k\succ 0$ for $k=1,2$. Let $\lambda\opt$ be the optimal solution in the variable $\lambda$ of the following quadratic programming
     \be \label{eq:reformKis2}
        \begin{array}{cl}
        \min & \lambda_1^2\mathrm{Tr}(\Sigma_1)+\lambda_2^2\mathrm{Tr}(\Sigma_2)+2\lambda_1\lambda_2\mathrm{Tr}(\Sigma_1\Phi)\\
        \st & \lambda \in \Delta \\
        & \ell_{\mc Z}\left(\sum_{k \in [2]}\lambda_k\m_k - \m_1\right) \le \delta,
    \end{array}
\ee
  where $\Phi =\Sigma_2^{1/2}\left(\Sigma_2^{1/2}\Sigma_1\Sigma_2^{1/2}\right)^{-1/2}\Sigma_2^{1/2}$. Then $\mc C_{\lambda\opt}^W$ is the optimal consensus mechanism in $\mbb C^W$.
\end{proposition}
This special case allows us to interpolate easily between two posteriors. In our simulation and empirical analysis, we use this case for a Wasserstein interpolation between the two extreme posteriors of using either only the training data or the full panel. We demonstrate that the full Wasserstein mechanism for general $K$ dominates out-of-sample a simple Wasserstein interpolation, that is, fusing more layers of additional information results in a better look-ahead-bias and variance trade-off out-of-sample.

%%%%%%%%%%%%%%%%%%%%%%
\section{Optimal Backward Kullback-Leibler Consensus Mechanism}
\label{sec:backward}

In this section, we revisit the KL divergence setting, and we consider $\varphi$ as the backward KL-divergence, i.e., for any two probability distributions $\nu_1$ and $\nu_2$, we use $\varphi(\nu_1,\nu_2)=\KL(\nu_2\parallel\nu_1)$. Note that this choice of $\varphi$ places the posteriors $\pi_k$ as the first argument in the divergence. We use $\mc C_\lambda^B$ to denote the consensus mechanism induced by the backward KL-divergence, and $\mc C_\lambda^B$ is formally defined as follows.

\begin{definition}[Backward KL consensus mechanisms]
    A mechanism $\mc C_\lambda^B$ with weight $\lambda$ is called a backward KL consensus mechanism if, for any collection of posterior distributions $\{\pi_k\}$, we have
\be \label{eq:consensus-B}
    \mc C_\lambda^B(\{\pi_k\})  = \arg \min_{\pi \in \mc M}  \sum_{k \in [K]}\lambda_k~\KL (\pi_k \parallel \pi).
\ee
\end{definition} 

Given a simplex $\Delta_K$, we are interested in the set of backward KL mechanisms $\mbb C^B \Let \{ \mc C_\lambda^B: \lambda \in \Delta_K\}$. Problem~\eqref{eq:consensus-B} is infinite-dimensional as the decision variable $\pi$ ranges over the space of all probability distributions supported on $\mathbb{R}^n$. Fortunately, when the posteriors $\pi_k$ are Gaussian under Assumption~\ref{a:bayes}, we can characterize $\mc C_\lambda^B$ analytically. The next proposition states that $\mc C_\lambda^B$ is a mixture of normal distributions.

\begin{proposition}[Closed form of the consensus~{\cite[Theorem~1]{ref:da2020kullback}}] \label{prop:consensus-B}
    Suppose that $\pi_k \sim \N(\m_k, \cov_k)$ with $\cov_k \succ 0$ for any $k \in [K]$. For any $\lambda \in \Delta_K$, we have $\mc C_\lambda^B(\{\pi_k\}) = \sum_{k\in[K]}\lambda_k\pi_k$, a mixture of normal distributions.
\end{proposition}

We are now ready to state the main result of this section, which provides a tractable reformulation of problem~\eqref{eq:prob} under the backward KL consensus.
\begin{theorem}[Optimal backward KL mechanism] \label{thm:reversereform}
    Suppose that $\pi_k \sim \N(\m_k, \cov_k)$ with $\cov_k\succ 0$ for any $k \in [K]$. Let $\lambda\opt$ be the optimal solution in the variable $\lambda$ of the following problem
    \be \label{eq:reversereforquad}
        \begin{array}{cl}
        \min & \sum_{k \in [K]}\lambda_k \Tr{\m_k\m_k^\top+\cov_k}-\sum_{i,j\in[K]}\lambda_i\lambda_j\m_i^\top\m_j\\
        \st & \lambda \in \Delta_K \\
        &  \ell_{\mc Z}\left(\sum_{k \in [K]}\lambda_k\m_k - \m_1\right) \le \delta.
    \end{array}
\ee
 Then $\mc C_{\lambda\opt}^B$ is the optimal consensus mechanism in $\mbb C^B$, that is, $\mc C_{\lambda\opt}^B$ is $(\ell_{\mc Z},\mu_1,\delta)$-bias tolerable and has minimal variance.
\end{theorem}

The bias tolerance constraint in problem~\eqref{eq:reversereforquad} coincides with the bias tolerance constraint for the optimal Wasserstein consensus problem~\eqref{eq:consensus-W-refor}. Indeed, the result of Proposition~\ref{prop:consensus-B} implies that the aggregated mean for the backward KL consensus is also a weighted average of $\mu_k$, and the coincidence in the bias tolerance constraint follows. The result of Proposition~\ref{prop:bias-W} can be reapplied to reformulate the constraints of~\eqref{eq:reversereforquad} into either linear or second-order cone constraints. On the downside, notice that the objective function of~\eqref{eq:reversereforquad} is a concave quadratic function. Consequently, solving~\eqref{eq:reversereforquad} is generally difficult. While problem~\eqref{eq:reversereforquad} suffers from non-convexity, it does not require simultaneous diagonalizability as is required in Theorem~\ref{thm:refor}.

%%%%%%%%%%%%%%%%%%%%%%%%%%%%%%%%%%%
\section{Simulation and Empirical Analysis}
\label{sec:numerics}

We study the look-ahead-bias and variance trade-off of our BVMI framework in simulations and empirically. We work with a panel of asset returns data (either simulated or empirical data). We separate the panel data into three portions: \textit{training}, \textit{testing}, and an additional \textit{out-of-sample testing} periods, each having a length $T^{train},T^{test}$ and $T^{oos\text{-}test}$ respectively. For both the simulated and empirical data, we have access to the full panel of returns, which allows us to study various patterns of missing data. We obtain missing data by masking entries in the panel following certain missing mechanisms. We apply the masking procedure to the training period, while we preserve the completeness of testing and out-of-sample testing periods for evaluation purposes.\footnote{All optimization problems are implemented in Python. \label{fn:computer}We use an Intel Core i5 CPU (1.40 GHz) computer with 8GB of memory. Data and codes are available at \url{https://github.com/xuhuiustc/time-series-imputation}.} 

Following the BVMI framework, we feed the data set with missing values to the posterior generator module, as explained in Section~\ref{sec:generator}. Recall that we have a total of $K$ posterior distributions. The truncation times $T_k$'s are set to be equispaced between $T^{train}$ and $T$ with $T_1=T^{train}$ and $T_K=T$, where $T = T^{train}+T^{test}$. We assume a non-informative flat prior $\pi_0$ on the mean parameter $\theta$. The covariance matrix $\Omega$ is taken as the ground truth in the case of simulated data, or taken as the sample covariance matrix of the complete panel in the empirical study. Our focus is the first-order problem of the uncertainty in the mean estimation. Relative to the mean, the variance estimation is substantially more precise, as discussed among others in \cite{Kozak_etal2020}.

The generated multiple posterior distributions are used in the consensus mechanism module to output an optimal consensus posterior, the details of which we provide in Sections~\ref{sec:consensus} to~\ref{sec:backward}. More specifically, we use the following three aggregation mechanisms: 
\begin{enumerate}[label=(\roman*)]
    \item The forward KL mechanism corresponding to the reformulation~\eqref{eq:refor-F}.
    \item The full Wasserstein mechanism corresponding to  the reformulation~\eqref{eq:consensus-W-refor}, together with the set $\mathcal{Z}$ chosen as the Euclidean ball as in Proposition~\ref{prop:bias-W}~($i$).
\item The restricted Wasserstein mechanism suggested in Proposition~\ref{prop:reformKis2}, which (effectively) only considers the consensus posterior as the Wasserstein barycenter of the two naive posteriors $\pi_1$ and $\pi_K$.
\end{enumerate}
We do not consider the backward KL mechanism in Section~\ref{sec:backward} as it is non-convex and hence it is challenging to solve for large values of $K$. To satisfy the simultaneous diagonalizability assumption for the forward KL mechanism and the full Wasserstein mechanism, we project the covariance matrices of the posteriors to the same eigenbasis. 
We choose the orthogonal matrix $V$ (as in Lemma~\ref{lemma:projection-V}) as the eigenbasis of the covariance of the first posterior, and we follow Lemma~\ref{lemma:projection-V} \ref{lemma:projection-V1} to project the other posteriors onto this basis. As we have remarked, the tolerance parameter $\delta$ in~\eqref{eq:refor-F} and~\eqref{eq:consensus-W-refor} indicates how much of look-ahead-bias is accepted in the aggregated posterior $\pi\opt$. We select the $\delta$ parameter from the $10$ values equi-spaced between $0$ and $\delta_{max}$, where $\delta_{max}$ is computed using~\eqref{eq:biasF} or~\eqref{eq:euclideanball} by setting $\lambda_K=1$, i.e., $\delta_{max}=\max_{j\in[n]}\left|c_{Kj}/d_{Kj}^{-1}-v_j^\top\m_1\right|$ for the forward KL reformulation and $\delta_{max} = \|\mu_K-\mu_1\|_2^2$ for the Wasserstein reformulation.

Corresponding to each selected $\delta$ parameter, the aggregated posterior $\pi\opt$ is then output to the multiple imputation sampler module to create multiply-imputed data sets, as discussed in Section~\ref{sec:sampler}. More specifically, we follow the Conditional Expectation Bayesian Imputation algorithm, i.e., Algorithm~\ref{algo:conditional_Bayes}. In each round of imputation, we first sample the mean parameter $\theta$ from the aggregated posterior $\pi\opt$. Then for each $t\in[T^{train}]$, we impute $Y_t$ by $\theta_{Y_t}+\Omega_{Y_tX_t}(\Omega_{X_t})^{-1}(X_t-\theta_{X_t})$, which is the conditional mean of the posterior of $Y_t$ given $(\theta,\mc D)$.

Our downstream task is portfolio optimization. More specifically, we construct portfolios with the highest mean return. This captures the core of the problem and it is straightforward to extend this to other portfolio objectives. In more detail, the portfolio weights solve the following optimization problem during the training period:
\[
        \max \left\{ \EE_{\PP^{\text{imputed}}} [w^\top Z] : \| w \|_2 = 1 \right\},
\]
where $\tilde Z_t$ denotes the imputed value of $Z_t$ and $\PP^{\text{imputed}} = \sum_{t \in [T^{train}]} \delta_{\tilde Z_t}/T^{train}$ is the empirical distribution supported on the imputed data. The analytical solution equals %Equivalently, the allocation $w$ admits the analytical form
    \[
    w = \ds \frac{\frac{1}{T^{train}}\sum_{t\in[T^{train}]}\tilde Z_t}{\|\frac{1}{T^{train}}\sum_{t\in[T^{train}]}\tilde Z_t\|_2}. % = \frac{\sum_{t\in[T^{train}]}\tilde Z_t}{\|\sum_{t\in[T^{train}]}\tilde Z_t\|_2}.
    \]
  The estimated portfolio weights are used during the testing period, which results in the average portfolio return of
    \[
    R_{test} = \frac{1}{T^{test}}\sum_{t\in[T^{train}+1,T]} w^\top Z_t.
    \]
    Similarly, when the portfolio weights $w$ are applied on the out-of-sample testing period, we obtain the averaged portfolio return 
    \[
    R_{oos\text{-}test}=\frac{1}{T^{oos\text{-}test}}\sum_{t\in[T+1,T+T^{oos\text{-}test}]} w^\top Z_t.
    \]
The difference in the averaged portfolio return between the two periods defines our \textit{regret} measure
\[
\Delta R = R_{test}-R_{oos\text{-}test}.
\]
Recall that $\mc D$ is the realization of observed data and missing masks. Conditional on $\mc D$, $\Delta R$ is a random variable where the randomness comes from multiple imputations. Thus we have the conditional moments $\mathbb{E}[\Delta R|\mc D]$ and $\mathbb{V}\mathrm{ar}[\Delta R|\mc D]$, whose randomness can, in turn, be averaged out over realizations of $\mathcal{D}$.  We are interested in finding a $\delta$ that minimizes the ``expected conditional mean squared error'', or ECMSE, defined as
% \[
% (\mathbb{E}[\Delta R])^2
% \]
\[
     \max\{\mathbb{E}[\Delta R],0\}^2 + \mathbb{E}[\mathbb{V}\mathrm{ar}[\Delta R|\mc D]].
\]
Observe that this ECMSE shares the same bias-variance trade-off structure as problem~\eqref{eq:prob} (with $\mathcal{Z}$ chosen according to Example~\ref{ex:penalty}~($i$)) with the exception that we consider a one-sided bias in the above formulation as an investor's objective is to avoid a negative regret or disappointment. We denote the first term above as ECBias$^2$ and the second term above as ECVar. In the following we study the three measures ECMSE, ECBias$^2$ and ECVar as $\delta$ varies, under four different missing mechanisms that include missing at random, a dependent block missing structure, and when the missing pattern depends on realizations of returns (noticeably the last missing pattern does not satisfy the ignorability assumption, i.e., our Assumption~\ref{a:ignorable}).

\subsection{Simulation Results}
\label{sec:synexp}

In our simulation, we generate returns for $n=10$ stocks based on a factor model and study the effect of data imputation for different missing patterns. We consider the time horizons $T^{train}=T
^{test}=100$, $T^{oos\text{-}test}=1000$. The large size of testing set is designed to amplify the look-ahead-bias. We generate the returns vector from a factor model, namely, the vectors $Z_t$, $t=1,\ldots,T^{train}+T^{test}+T^{oos\text{-}test}$ are generated~i.i.d.~from $\mc N_n(\theta,\Omega)$, where for each $i=1,\ldots,n$, $\theta_i=0.2\beta_i+\alpha_i$, $\beta_i=1$, $\alpha_i$ is equispaced between $-0.3$ and $0.3$, and $\Omega = \beta\beta^\top+I_n$.

We consider four types of missing mechanisms.
\begin{itemize}
    \item \textit{Missing completely at random} -- we mask each entry in the training period as missing independently at random with a fixed probability $50\%$.
    \item \textit{Missing at random} -- for each stock $i$, we sample a Bernoulli variable $S_i$ with success probability $50\%$. If $S_i=1$, we mask the entries corresponding to stock $i$ as missing independently at random with probability $50\%$, otherwise we mask the entries corresponding to stock $i$ as missing independently at random with probability $70\%$.
    \item \textit{Block missing} -- we mask the first $30\%$ of the training period as missing. 
    \item \textit{Missing by value} -- we mask an entry in the training period as missing if it has an absolute value greater than $0.3$. 
\end{itemize}

\begin{figure}[h!]
\tcapfig{Model comparison on simulated data}
    \centering
    \subfigure[Missing completely at random]{\label{fig:synth-mcar-comp}
		\includegraphics[width=0.4\columnwidth]{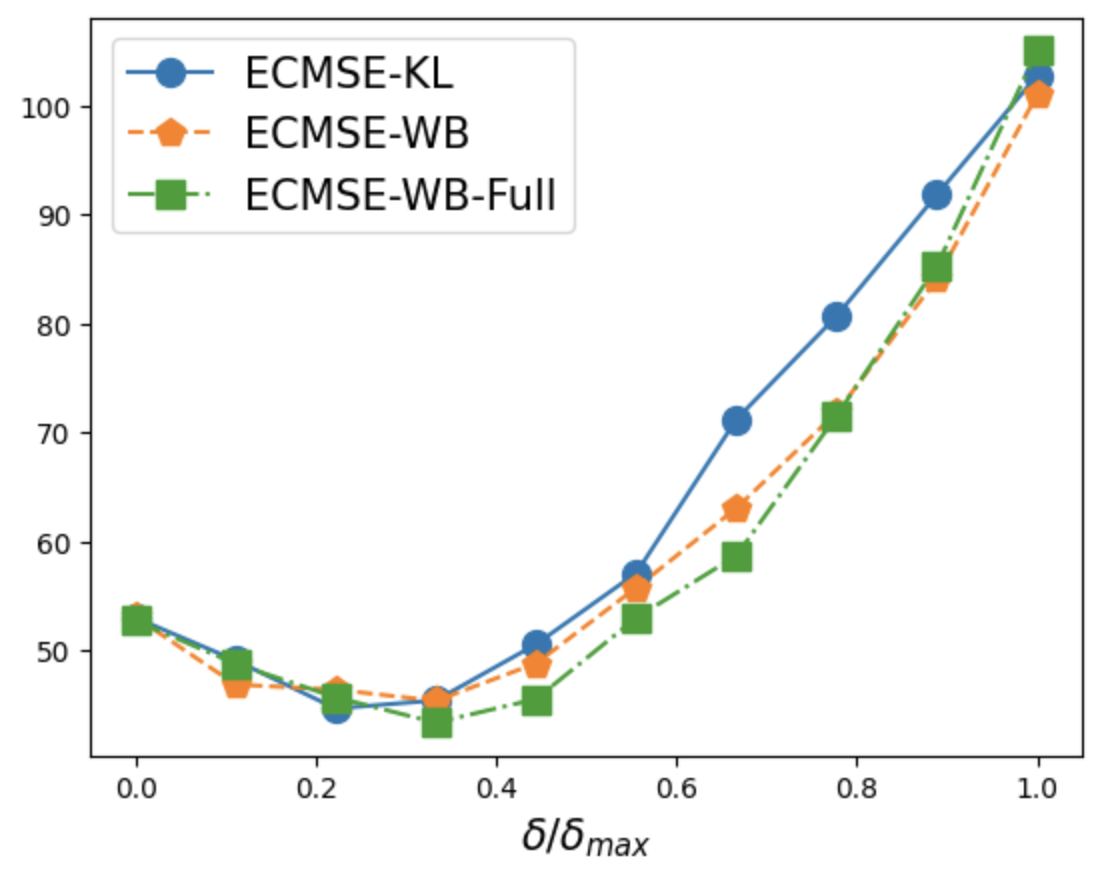}} \hspace{1mm}
	\subfigure[Missing at random]{\label{fig:synth-mar-comp}
	\includegraphics[width=0.4\columnwidth]{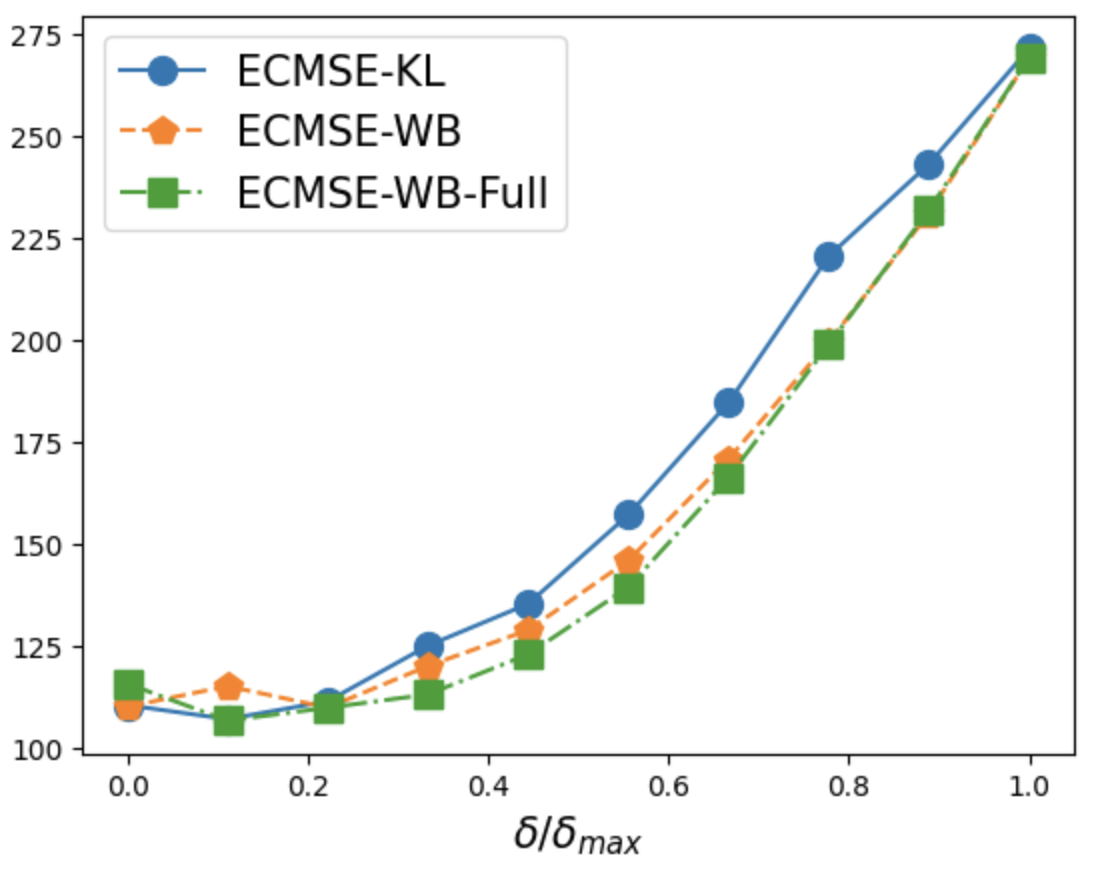}}
	    \hspace{1mm}
	\subfigure[Block missing]{\label{fig:synth-block-comp}
		\includegraphics[width=0.4\columnwidth]{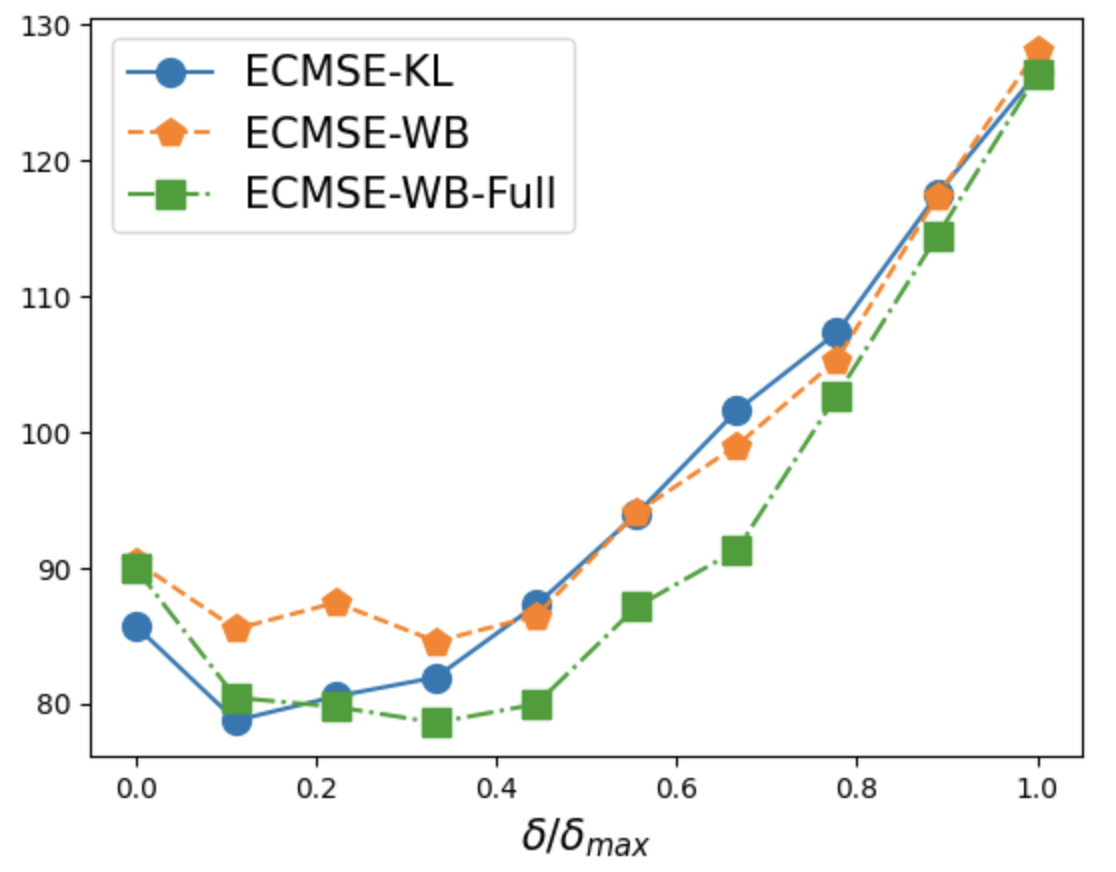}} 
	\subfigure[Missing by value]{\label{fig:synth-nonrandom-comp} 
		\includegraphics[width=0.4\columnwidth]{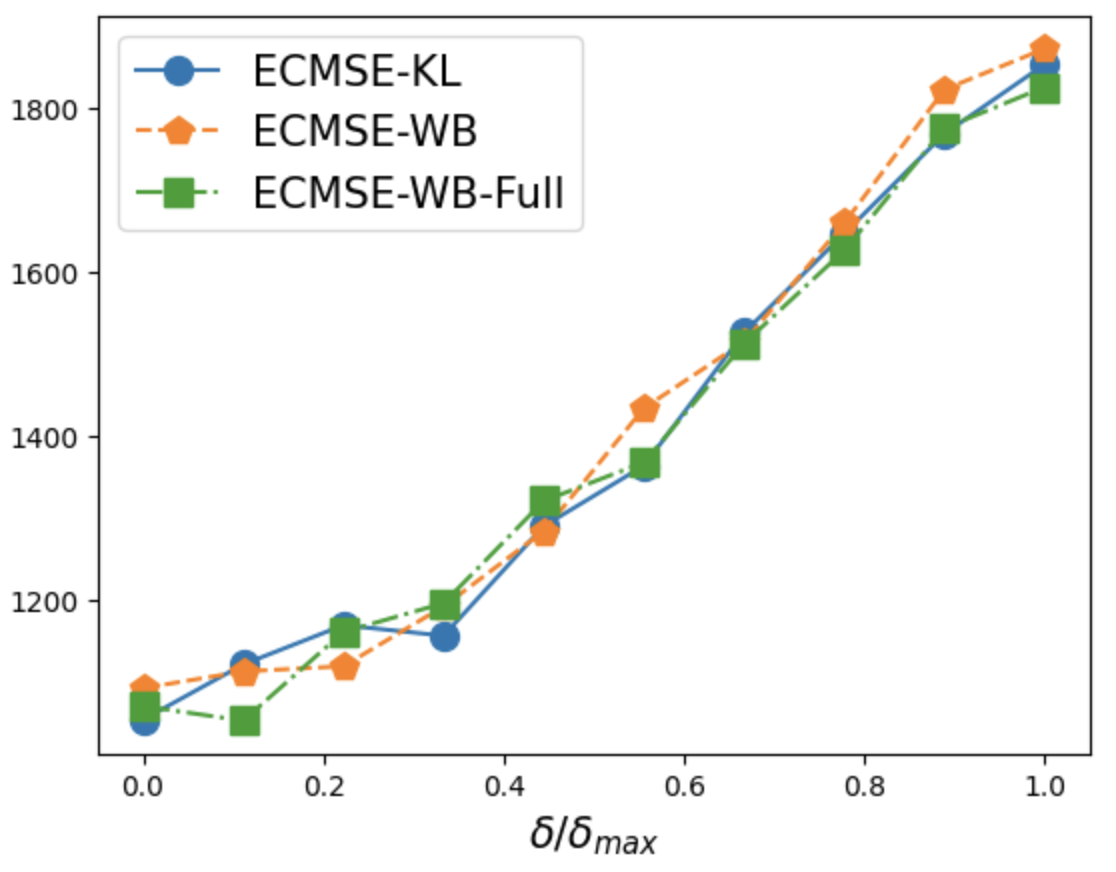}} 
		\bnotefig{These figures compare the expected conditional mean squared error (ECMSE) for different imputation schemes on simulated data. We impute the missing data using the forward KL mechanism (blue solid line with circle), full Wasserstein mechanism (green dash-dot line with square), and restricted Wasserstein mechanism (dashed orange line with pentagon) for different weights $\delta$, which control the trade-off between look-ahead-bias and variance.}
    \label{fig:synth-comp}
\end{figure}

For each missing mechanism, we run $500$ independent simulations, wherein for each simulated sample the imputation is repeated for $10$ times, and we carry out the downstream task on each imputed data set. 
Our main focus is the effect of the parameter $\delta$ for the three aggregation mechanisms, forward KL mechanism, full Wasserstein mechanism, and restricted Wasserstein mechanism. Figure~\ref{fig:synth-comp} shows the comparison of the ECMSE for the three mechanisms and the four missing patterns. Importantly, for all missing patterns and mechanisms, the expected conditional mean squared error is minimized for a non-trivial value of $\delta$. This means the special cases of no-look-ahead-bias ($\delta=0$) or using all the data with maximum look-ahead-bias ($\delta=\delta_{max}$) are not optimal. Moreover, either the forward KL mechanism or the full Wasserstein mechanism attains a smaller minimum ECMSE than the restricted Wasserstein mechansim, showing the power of using more layers of information in the aggregation mechanism.

\begin{figure}[h!]
\tcapfig{Look-ahead-bias and variance with forward KL mechanism on simulated data}
    \centering
    \subfigure[Missing completely at random]{\label{fig:synth-mcar-kl}
		\includegraphics[width=0.4\columnwidth]{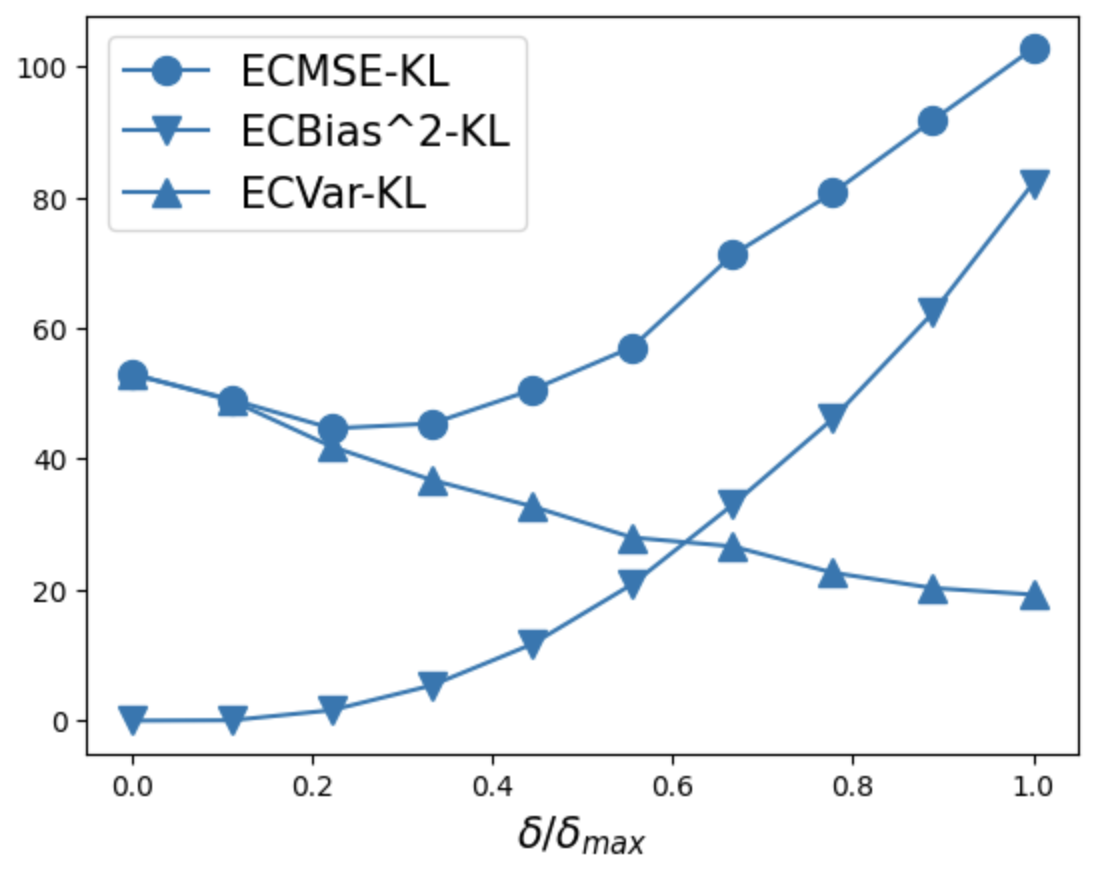}} \hspace{1mm}
	\subfigure[Missing at random]{\label{fig:synth-mar-kl}
	\includegraphics[width=0.4\columnwidth]{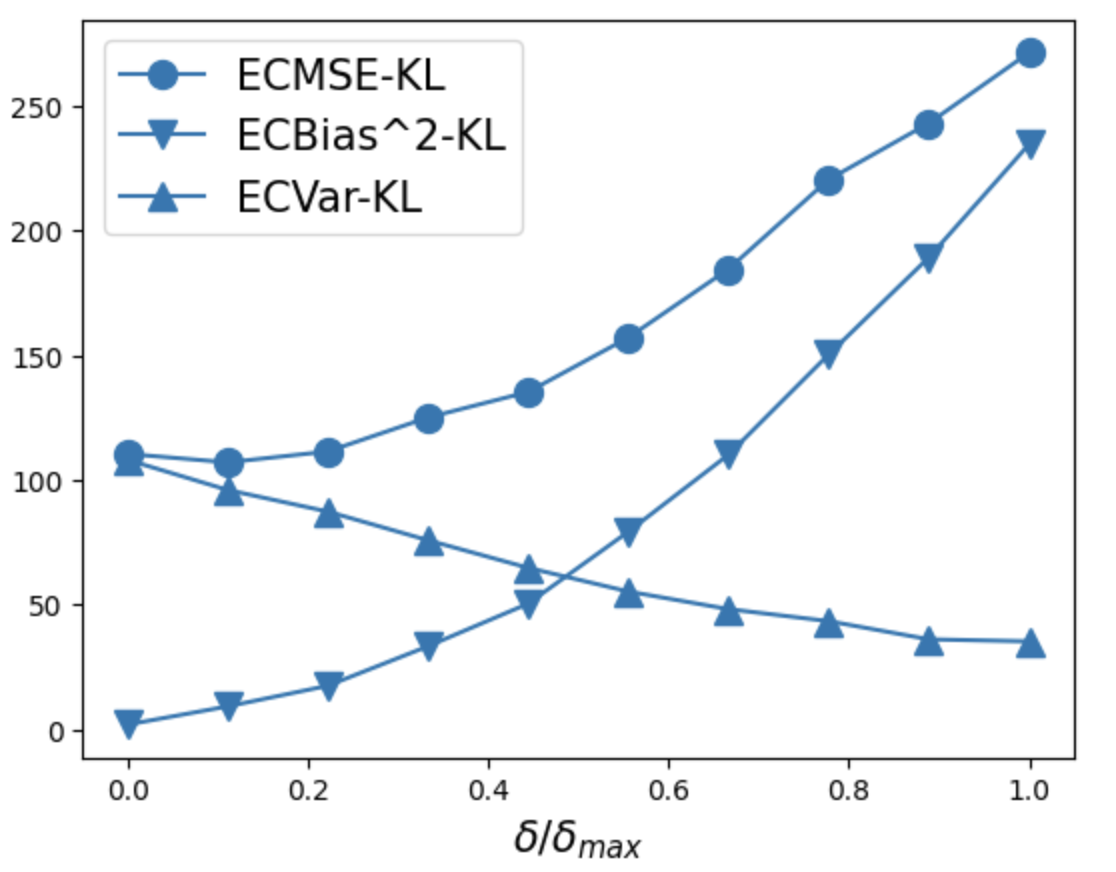}}
	    \hspace{1mm}
	\subfigure[Block missing]{\label{fig:synth-block-kl}
		\includegraphics[width=0.4\columnwidth]{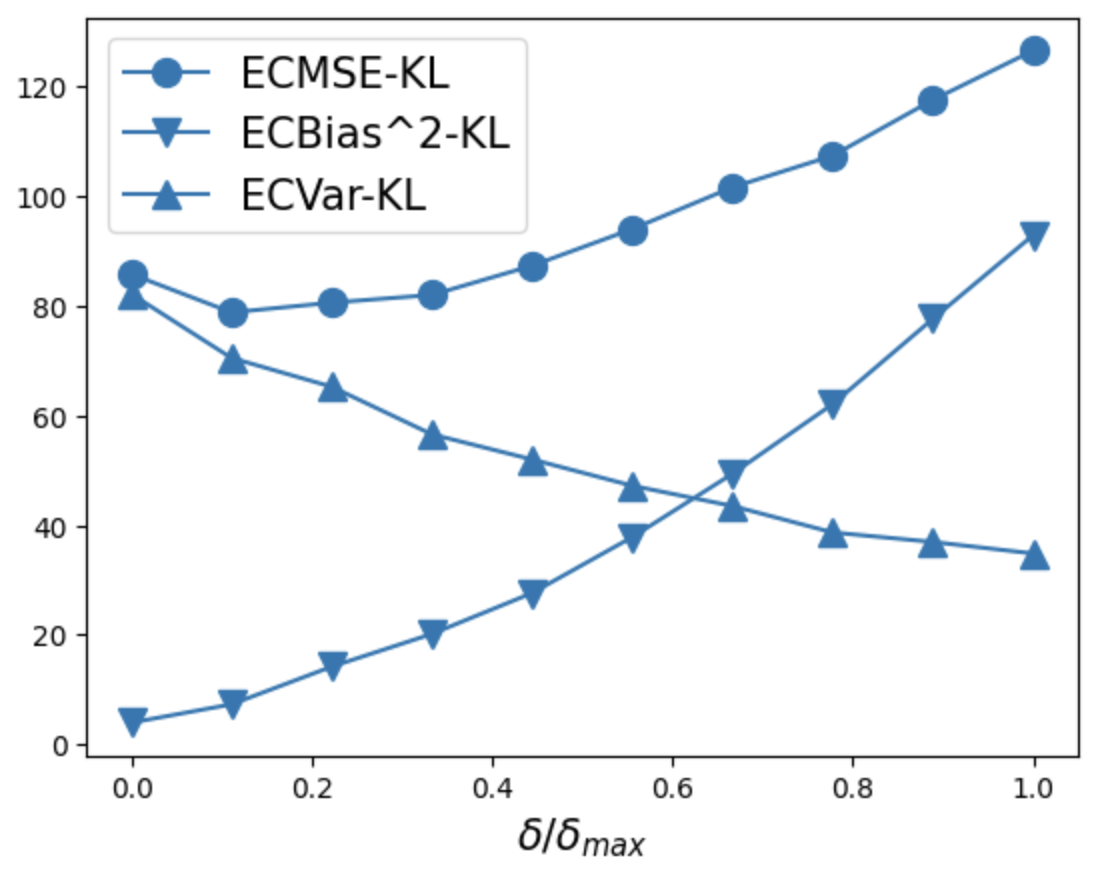}} 
	\subfigure[Missing by value]{\label{fig:synth-nonrandom-kl} 
		\includegraphics[width=0.4\columnwidth]{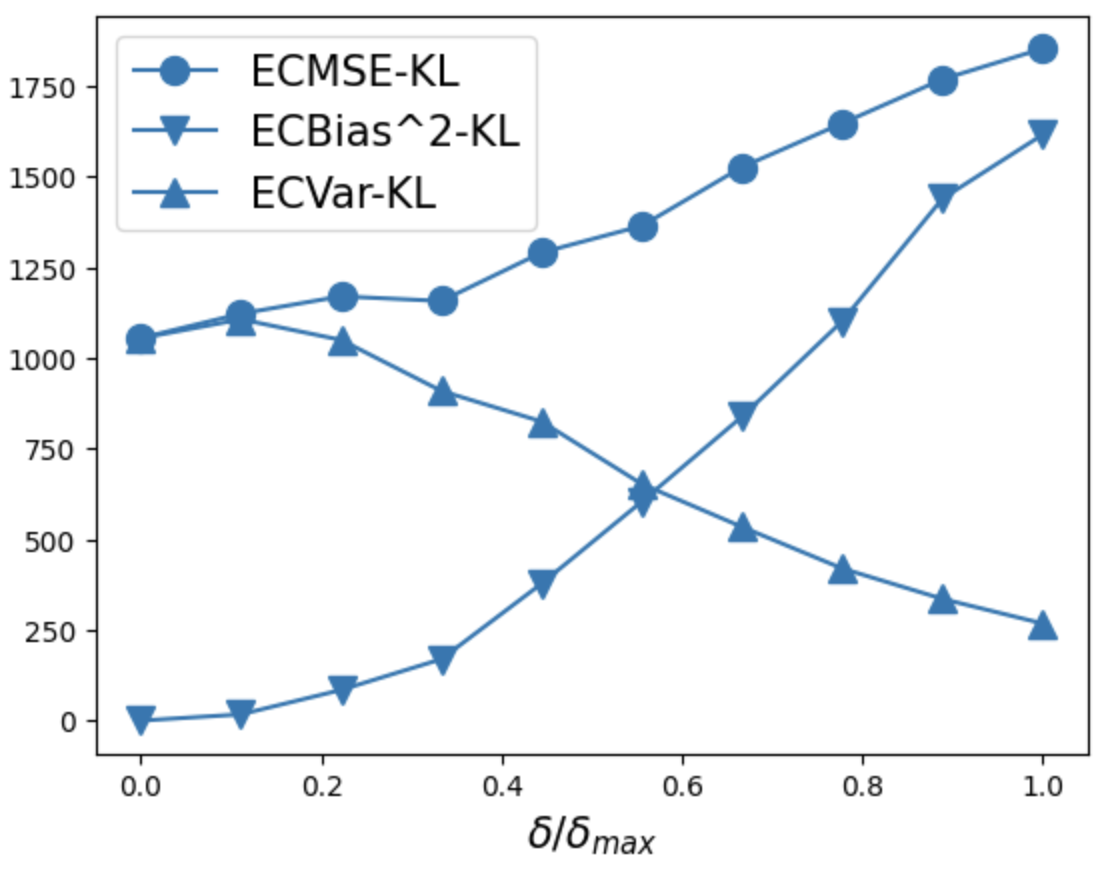}} 
	\bnotefig{These figures show the look-ahead-bias measured by the regret ECBias$^2$, the variance ECVar and the expected conditional mean squared error ECMSE for simulated data with the forward KL mechanism for different weights $\delta$, which control the trade-off between look-ahead-bias and variance.}
    \label{fig:synth-kl}
\end{figure}

Figure \ref{fig:synth-kl} sheds further light on the underlying trade-off and shows the ECBias$^2$ and ECVar for the forward KL mechanism. In all cases, the look-ahead-bias measured by ECBias$^2$ is increasing in $\delta$, while the variance measured by ECVar decreases. This is the reason why ECMSE can achieve its optimal value for non-trivial values of $\delta$. Figures \ref{fig:synth-wbfull} and \ref{fig:synth-wb} in the Appendix show the corresponding results for the full and restricted Wasserstein mechanism with the same monotonicity results.

\subsection{Empirical Application}

\label{sec:experiment}

In our empirical application, we consider two different panels of daily stock returns that cover different time periods. Our first data set is a representative sample of $n=10$ large capitalization stocks\footnote{The stocks are `DVA', `DE', `DAL',`ESRX', `XOM', `FFIV', `FB', `FAST', `ESS' and `EL'.} from 01/01/2015 to 01/01/2017. Our second panel consists of the $n=10$ standard industry portfolios from Kenneth French's data library from 01/03/2017 to 01/02/2019.\footnote{The industry portfolios are `NoDur', `Durbl', `Manuf', `Enrgy', `HiTec', `Telcm', `Shops', `Hlth', `Utils' and `Other', downloaded from \url{https://mba.tuck.dartmouth.edu/pages/faculty/ken.french/data_library.html}.} We consider $T^{train}=200$, $T^{test}=100$ and $T^{oos\text{-}test}=100$. All returns are annualized. Similar to before, we consider four types of missing mechanisms:

\begin{figure}[h!]
\tcapfig{Model comparison for individual stocks}
    \centering
    \subfigure[Missing completely at random]{\label{fig:real-mcar-s} 
		\includegraphics[width=0.4\columnwidth]{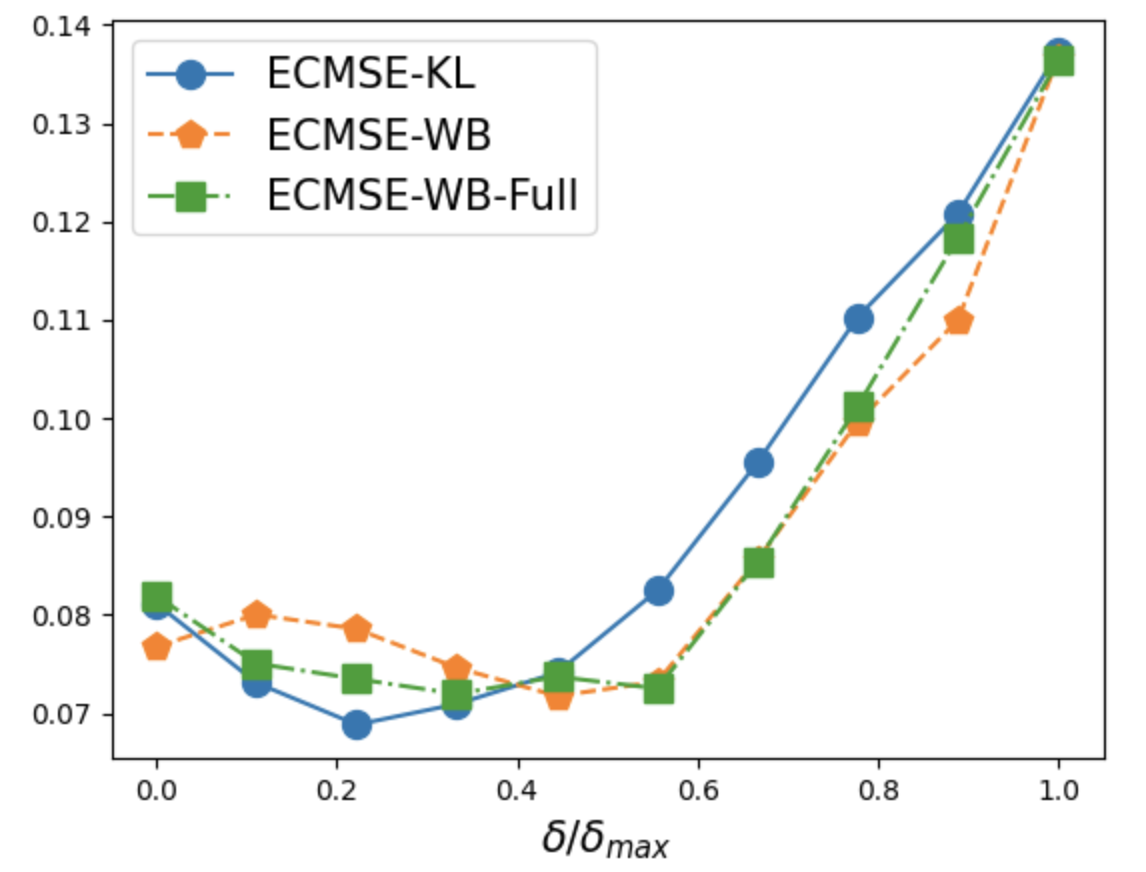}} \hspace{1mm}
	\subfigure[Missing at random]{\label{fig:real-mar-s}
	\includegraphics[width=0.4\columnwidth]{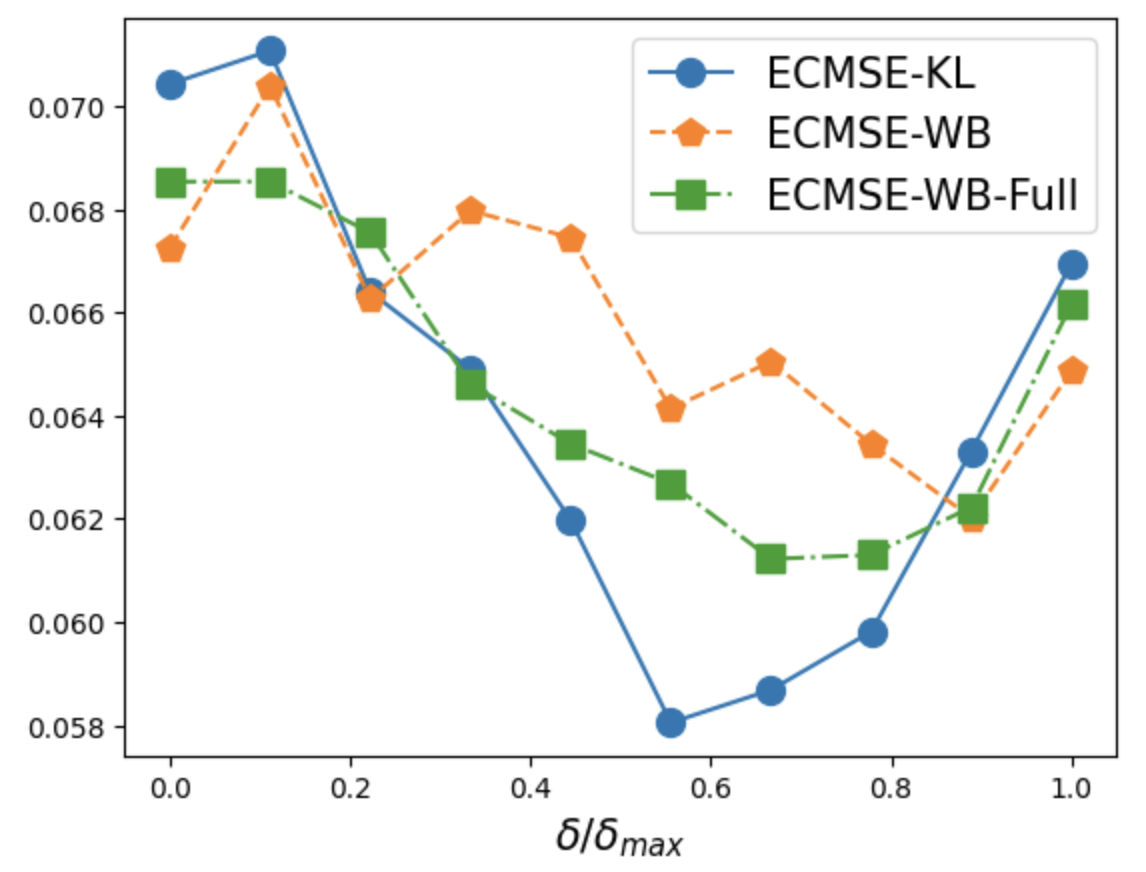}}
	    \hspace{1mm}
	\subfigure[Block missing]{\label{fig:real-block-s} 
		\includegraphics[width=0.4\columnwidth]{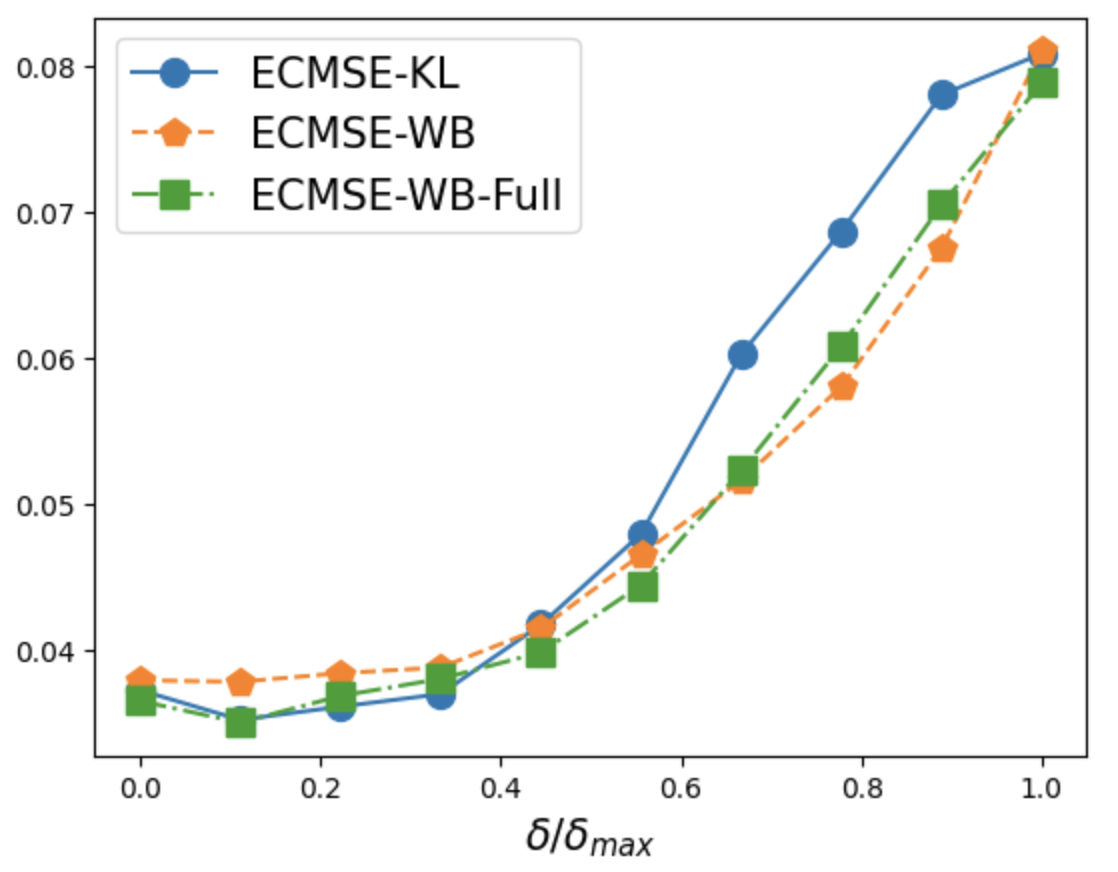}} 
		\subfigure[Missing by value]{\label{fig:real-nonrandom-s} 
		\includegraphics[width=0.4\columnwidth]{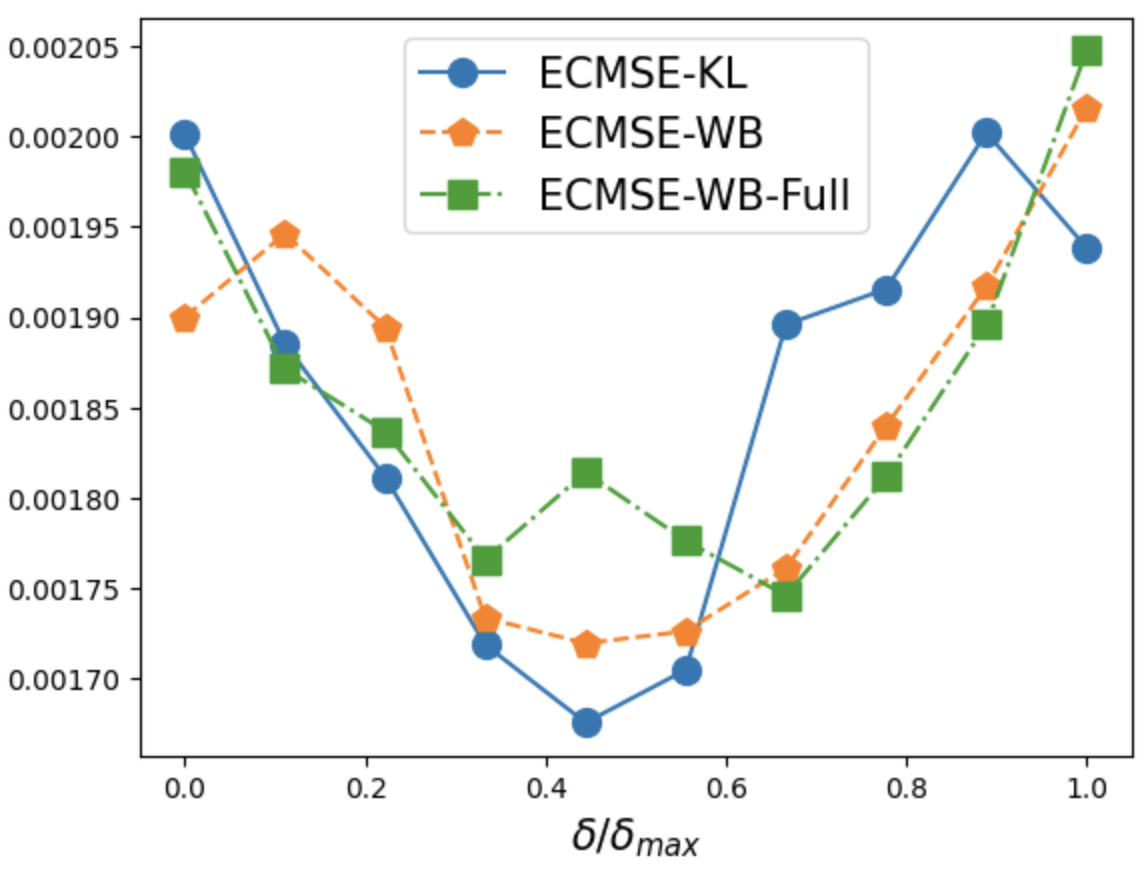}} 
				\bnotefig{These figures compare the expected conditional mean squared error (ECMSE) for different imputation schemes for individual stock returns. We impute the missing data using the forward KL mechanism (blue solid line with circle), full Wasserstein mechanism (green dash-dot line with square), and restricted Wasserstein mechanism (dashed orange line with pentagon) for different weights $\delta$, which control the trade-off between look-ahead-bias and variance. The samples are daily returns of 10 representative large capitalization stocks from 01/01/2015 to 01/01/2017.}
    \label{fig:reals}
\end{figure}

\begin{figure}[h!]
\tcapfig{Look-ahead-bias and variance with forward KL mechanism for individual stocks}
    \centering
    \subfigure[Missing completely at random]{\label{fig:real-mcar-s-kl} 
		\includegraphics[width=0.4\columnwidth]{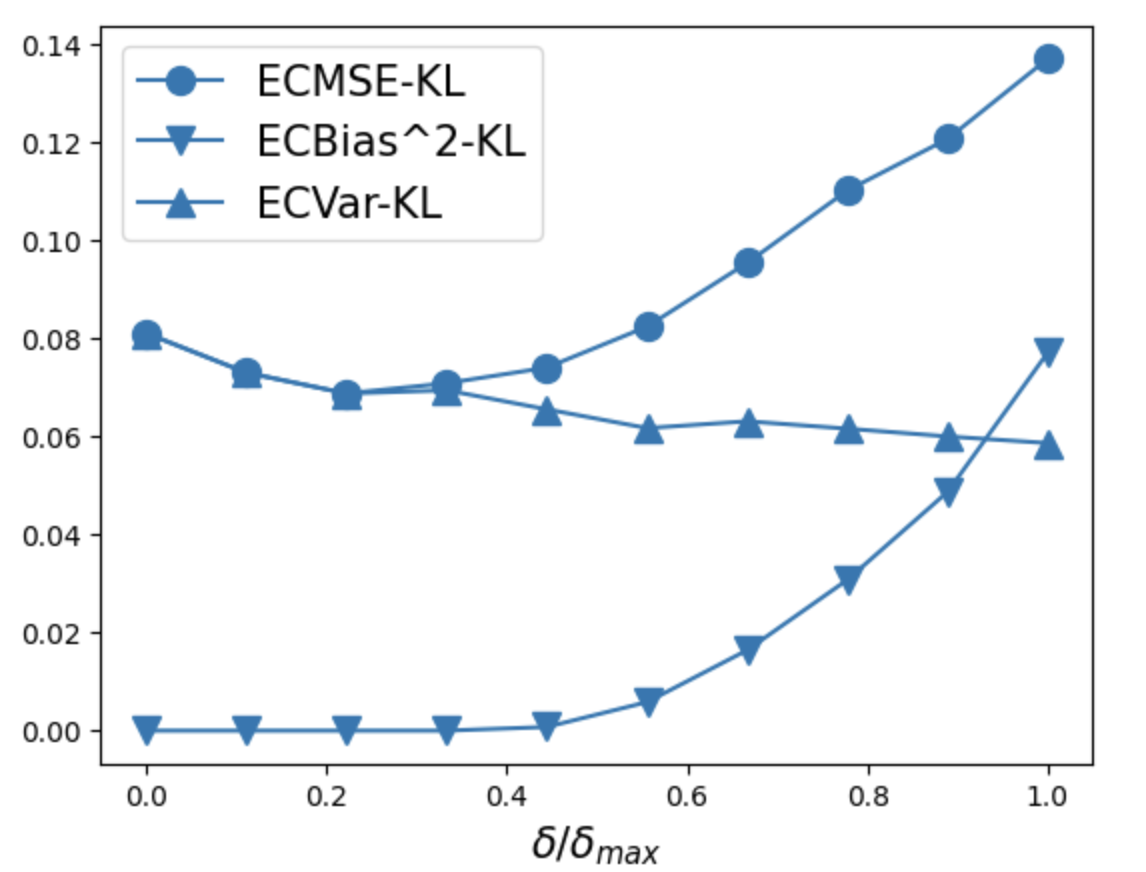}} \hspace{1mm}
	\subfigure[Missing at random]{\label{fig:real-mar-s-kl}
	\includegraphics[width=0.4\columnwidth]{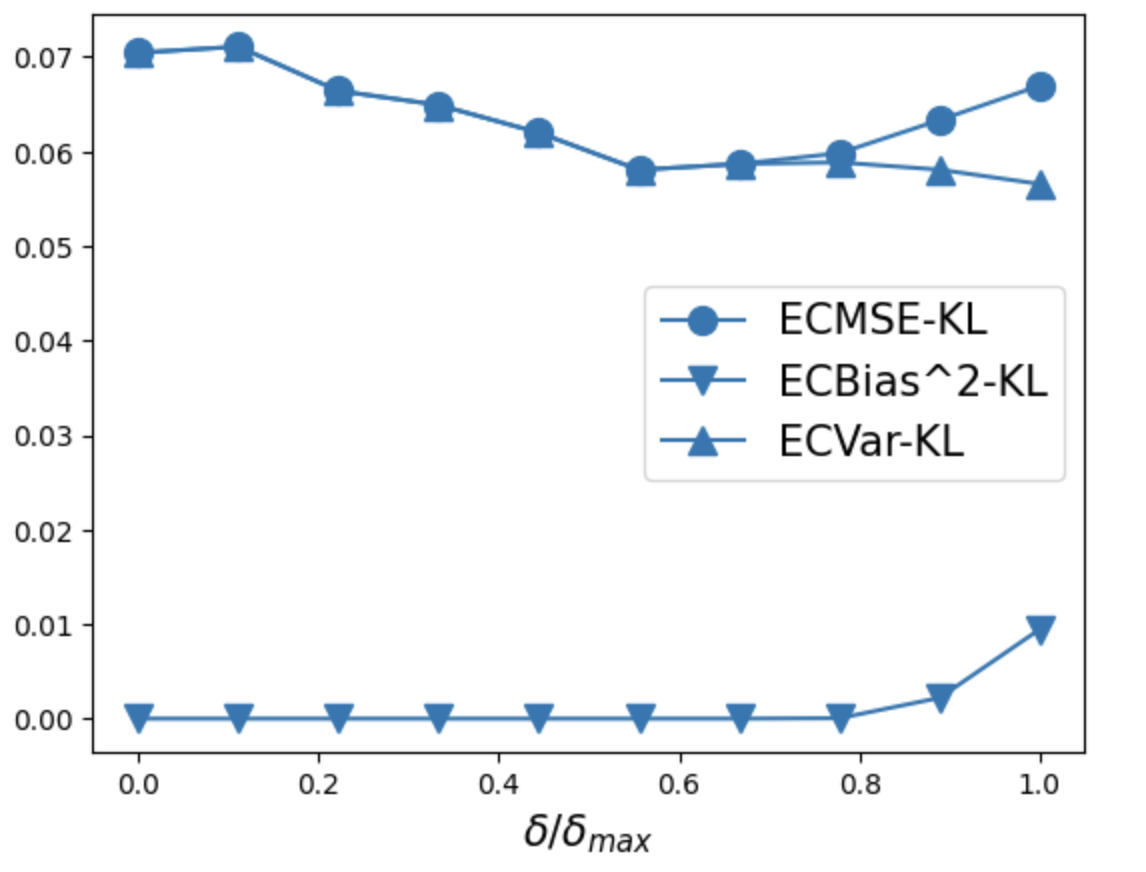}}
	    \hspace{1mm}
	\subfigure[Block missing]{\label{fig:real-block-s-kl} 
		\includegraphics[width=0.4\columnwidth]{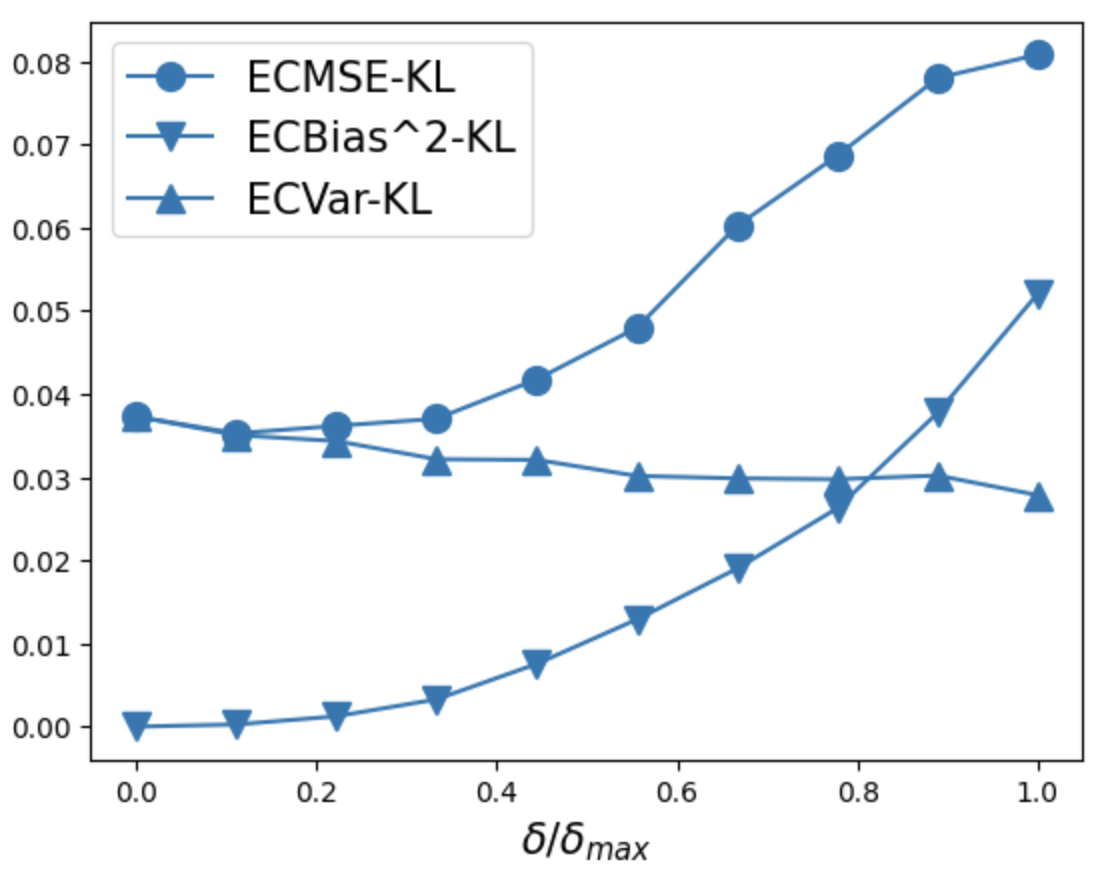}} 
		\subfigure[Missing by value]{\label{fig:real-nonrandom-s-kl} 
		\includegraphics[width=0.4\columnwidth]{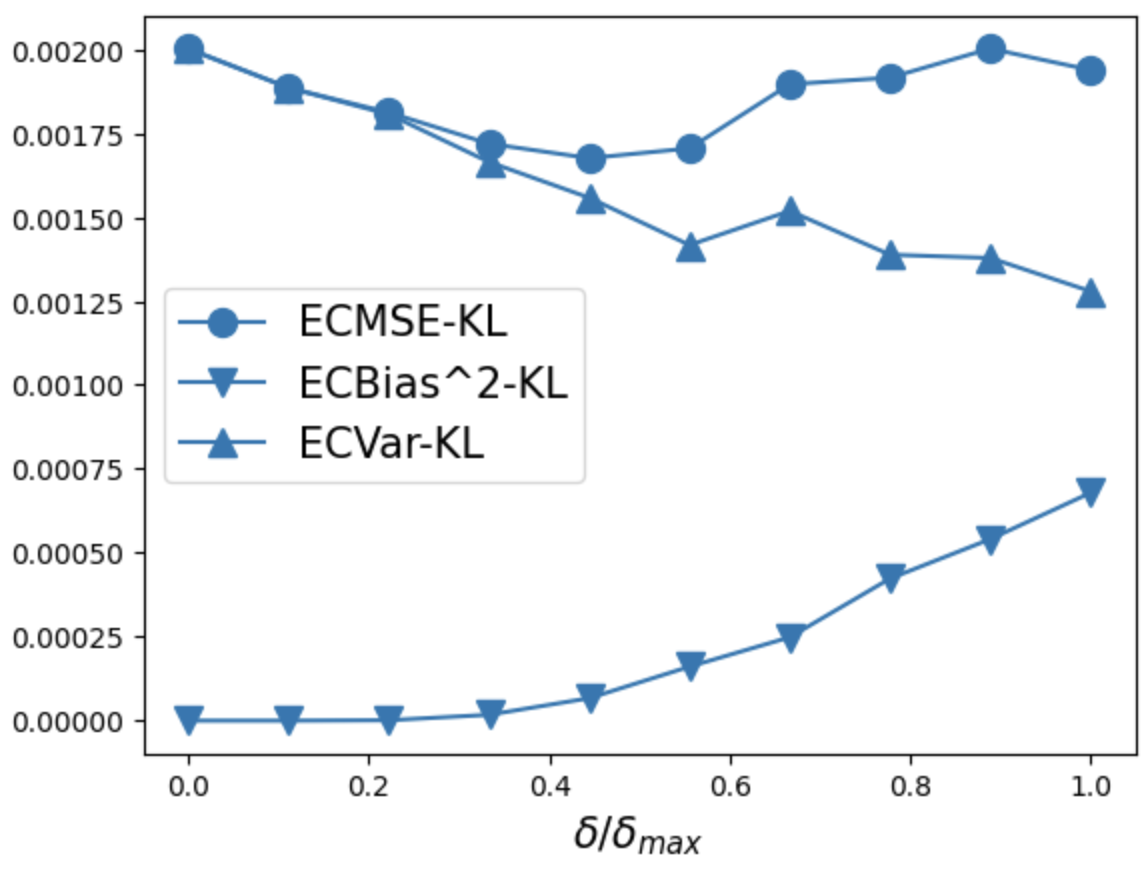}} 
			\bnotefig{These figures show the look-ahead-bias measured by the regret ECBias$^2$, the variance ECVar and the expected conditional mean squared error ECMSE for individual stocks with the forward KL mechanism for different weights $\delta$, which control the trade-off between look-ahead-bias and variance. The samples are daily returns of 10 representative large capitalization stocks from 01/01/2015 to 01/01/2017.}
  \label{fig:reals-kl}
\end{figure}

\begin{figure}[h!]
\tcapfig{Model comparison for industry portfolios}
    \centering
    \subfigure[Missing completely at random]{\label{fig:real-mcar-17} 
		\includegraphics[width=0.4\columnwidth]{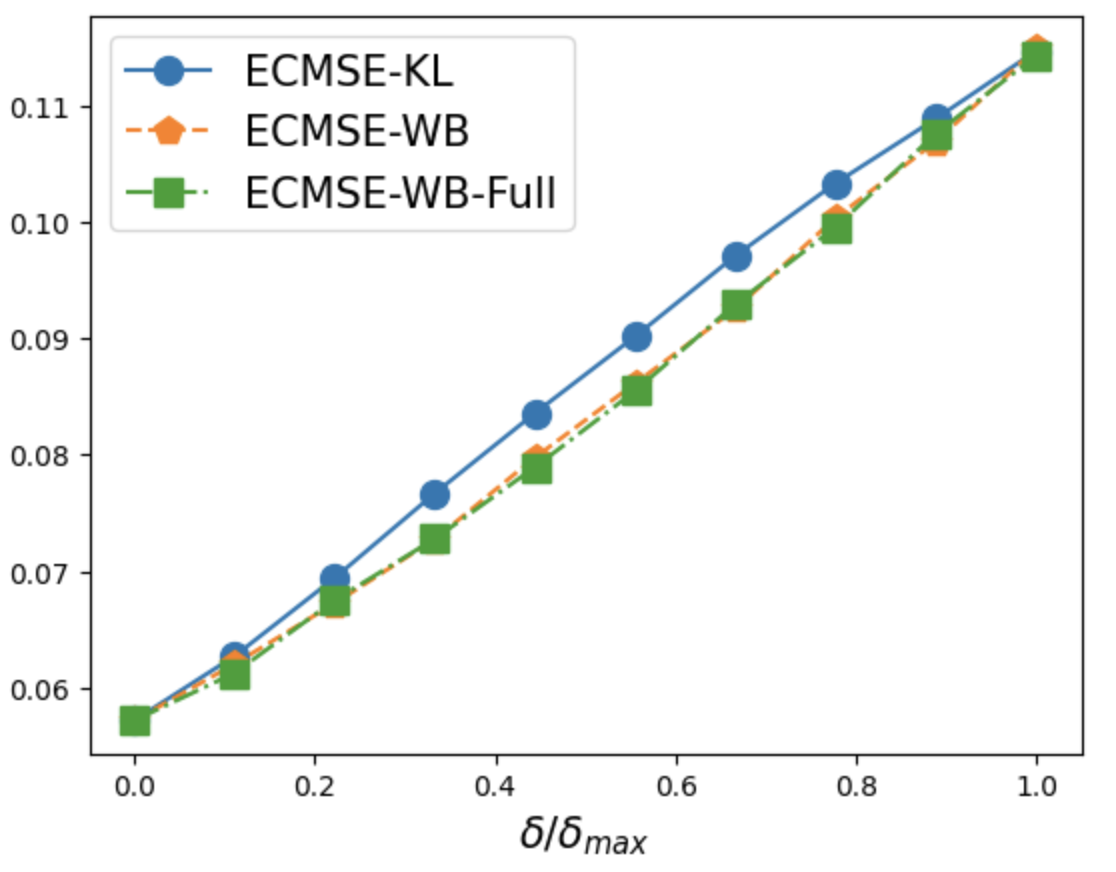}} \hspace{1mm}
	\subfigure[Missing at random]{\label{fig:real-mar-17}
	\includegraphics[width=0.4\columnwidth]{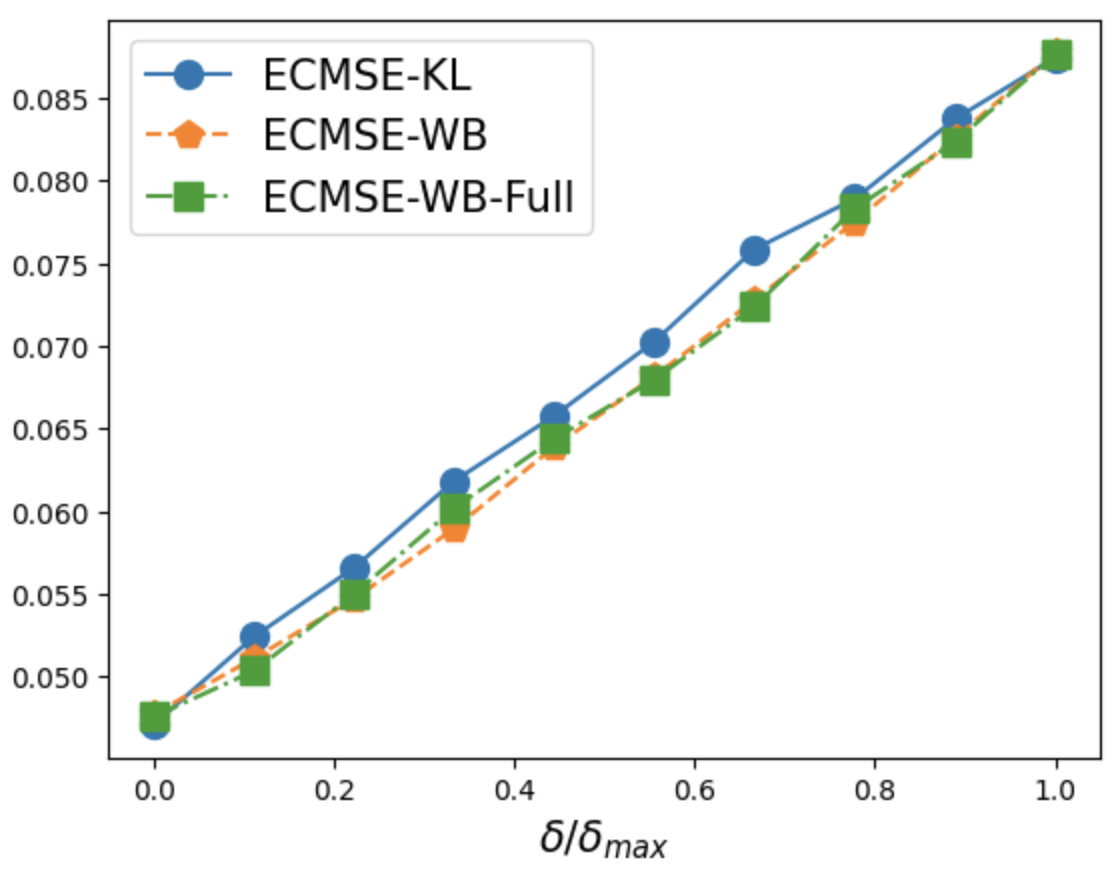}}
	    \hspace{1mm}
	\subfigure[Block missing]{\label{fig:real-block-17} 
		\includegraphics[width=0.4\columnwidth]{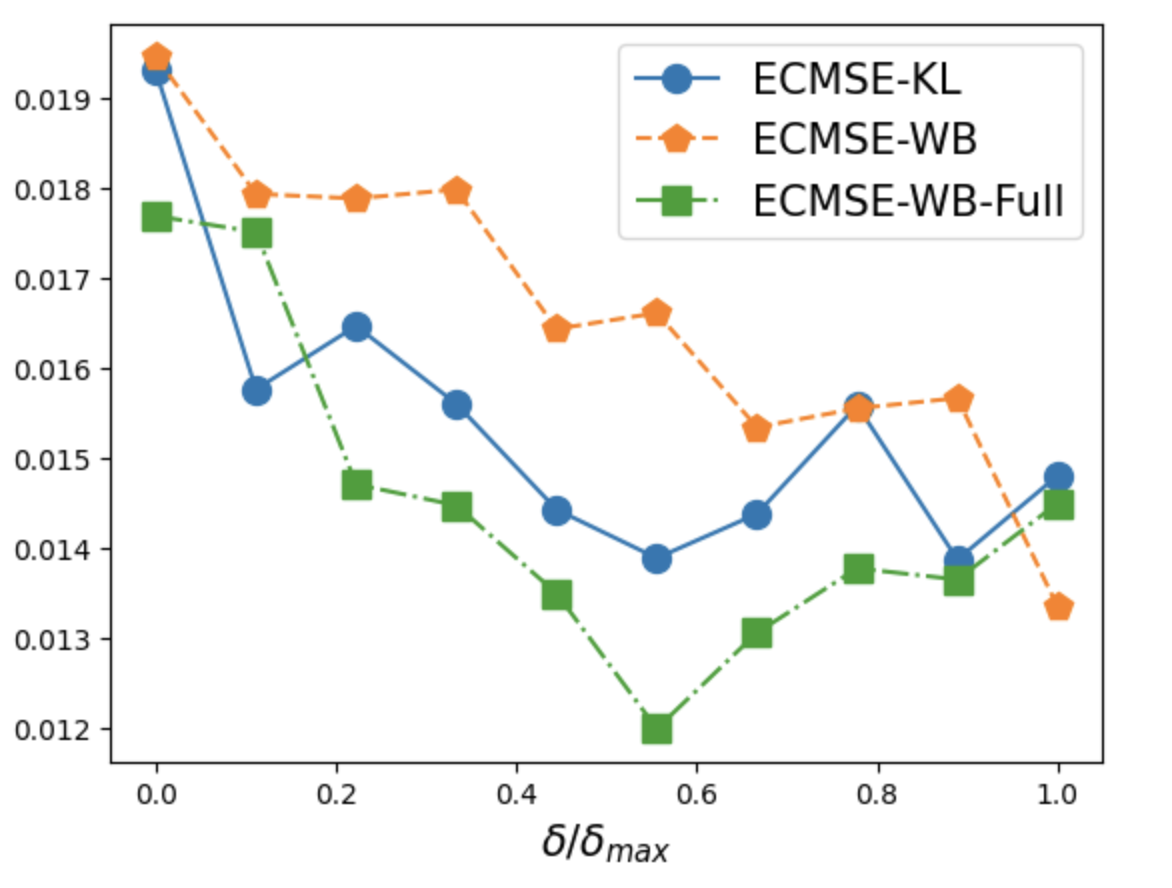}} 
		\subfigure[Missing by value]{\label{fig:real-nonrandom-17} 
		\includegraphics[width=0.4\columnwidth]{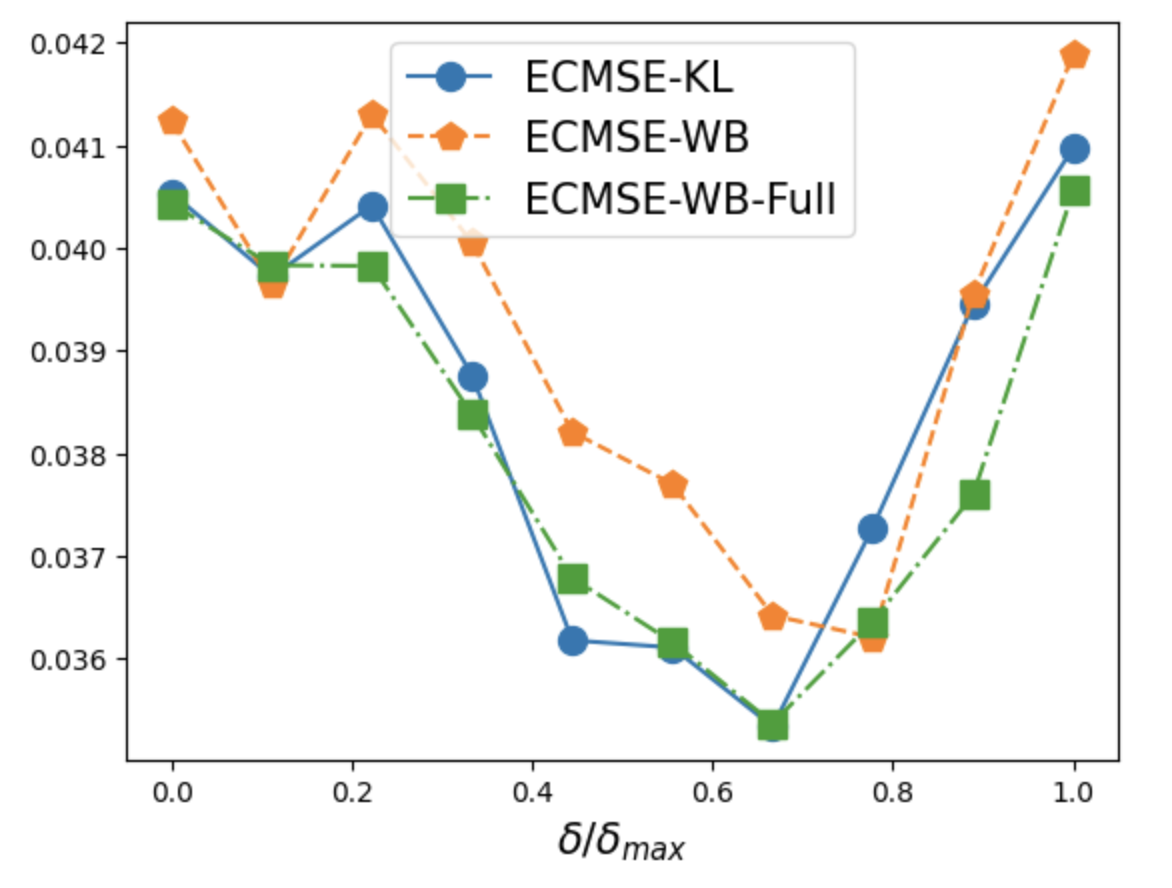}} 
	\bnotefig{These figures compare the expected conditional mean squared error (ECMSE) for different imputation schemes for portfolio returns. We impute the missing data using the forward KL mechanism (blue solid line with circle), full Wasserstein mechanism (green dash-dot line with square), and restricted Wasserstein mechanism (dashed orange line with pentagon) for different weights $\delta$, which control the trade-off between look-ahead-bias and variance. The samples are daily returns of 10 industry portfolios from 01/03/2017 to 01/02/2019.}
    \label{fig:real17}
\end{figure}

\begin{figure}[h!]
\tcapfig{Look-ahead-bias and variance with forward KL mechanism for industry portfolios}
    \centering
    \subfigure[Missing completely at random]{\label{fig:real-mcar-17-kl} 
		\includegraphics[width=0.4\columnwidth]{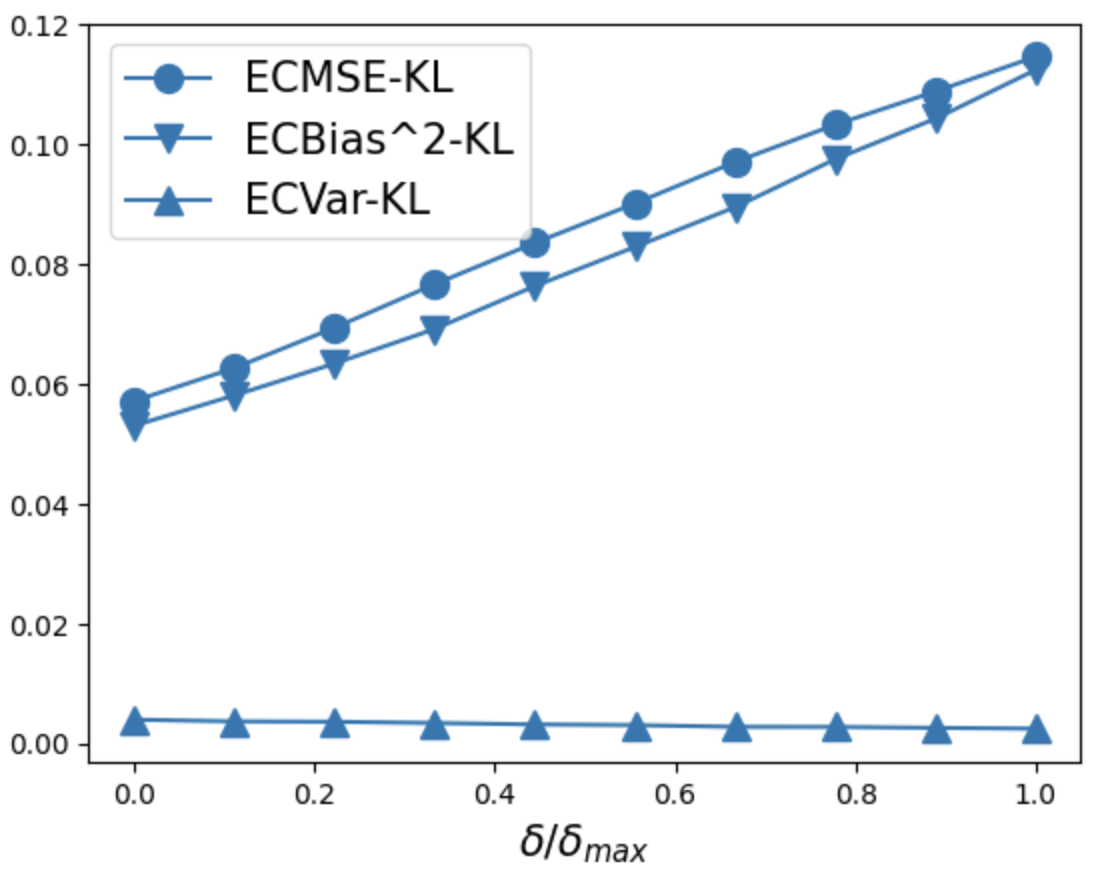}} \hspace{1mm}
	\subfigure[Missing at random]{\label{fig:real-mar-17-kl}
	\includegraphics[width=0.4\columnwidth]{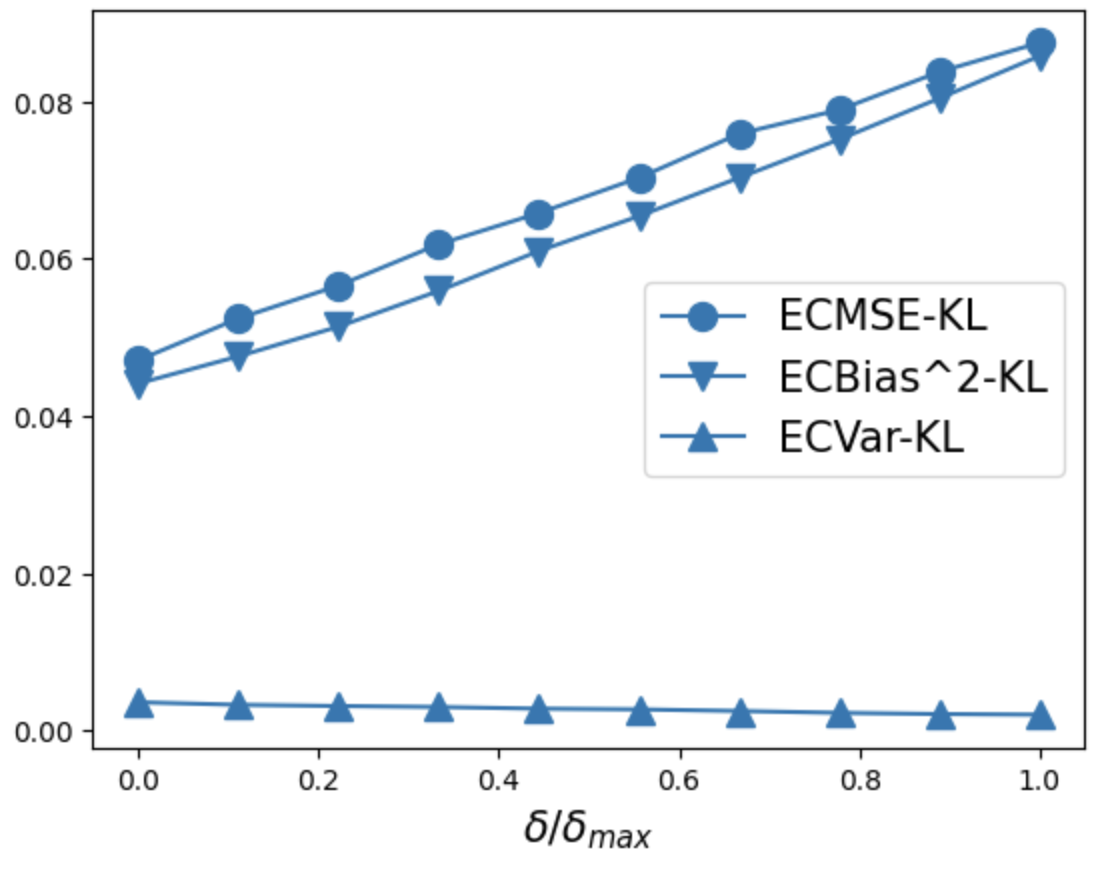}}
	    \hspace{1mm}
	\subfigure[Block missing]{\label{fig:real-block-17-kl} 
		\includegraphics[width=0.4\columnwidth]{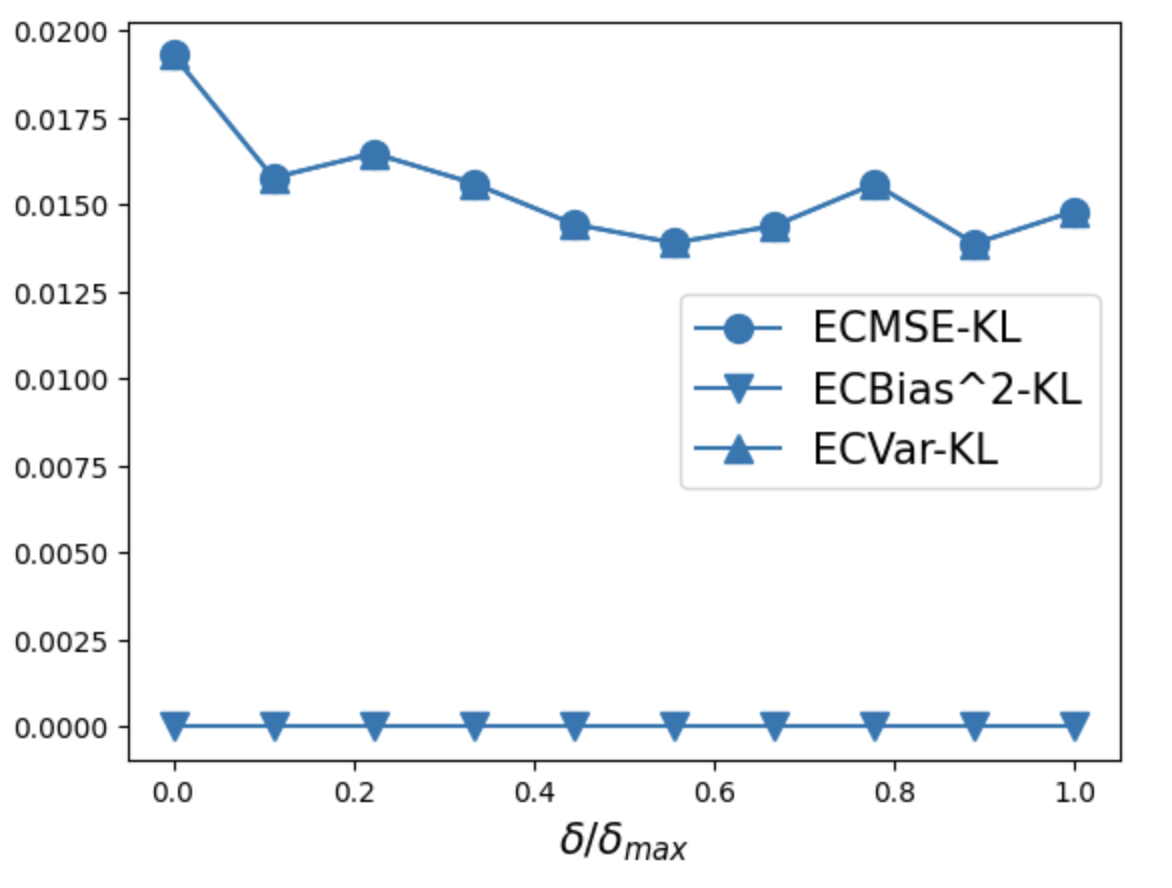}} 
		\subfigure[Missing by value]{\label{fig:real-nonrandom-17-kl} 
		\includegraphics[width=0.4\columnwidth]{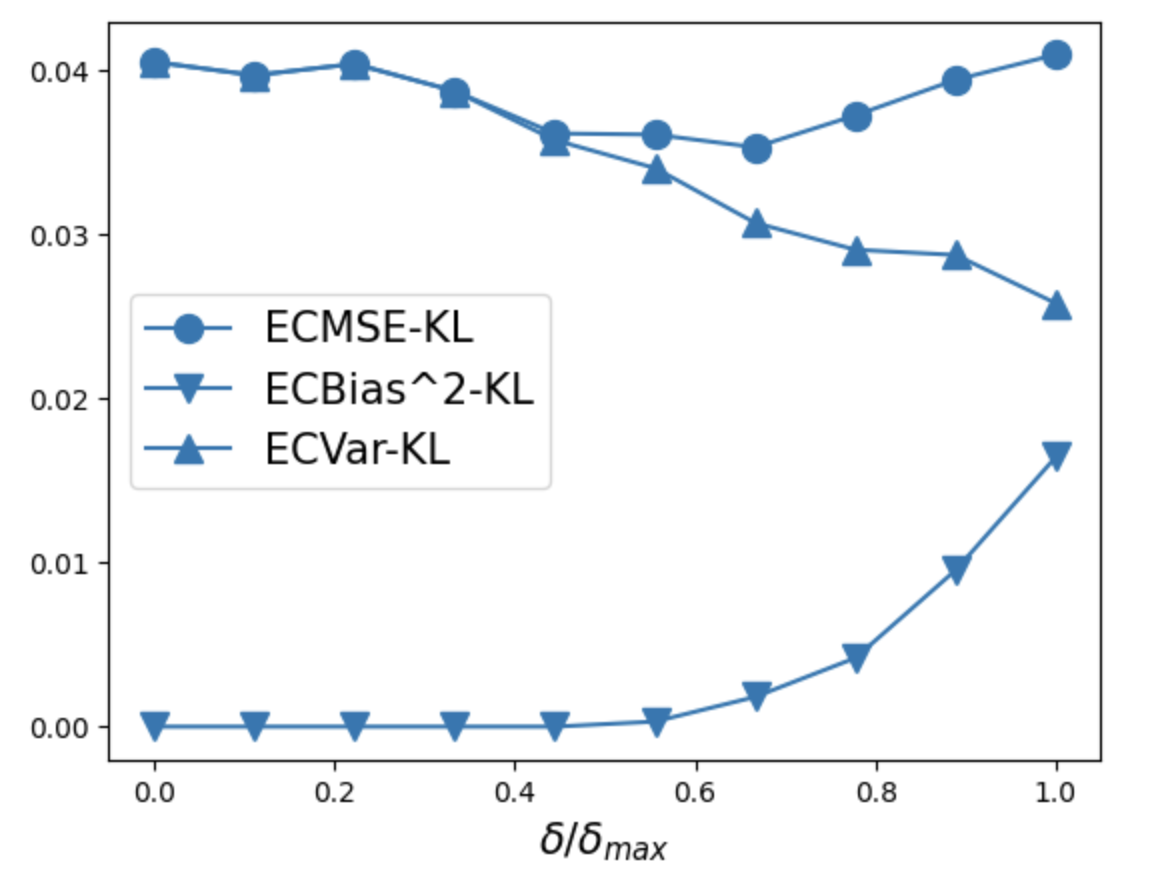}} 
		\bnotefig{These figures show the look-ahead-bias measured by the regret ECBias$^2$, the variance ECVar and the expected conditional mean squared error ECMSE for industry portfolios with the forward KL mechanism for different weights $\delta$, which control the trade-off between look-ahead-bias and variance. The samples are daily returns of 10 industry portfolios from 01/03/2017 to 01/02/2019.}
    \label{fig:real17-kl}
\end{figure}

\begin{itemize}
   \item \textit{Missing completely at random} -- we mask each entry in the training period as missing independently at random with a fixed probability $40\%$. The data is fixed as the first $T^{train}+T^{test}+T^{oos\text{-}test} = 400$ dates from the original data set. We repeat $50$ independent simulations, wherein for each sample we first generate the missing masks and then multiply-impute $100$ times. 
   
   \item \textit{Missing at random} -- for each industry portfolio $i$, we sample a Bernoulli variable $S_i$ with success probability $50\%$. If $S_i=1$, we mask the entries corresponding to $i$ as missing independently at random with probability $20\%$, otherwise we mask the entries corresponding to $i$ as missing independently at random with probability $50\%$. The rest of the procedure is the same as in the previous case.
   
    \item \textit{Block missing} -- we mask the first $40\%$ of the training period as missing. We conduct $50$ independent simulations. In the $k$-th sample, we use the $T^{train}+T^{test}+T^{oos\text{-}test}=400$ days of data starting from the $k$-th date in the original data set, and we multiply-impute $100$ times. 
    
    \item \textit{Missing by value} -- we mask an entry in the training period as missing if its absolute value is greater than $0.03$. The rest of the procedure is the same as in the previous case.
\end{itemize}

Figure \ref{fig:reals} compares the ECMSE as a function of $\delta$ for the three aggregation mechanisms (forward KL mechanism, full Wasserstein mechanism, and restricted Wasserstein mechanism) for individual stock returns. First, we observe that in most cases a non-trivial $\delta$ is optimal, that is, neither avoiding any look-ahead-bias ($\delta=0$) nor using all the data with full look-ahead-bias ($\delta=\delta_{max}$) minimizes the ECMSE. Second, the optimal value of $\delta$ depends on the missing pattern. Third, while the three mechanisms provide somewhat similar results and trade-offs, the more flexible forward KL and full Wasserstein mechanism dominate the restricted Wasserstein mechanism. Figure \ref{fig:reals-kl} provides further details for the trade-off between ECBsias$^2$ and ECVar for the forward KL mechanism. The results for the Wasserstein mechanism are similar and collected in Figures \ref{fig:reals-wbfull} and \ref{fig:reals-wb} in the Appendix. The monotonic dependency on $\delta$ is as expected. In most cases, the variance is only marginally decreased by using more data, while the increase in the bias seems to drive the trade-off.

Figures \ref{fig:real17}, \ref{fig:real17-kl},  \ref{fig:real17-wbfull} and  \ref{fig:real17-wb} show the corresponding results for the industry portfolios. While the general findings are the same as for the individual stocks, they highlight the important dependency on the missing pattern. In the case of missing at random a small value of $\delta$ is optimal as the look-ahead-bias is increasing much faster than the variance. However, for more complex missing patterns, where the returns are missing in blocks or for large realizations, the decrease in variance becomes more dominant and hence leads to a larger optimal value of $\delta$.

% \viet{Possibly we should write MCAR for Missing completely at random throughout the paper?...Actually, using the full name doesn't seem to hurt }

% \viet{Why are there no results for MAR?Will add}

%%%%%%%%%%%%%%%%%%%
\section{Conclusion}
\label{sec:conclusion}

In this paper we study the implications of imputation in a downstream out-of-sample task with time-series input data. We propose a Bayesian framework to introduce and precisely formulate the notion of look-ahead-bias involved in imputation procedures. We provide a solution for the optimal trade-off between look-ahead-bias and variance by a consensus posterior that fuses a number of individual posteriors corresponding to different layers of future information, and we use the (forward and backward) KL divergence and Wasserstein distance to derive tractable reformulations for finding the optimal consensus posterior. Our theoretical results focus on the fundamental case of estimating mean returns under simplifying assumptions in order to clearly illustrate the novel conceptual ideas. However, the stylized features we adopt have a solid motivation in decision-making with time-series data and can be easily relaxed. For example, an extension of our analysis to residuals following an AR(p) process is straightforward (\cite{ref:little2002statistical}). 

Last but not least, we want to emphasize that the impact of look-ahead-bias is important in financial, healthcare, and operations management applications and even relatively small improvements are valuable because their impact can be largely magnified in the downstream task. For example, an investment strategy that applies leverage or derivatives can substantially increase the numbers reported in this paper. We are the first to identify and study the problem of look-ahead-bias in the context of data imputation. As the use of increasingly large data sets is gaining attention in prescriptive analytics, and as most of these data sets have missing values, it is crucial to understand the implications of imputing missing values. We are confident that we have ``opened a floodgate" and that more follow-up papers will focus on the issue of look-ahead-bias quantification, which we believe is interesting, relevant, and largely overlooked.

\singlespacing
\bibliographystyle{econometrica}
{\small
\bibliography{bibliography}
}
\onehalfspacing

%\clearpage
%\newpage

\appendix

 \setcounter{equation}{0}
 \renewcommand{\theequation}{\thesection.\arabic{equation}}

\renewcommand{\theequation}{A.\arabic{equation}}%
\renewcommand{\thefigure}{A.\arabic{figure}} \setcounter{figure}{0}
\renewcommand{\thetable}{A.\Roman{table}} \setcounter{table}{0}

\section{Proof of Results}
\label{sec:appproof}
\begin{proof}[Proof of Lemma~\ref{lemma:pdcov}]
    Let $\xi$ by any non-zero vector of dimension $n$. From Assumption~\ref{a:1cell}, there exists $t\in[T_k]$, such that $\xi_{X_t}=\mc P_{M_t}^\perp(\xi)$ is non-zero. Since $\Omega_{X_t}$ is positive definite, we have $\xi_{X_t}^\top\Omega_{X_t}^{-1}\xi_{X_t}>0$. Therefore $$\xi^\top\left((\mc P_{M_t}^\perp)^{-1}(\Omega_{X_t}^{-1})\right)\xi>0.$$ 
    Thus 
    $$\xi^\top\left(\sum_{t \in [T_k]} (\mc P_{M_t}^\perp)^{-1}(\Omega_{X_t}^{-1})\right)\xi>0.$$ Since $\xi$ is arbitrary, we have $\sum_{t \in [T_k]} (\mc P_{M_t}^\perp)^{-1}(\Omega_{X_t}^{-1})$ is positive definite. Thus $\Sigma_k$ is positive definite. 
\end{proof}

\begin{proof}[Proof of Proposition~\ref{prop:bias-F}]
    From the expression of the aggregated mean $\msa$ in Proposition~\ref{prop:consensus-F}, we have
\begin{align*}
    \msa   &=  \big(\sum_{k \in [K]}\lambda_k\cov_k^{-1}\big)^{-1} \big(\sum_{k \in [K]}\lambda_k\cov_k^{-1}\m_k\big) \\
    &=  \big(\sum_{k \in [K]}\lambda_k V  D_k^{-1} V^\top\big)^{-1} \big(\sum_{k \in [K]}\lambda_k V  D_k^{-1} V^\top\m_k\big)  \\
    &=  V\big(\sum_{k \in [K]}\lambda_k  D_k^{-1} \big)^{-1} V^\top V \big(\sum_{k \in [K]}\lambda_k D_k^{-1} V^\top\m_k\big)  \\
    &=  V\big(\sum_{k \in [K]}\lambda_k D_k^{-1} \big)^{-1} \big(\sum_{k \in [K]}\lambda_k D_k^{-1} V^\top\m_k\big),
\end{align*}
where the second equality follows from the eigendecomposition of $\cov_k$ and the last equality follows because $V$ is an orthogonal matrix and thus $V^\top V = I_n$. By exploiting the definition of the set $\mc Z$, we have
\[
    \ell_{\mc Z}(\EE_{\mc C_\lambda^F(\{ \pi_k\})}[\theta]- \m_1) = \ell_{\mc Z}(\msa - \m_1) = \Sup{z \in \mc Z} z^\top (\msa - \m_1) = \Sup{z: \| V^\top z\|_1 \le 1} z^\top VV^\top (\msa - \m_1) = \|V^\top (\msa - \m_1)\|_{\infty},
\]
where the third equality follows from the fact that $VV^\top = I_n$, and the last equality follows by the properties of the dual norm pairs between $\| \cdot \|_1$ and $\| \cdot \|_\infty$. 
Combine with the expression for $\msa$ derived previously, we have
\begin{align*}
     \ell_{\mc Z}(\msa - \m_1) &= \| \big(\sum_{k \in [K]}\lambda_k  D_k^{-1} \big)^{-1} \big(\sum_{k \in [K]}\lambda_k D_k^{-1} V^\top\m_k\big) - V^\top \m_1 \|_\infty\\
     &= \max_{j\in[n]}\left|\frac{\sum_{k \in [K]} c_{kj} \lambda_k}{\sum_{k \in [K]} d_{kj}^{-1} \lambda_k} - v_j^\top \m_1\right|.
\end{align*}
Therefore the following holds
\[
    \{ \lambda \in \Delta_K: \ell_{\mc Z}(\EE_{\mc C_\lambda^F(\{ \pi_k\})}[\theta]-\m_1) \le \delta\} = \left\{ \lambda \in \Delta_K: 
       \ds - \delta \le \frac{\sum_{k \in [K]} c_{kj} \lambda_k}{\sum_{k \in [K]} d_{kj}^{-1} \lambda_k} - v_j^\top \m_1 \le \delta \quad \forall j \in [n]
    \right\},
\]
% \viet{should it be expectation of $\theta$ above?}
which is a feasible set prescribed using one simplex constraint and a collection of fractional linear constraints. This completes the proof.
\end{proof}

\begin{proof}[Proof of Theorem~\ref{thm:refor}]
By exploiting Proposition~\ref{prop:consensus-F}, we have $\mc C_{\lambda}^F(\{\pi_k\})=\mc N(\msa,\covsa)$, therefore problem~\eqref{eq:prob} is equivalent to
    \be \label{eq:probeq}
    \begin{array}{cl}
        \min & \Tr{\covsa} \\
        \st & \lambda \in \Delta_K\\
            & \covsa = \big(\sum_{k \in [K]}\lambda_k\cov_k^{-1}\big)^{-1} \in \PD^n,~\msa =\covsa\big(\sum_{k \in [K]}\lambda_k\cov_k^{-1}\m_k\big) \in \R^n\\
            & \ell_{\mc Z}(\msa - \m_1)  \le \delta.
    \end{array}
    \ee

By exploiting the diagonalization of $\cov_k$, we have
\begin{align*}
    \Tr{\covsa} &= \Tr{\big(\sum_{k \in [K]}\lambda_k\cov_k^{-1}\big)^{-1}} = \Tr{\big(\sum_{k \in [K]}\lambda_k V  D_k^{-1} V^\top\big)^{-1}} \\
    &= \Tr{ V (\sum_{k \in [K]} \lambda_k  D_k^{-1} )^{-1} V^\top} = \Tr{ (\sum_{k \in [K]} \lambda_k  D_k^{-1} )^{-1}} \\
    &=\sum_{j\in[n]} \frac{1}{ \sum_{k \in [K]}  d_{kj}^{-1} \lambda_k}.
\end{align*}

By exploiting Proposition~\ref{prop:bias-F}, we have the reformulation of the feasible set 
\[
    \{ \lambda \in \Delta_K: \ell_{\mc Z}(\msa-\m_1) \le \delta\} = \left\{ \lambda \in \Delta_K: 
       \ds - \delta \le \frac{\sum_{k \in [K]} c_{kj} \lambda_k}{\sum_{k \in [K]} d_{kj}^{-1} \lambda_k} - v_j^\top \m_1 \le \delta \quad \forall j \in [n]
    \right\}.
\]
Combine the above formulae, problem~\eqref{eq:probeq} becomes
\be \label{eq:refor-F-aux}
    \begin{array}{cl}
    \min & \ds\sum_{j\in[n]} \frac{1}{ \sum_{k \in [K]} d_{kj}^{-1} \lambda_k} \\
    \st & \lambda \in \Delta_K,\\
        & \ds - \delta \le \frac{\sum_{k \in [K]} c_{kj} \lambda_k}{\sum_{k \in [K]}  d_{kj}^{-1} \lambda_k} - v_j^\top \m_1 \le \delta \quad \forall j \in [n],
\end{array}
\ee
Because each element $d_{kj}$ is strictly positive and any feasible solution $\lambda$ should satisfy $\lambda \neq 0$, the fractional constraint of problem~\eqref{eq:refor-F-aux} can be rewritten as systems of linear inequality constraints
\[
    \forall j \in [n]: \qquad 
    \left\{
        \begin{array}{l}
            \sum_{k \in [K]} c_{kj} \lambda_k \le (\delta + v_j^\top \m_1) \sum_{k \in [K]}  d_{kj}^{-1} \lambda_k \\
            \sum_{k \in [K]} c_{kj} \lambda_k \ge (v_j^\top \m_1 - \delta) \sum_{k \in [K]}  d_{kj}^{-1} \lambda_k.
        \end{array}
    \right.
\]
Moreover, each fractional term in the objective function of~\eqref{eq:refor-F-aux} can be reformulated as a second-order cone following Section~2.3 in~\cite{ref:lobo1998socp}. This completes the proof.
\end{proof}

\begin{proof}[Proof of Lemma~\ref{lemma:projection-V}]
    
    By exploiting the definition of $\mc N_V$, we can parametrize the optimal solution as  $\pi_k' \sim \mc N(\m_k', \cov_k')$ where $\cov_k' = V  D_k V^\top$ and $D_k = \diag(d_k)$ is a diagonal matrix.
    
    First, consider Assertion~\ref{lemma:projection-V1}. By utilizing the expression of the KL divergence between Gaussian distributions in Definition~\ref{def:KL}, we have
    \[
        \begin{array}{rcl}
        (\m_k', \cov_k',  d_k) = \arg& \min & (\msa_k - \m_k)^\top \covsa_k^{-1} (\msa_k - \m_k) + \log \det \covsa_k + \Tr{\covsa_k^{-1} \cov_k} \\
        & \st & \msa_k \in \R^n,~\covsa_k \in \PD^n,~\wh d_k \in \R_{++}^n,~\covsa_k =V \diag(\wh d_k) V^\top.
        \end{array}
    \]
    For any $\covsa_k \in \PD^n$, the optimal solution in the variable $\msa_k$ is $\m_k' = \m_k$. The optimal solution in the variable $\wh d_k$ satisfies
    \[
         d_k = \arg\Min{\wh d_k \in \R_{++}^n}~\log\det \diag(\wh d_k) + \Tr{ \diag(\wh d_k^{-1}) V^\top \cov_k V},
    \]
    which admits the optimal solution
    $(d_{k})_j = v_j^\top \cov_k v_j$, for all $j \in [n]$.
    
    Next, consider Assertion~\ref{lemma:projection-V2}.  Utilizing the expression of the KL divergence between Gaussian distributions in Definition~\ref{def:KL}, we have
    \[
        \begin{array}{rcl}
        (\m_k', \cov_k',  d_k) = \arg& \min & (\msa_k - \m_k)^\top (\cov_k)^{-1} (\msa_k - \m_k) - \log \det \covsa_k + \Tr{\covsa_k (\cov_k)^{-1}} \\
        & \st & \msa_k \in \R^n,~\covsa_k \in \PD^n,~\wh d_k \in \R_{++}^n,~\covsa_k = V \diag(\wh d_k) V^\top.
        \end{array}
    \]
    Because $\cov_k \in \PD^p$, the optimal solution in the variable $\msa_k$ is $\m_k' = \m_k$. The optimal solution in the variable $\wh d_k$ satisfies
    \[
         d_k = \arg \Min{\wh d_k \in \R_{++}^n} ~-\log\det \diag(\wh d_k) + \Tr{\diag(\wh d_k) V^\top (\cov_k)^{-1} V},
    \]
    which admits the optimal solution $( d_{k})_j = [V^\top (\cov_k)^{-1} V]_{jj}$ for all $ j \in [n]$. This completes the proof.
\end{proof}

\begin{proof}[Proof of Proposition~\ref{prop:bias-W}]
 Note that by Proposition~\ref{prop:consensus-W}, the first moment of the solution $\pi\opt$ to~\eqref{eq:consensus-W} satisfies
$\mathbb{E}_{\pi\opt}[\theta] = \msa =\sum_{k \in [K]}\lambda_k\m_k$.  
Now $(i)$ and $(iii)$ follow from definition. For $(ii)$, note that
\[
        \Lambda_{\ell_{\mc Z}, \delta}^W( \{\pi_k\}) = \left\{ \lambda \in \Delta_K: 
      \sup_{Az\leq b} z^\top\left(\sum_{k \in [K]}\lambda_k\m_k-\m_1\right)\leq \delta
    \right\}.
    \] 
    The dual linear programming of $\sup_{Az\leq b} z^\top\left(\sum_{k \in [K]}\lambda_k\m_k-\m_1\right)$ is
      \[
        \begin{array}{cl}
        \inf & b^\top w\\
        \st & A^\top w = \sum_{k\in[K]}\lambda_k\mu_k-\mu_1,~ w \geq 0,
        \end{array}
        \]
    and by strong duality, the primal and dual optimal objective values are equal. Therefore we can reformulate $\sup_{Az\leq b} z^\top\left(\sum_{k \in [K]}\lambda_k\m_k-\m_1\right)\leq \delta$ as the existence of $w\geq0$ such that $b^\top w\leq\delta$ and $A^\top w = \sum_{k\in[K]}\lambda_k\mu_k-\mu_1$.
\end{proof}

\begin{proof}[Proof of Theorem~\ref{thm:consensus-W}]
    Because the covariance matrices $\cov_k$ are simultaneously diagonalizable, Remark~4.4 in~\cite{ref:alvarez2016fixed} implies that the covariance matrix $\covsa$ is known in closed form
    \[
        \covsa = \big( \sum_{k \in [K]} \lambda_k \cov_k^\half \big)^2.
    \]
    Since the trace operator is a linear operator, we find
    \[
        \Tr{\covsa} = \sum_{i, j \in [K]} \lambda_i \lambda_j \Tr{\cov_i^\half \cov_j^\half}.
    \]
    This completes the proof.
\end{proof}

\begin{proof}[Proof of Proposition~\ref{prop:reformKis2}]
It is straightforward to compute
\begin{align*}
 \mathrm{Tr}(\mathrm{Cov}_{\mc C_{\lambda}(\{ \pi_k\})}[\theta]) & = \mathrm{Tr}(\widehat\Sigma)\\
& = \mathrm{Tr}\left((\lambda_1I_n+\lambda_2\Phi)\Sigma_1(\lambda_1I_n+\lambda_2\Phi)\right)\\
& = \lambda_1^2\mathrm{Tr}(\Sigma_1)+\lambda_1\lambda_2(\mathrm{Tr}(\Sigma_1\Phi)+\mathrm{Tr}(\Phi\Sigma_1))+\lambda_2^2\mathrm{Tr}(\Phi\Sigma_1\Phi)\\
& = \lambda_1^2\mathrm{Tr}(\Sigma_1)+2\lambda_1\lambda_2\mathrm{Tr}(\Sigma_1\Phi) + \lambda_2^2\mathrm{Tr}(\Sigma_2),
\end{align*}
which completes the proof.
\end{proof}

\begin{proof}[Proof of Theorem~\ref{thm:reversereform}]
By exploiting Proposition~\ref{prop:consensus-B}, the first and second moments of the posterior distribution $\pi\opt = \mc C_{\lambda}^B(\{\pi_k\})$ satisfy
\begin{subequations}
\begin{equation}\label{eq:backproofmean}
\mathbb{E}_{\pi\opt}[\theta] = \sum_{k \in [K]}\lambda_k\mathbb{E}_{\pi_k}[\theta]=\sum_{k \in [K]}\lambda_k\m_k,
\end{equation}
and
\[
\mathbb{E}_{\pi\opt}[\theta\theta^\top]=\sum_{k \in [K]}\lambda_k\mathbb{E}_{\pi_k}[\theta\theta^\top]=\sum_{k \in [K]}\lambda_k(\m_k\m_k^\top +\cov_k).
\]
Thus, the covariance matrix under $\pi\opt$ is
\begin{equation}\label{eq:backproofcov}
\mathrm{Cov}_{\pi\opt}[\theta] = \mathbb{E}_{\pi\opt}[\theta\theta^\top] - \mathbb{E}_{\pi\opt}[\theta]\mathbb{E}_{\pi\opt}[\theta]^\top = \sum_{k \in [K]}\lambda_k(\m_k\m_k^\top+\cov_k)-(\sum_{k \in [K]}\lambda_k\m_k)(\sum_{k \in [K]}\lambda_k\m_k)^\top.
\end{equation}
\end{subequations}
Recall that $\lambda\opt$ solves
\[\begin{array}{cl}
        \min & \mathrm{Tr}(\mathrm{Cov}_{\pi\opt}[\theta]) \\
        \st & \lambda \in \Delta_K, ~\ell_{\mc Z}(\EE_{\pi\opt}[\theta]- \m_1) \le \delta.
    \end{array}
\]
The reformulation~\eqref{eq:reversereforquad} follows by substituting~\eqref{eq:backproofmean} and~\eqref{eq:backproofcov} into the above optimization problem. 
\end{proof}

\section{Simulation Study}
\label{sec:appendix}

%\subsection{Simulation Study}

\begin{figure}[H]
\tcapfig{Look-ahead-bias and variance with full Wasserstein mechanism on simulated data}
    \centering
    \subfigure[Missing completely at random]{\label{fig:synth-mcar-wbfull}
		\includegraphics[width=0.45\columnwidth]{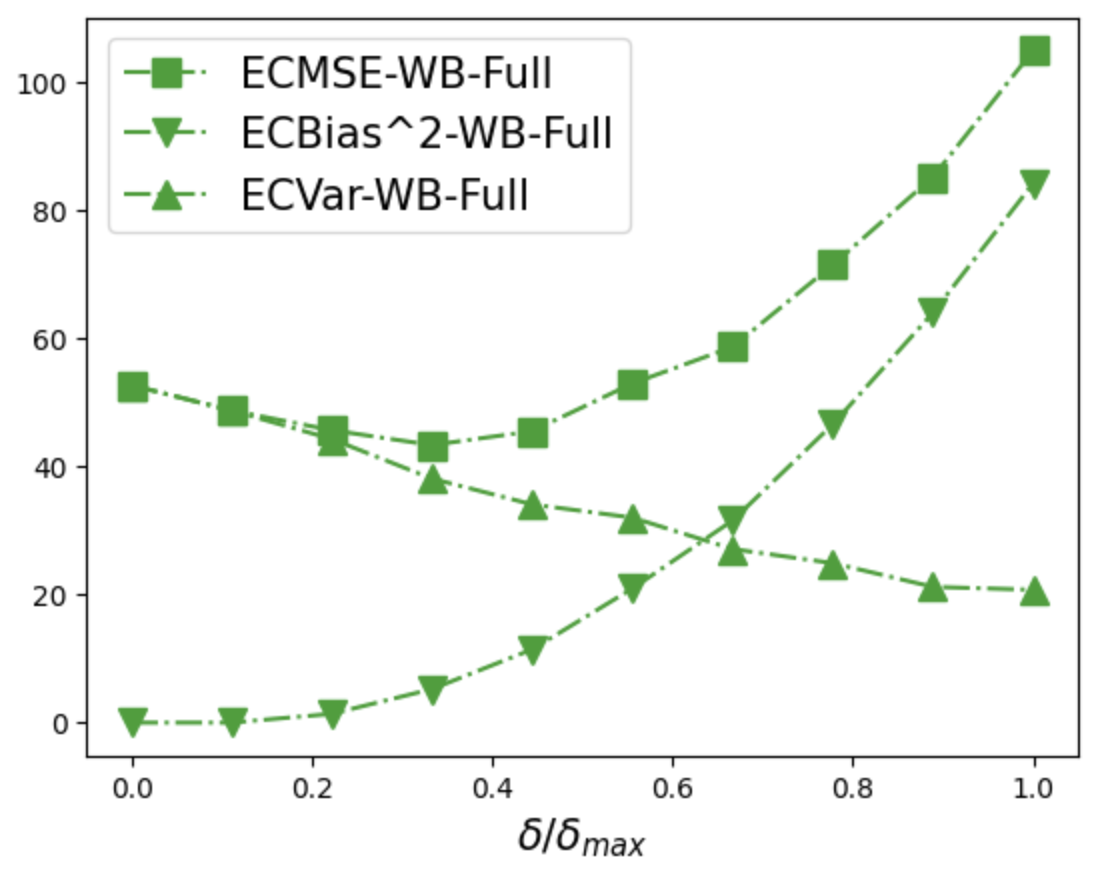}} \hspace{1mm}
	\subfigure[Missing at random]{\label{fig:synth-mar-wbfull}
	\includegraphics[width=0.45\columnwidth]{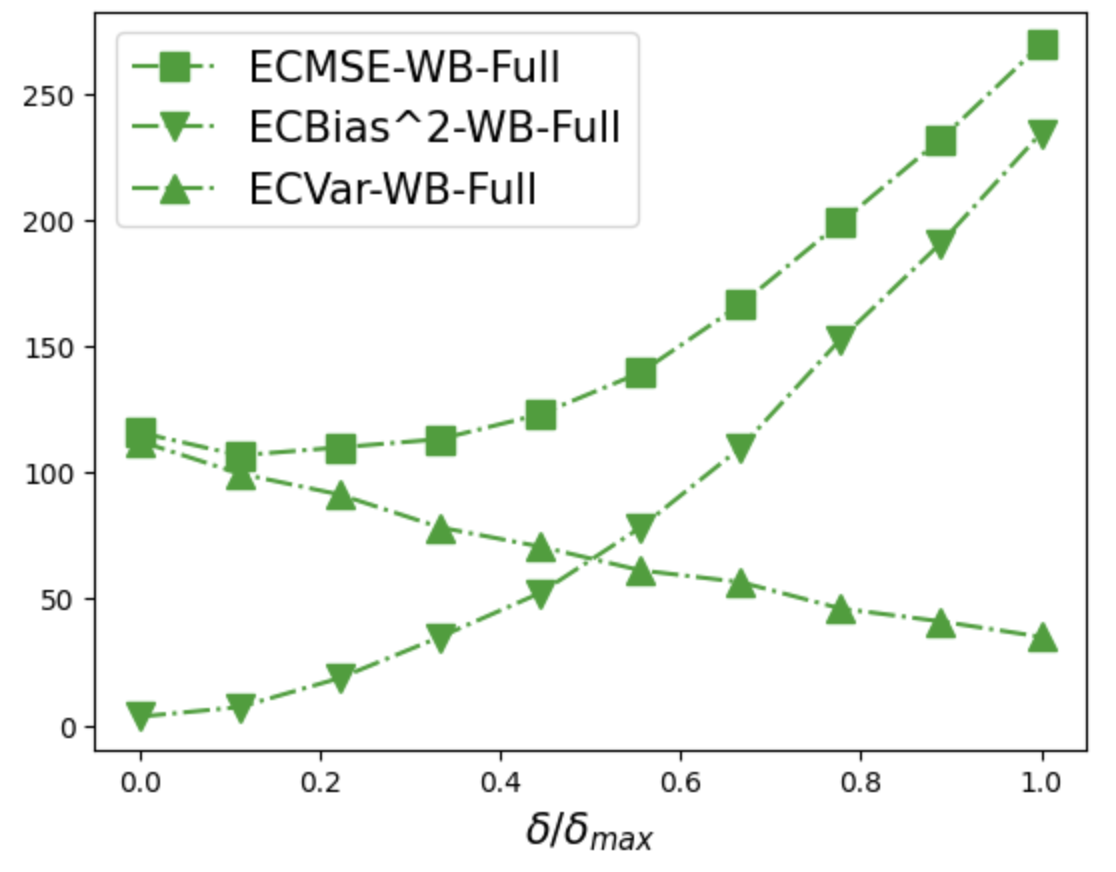}}
	    \hspace{1mm}
	\subfigure[Block missing]{\label{fig:synth-block-wbfull}
		\includegraphics[width=0.45\columnwidth]{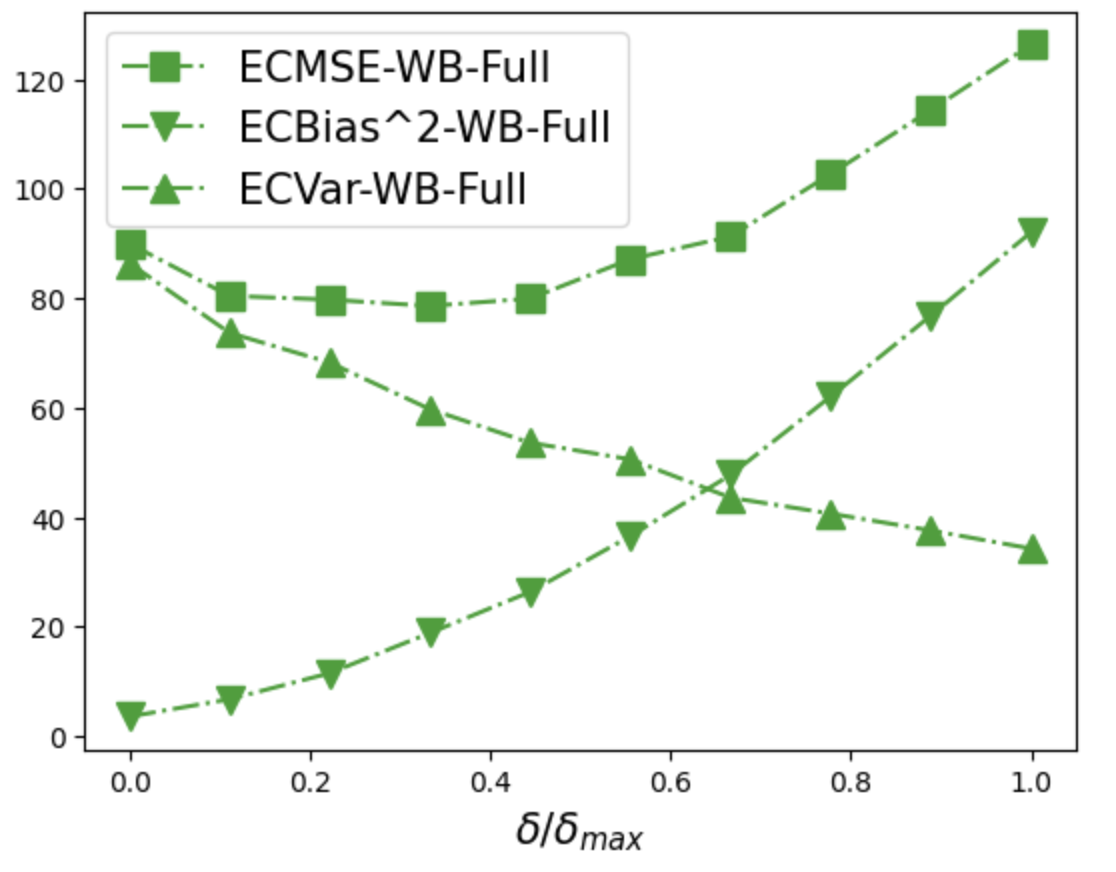}} 
	\subfigure[Missing by value]{\label{fig:synth-nonrandom-wbfull} 
		\includegraphics[width=0.45\columnwidth]{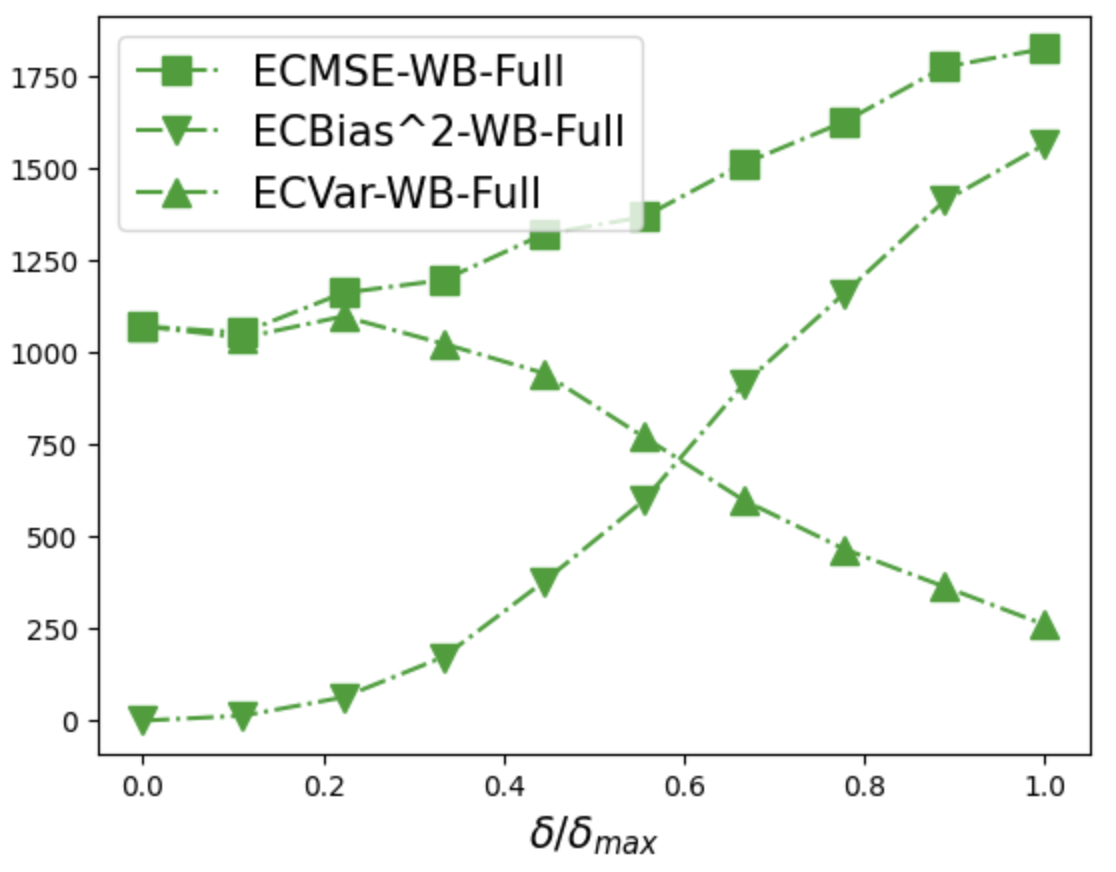}} 
			\bnotefig{These figures show the look-ahead-bias measured by the regret ECBias$^2$, the variance ECVar and the expected conditional mean squared error ECMSE for simulated data with the full Wasserstein mechanism for different weights $\delta$, which control the trade-off between look-ahead-bias and variance.}
    \label{fig:synth-wbfull}
\end{figure}

\begin{figure}[H]
\tcapfig{Look-ahead-bias and variance with restricted Wasserstein mechanism on simulated data}
    \centering
    \subfigure[Missing completely at random]{\label{fig:synth-mcar-wb}
		\includegraphics[width=0.45\columnwidth]{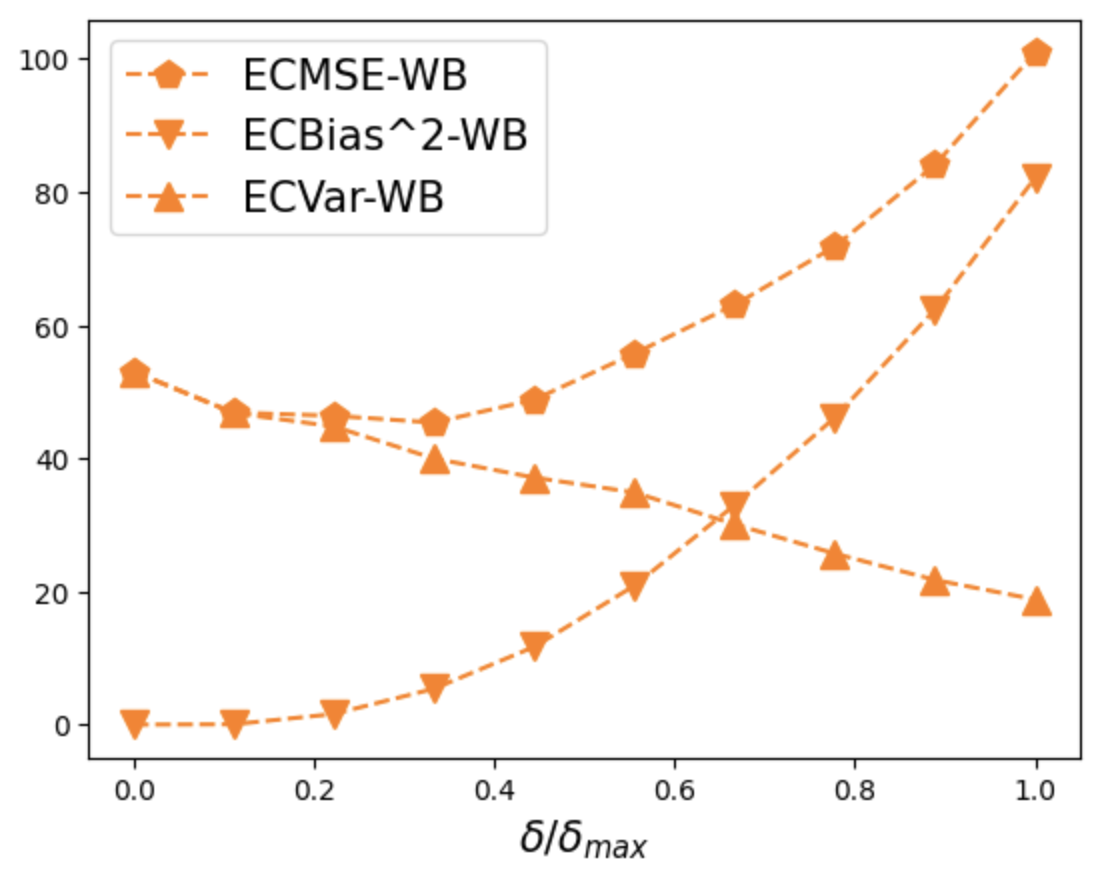}} \hspace{1mm}
	\subfigure[Missing at random]{\label{fig:synth-mar-wb}
	\includegraphics[width=0.45\columnwidth]{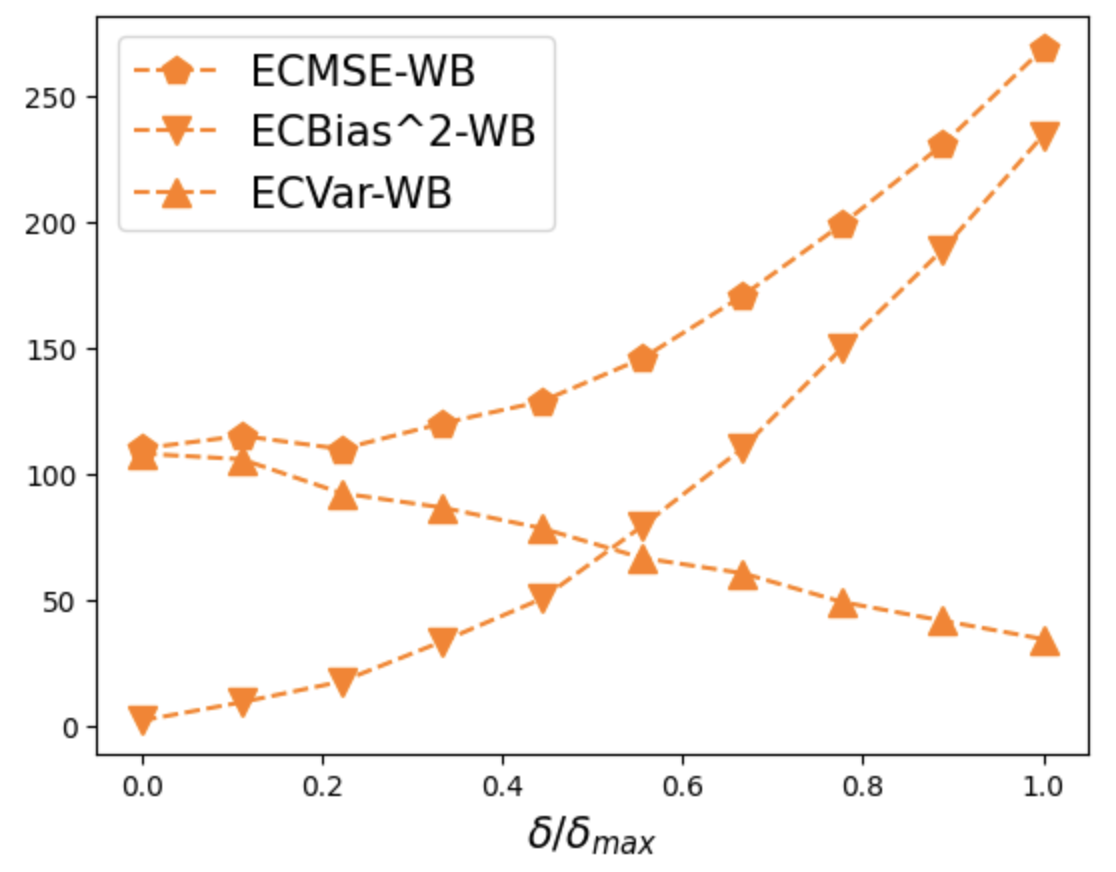}}
	    \hspace{1mm}
	\subfigure[Block missing]{\label{fig:synth-block-wb}
		\includegraphics[width=0.45\columnwidth]{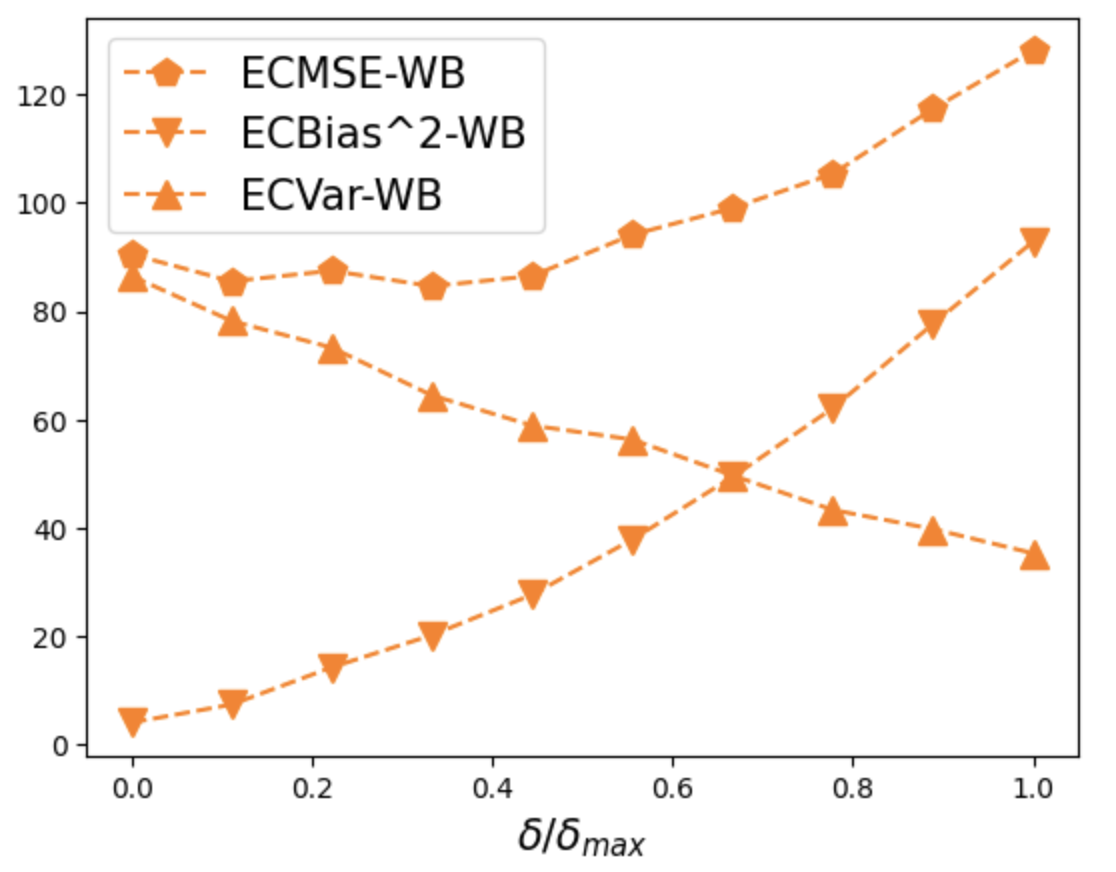}} 
	\subfigure[Missing by value]{\label{fig:synth-nonrandom-wb} 
		\includegraphics[width=0.45\columnwidth]{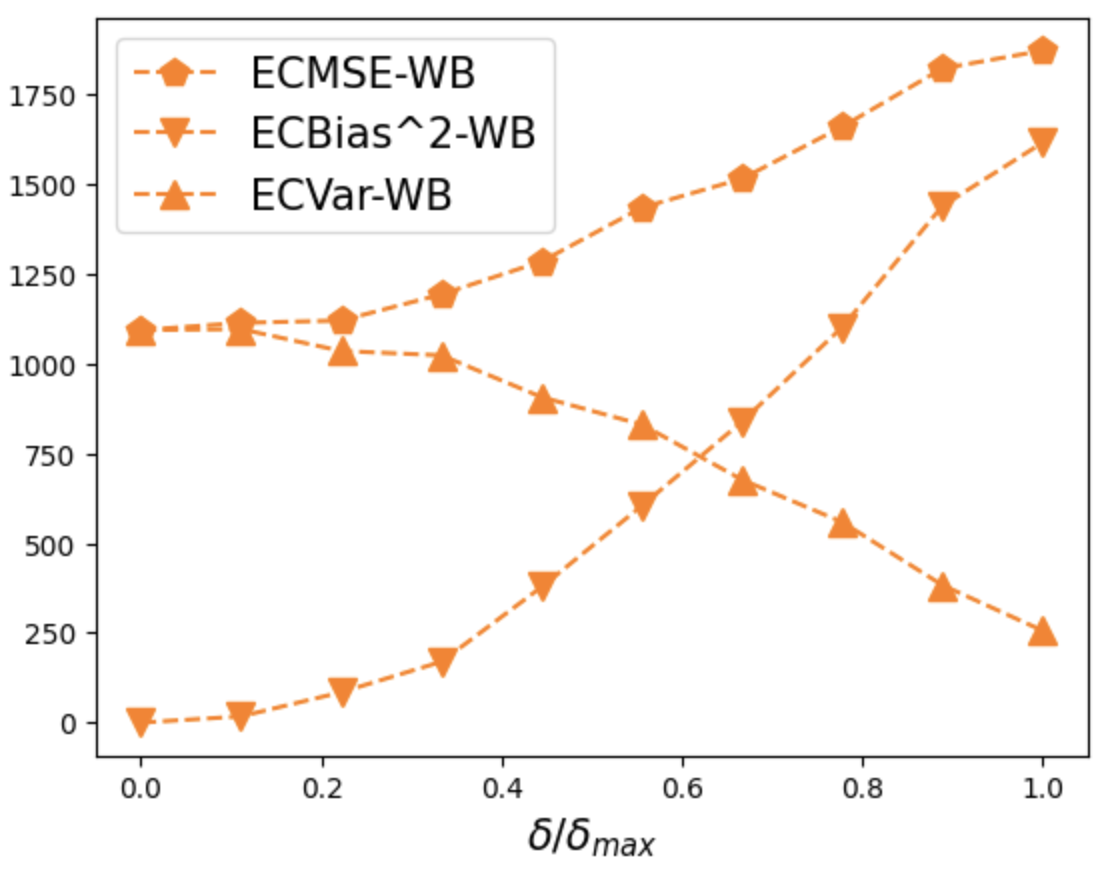}} 
					\bnotefig{These figures show the look-ahead-bias measured by the regret ECBias$^2$, the variance ECVar and the expected conditional mean squared error ECMSE for simulated data with the restricted Wasserstein mechanism for different weights $\delta$, which control the trade-off between look-ahead-bias and variance.}
    \label{fig:synth-wb}
\end{figure}

\section{Empirical Results}

\subsection{Imputation of Individual Stock Returns}
\label{sec:appnumerics-stock} 

%We report the trends of ECMSE, ECBsias$^2$ and ECVar with respect to the $\delta$ parameter for each of the three aggregation mechanisms (i.e.~forward KL, full Wasserstein and restricted Wasserstein) in Figures~\ref{fig:reals-kl},~\ref{fig:reals-wbfull} and~\ref{fig:reals-wb}.

\begin{figure}[H]
\tcapfig{Look-ahead-bias and variance with full Wasserstein mechanism for individual stocks}
    \centering
    \subfigure[Missing completely at random]{\label{fig:real-mcar-s-wbfull} 
		\includegraphics[width=0.45\columnwidth]{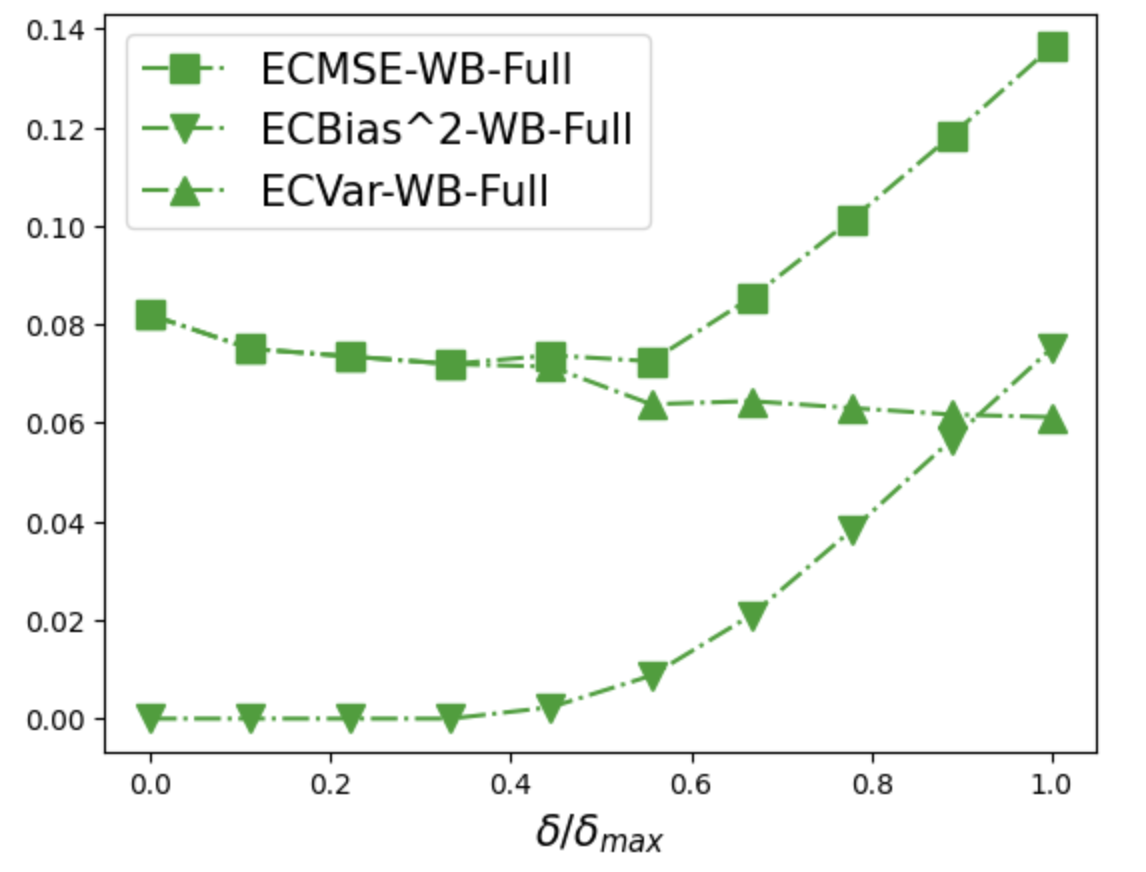}} \hspace{1mm}
	\subfigure[Missing at random]{\label{fig:real-mar-s-wbfull}
	\includegraphics[width=0.45\columnwidth]{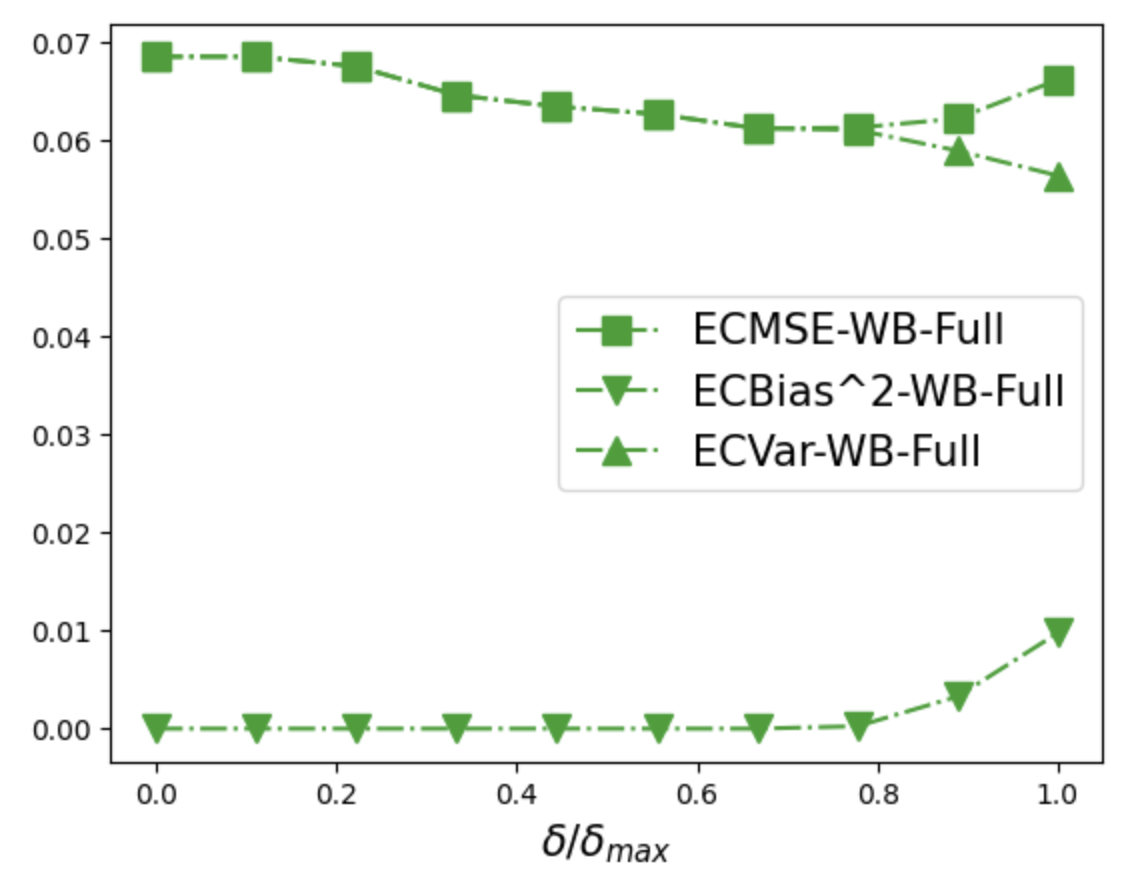}}
	    \hspace{1mm}
	\subfigure[Block missing]{\label{fig:real-block-s-wbfull} 
		\includegraphics[width=0.45\columnwidth]{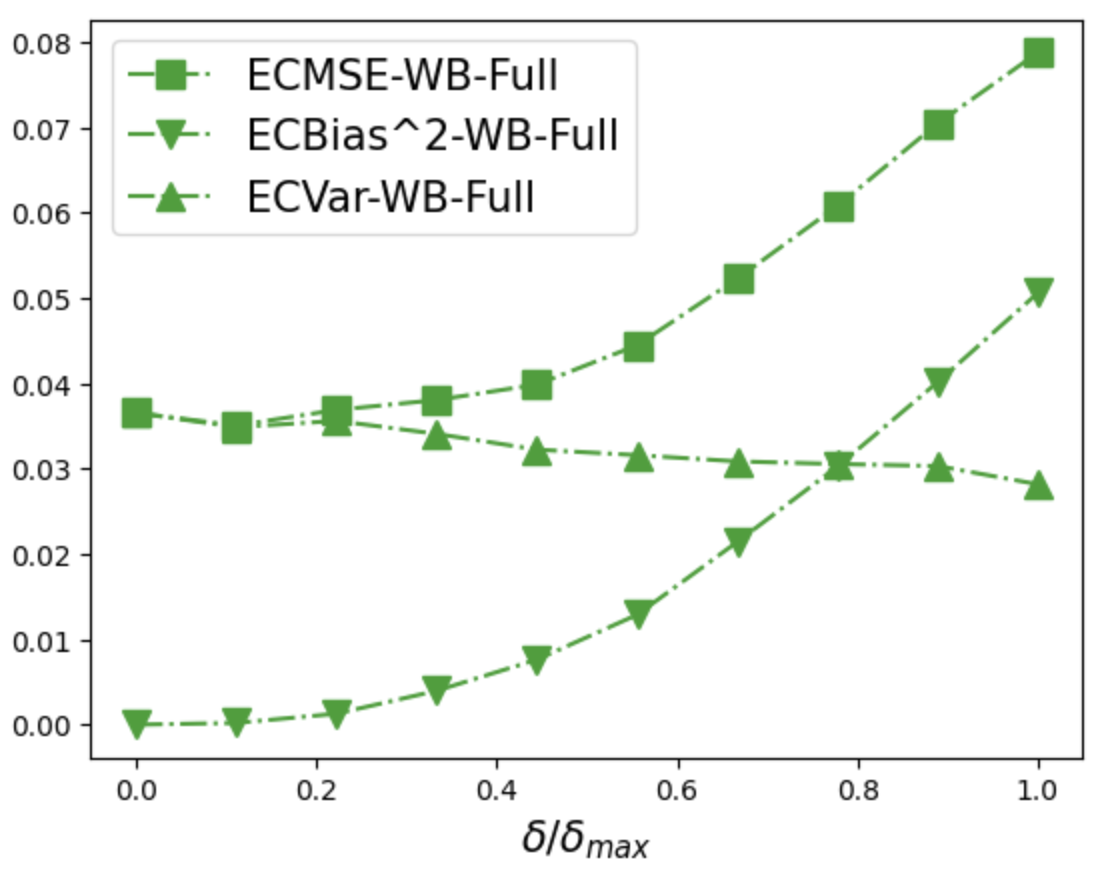}} 
		\subfigure[Missing by value]{\label{fig:real-nonrandom-s-wbfull} 
		\includegraphics[width=0.45\columnwidth]{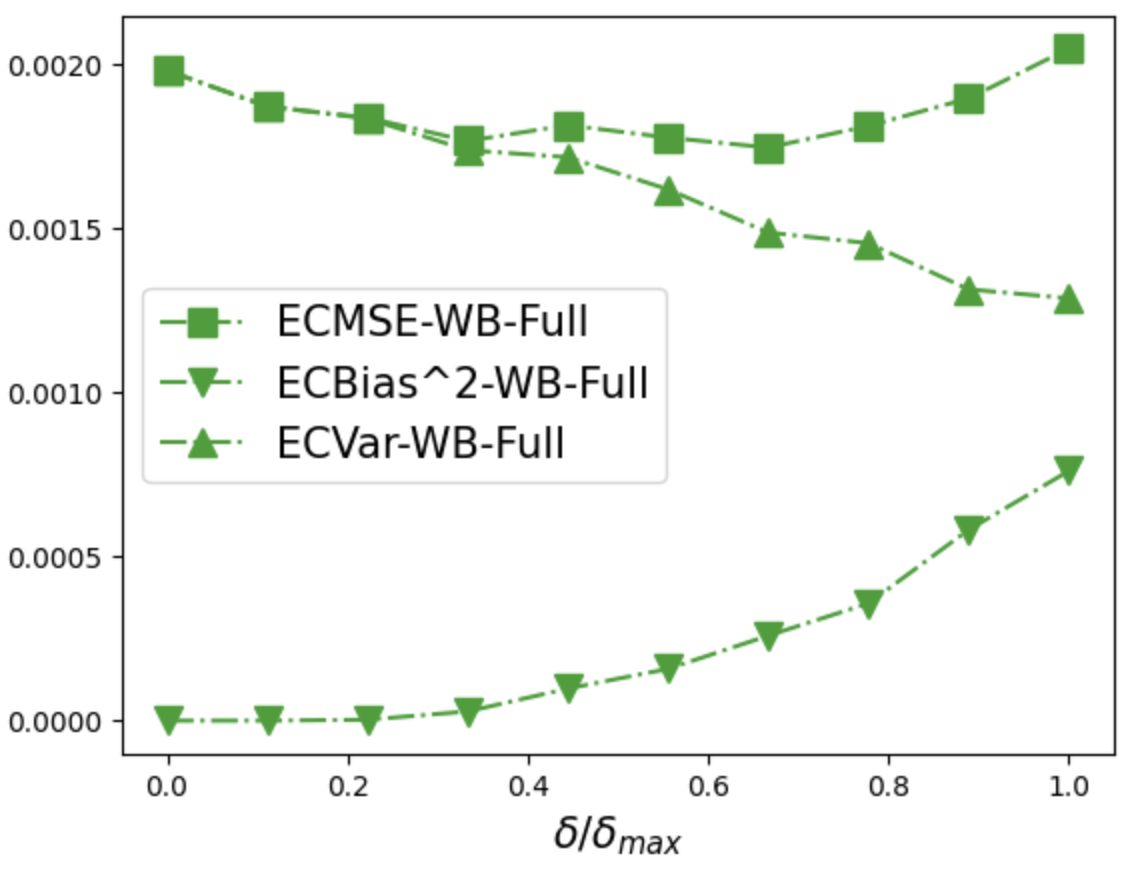}} 
			\bnotefig{These figures show the look-ahead-bias measured by the regret ECBias$^2$, the variance ECVar and the expected conditional mean squared error ECMSE for individual stocks with the full Wasserstein mechanism for different weights $\delta$, which control the trade-off between look-ahead-bias and variance. The samples are daily returns of 10 representative large capitalization stocks from  01/01/2015 to 01/01/2017.}
    \label{fig:reals-wbfull}
\end{figure}

\begin{figure}[H]
\tcapfig{Look-ahead-bias and variance with restricted Wasserstein mechanism for individual stocks}
    \centering
    \subfigure[Missing completely at random]{\label{fig:real-mcar-s-wb} 
		\includegraphics[width=0.45\columnwidth]{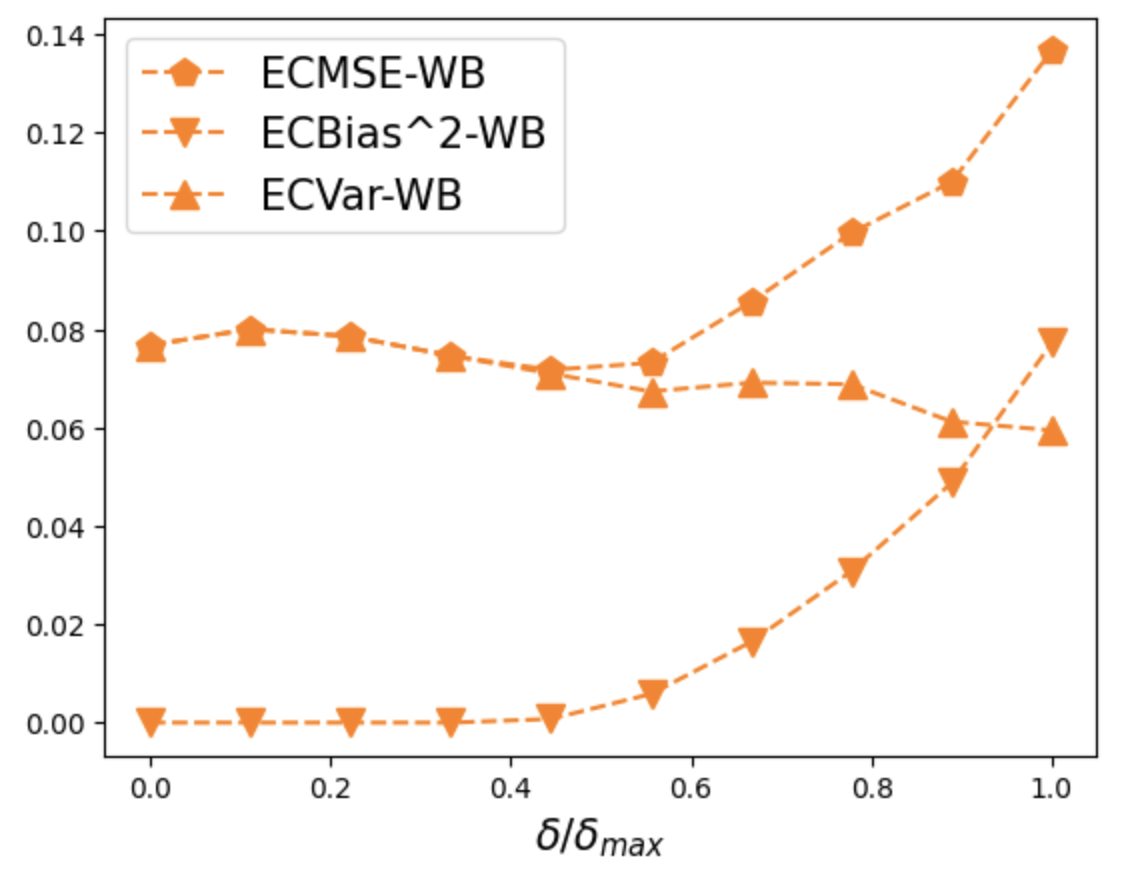}} \hspace{1mm}
	\subfigure[Missing at random]{\label{fig:real-mar-s-wb}
	\includegraphics[width=0.45\columnwidth]{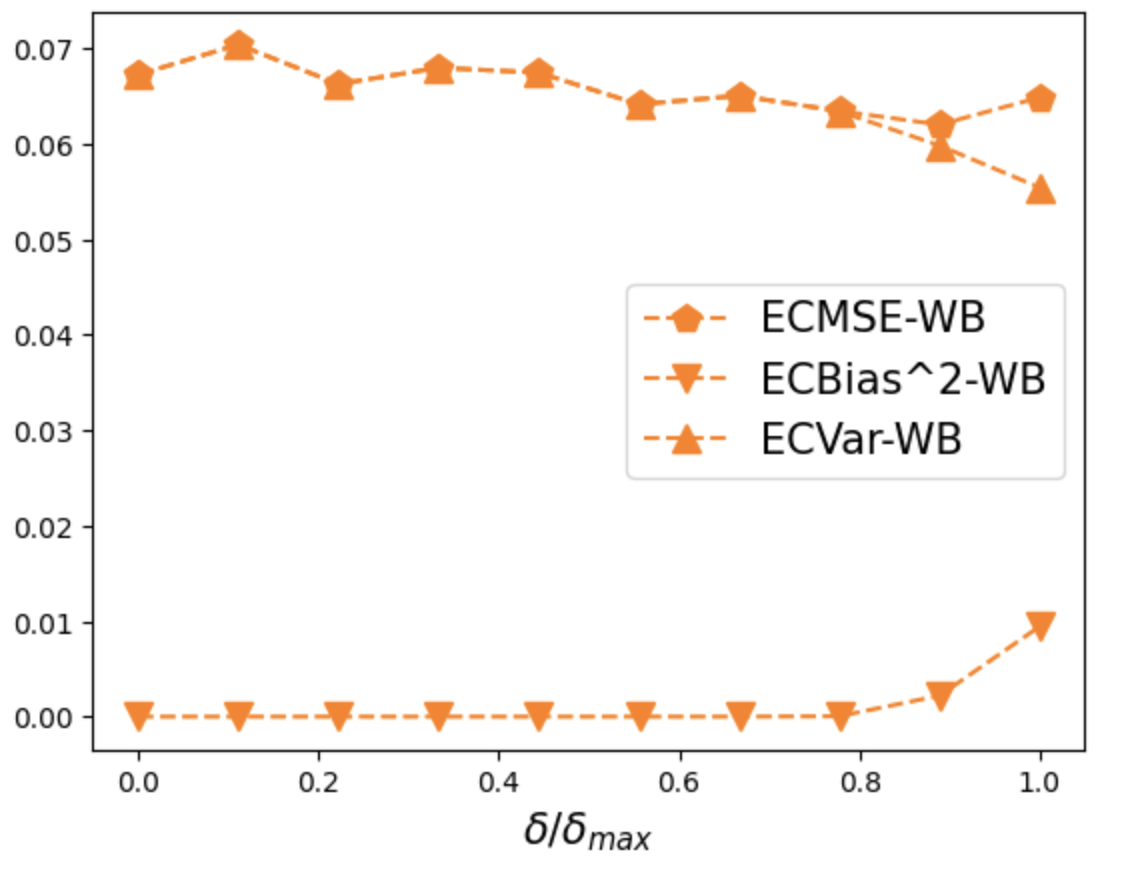}}
	    \hspace{1mm}
	\subfigure[Block missing]{\label{fig:real-block-s-wb} 
		\includegraphics[width=0.45\columnwidth]{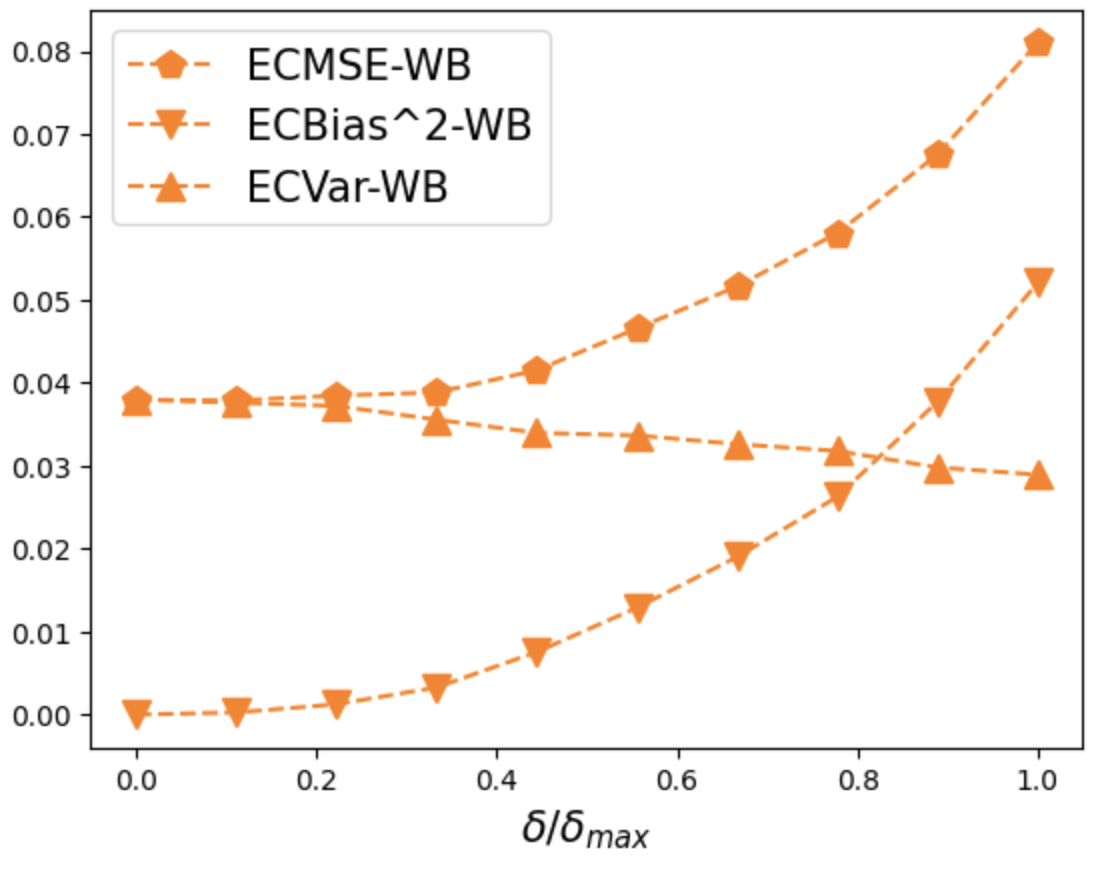}} 
		\subfigure[Missing by value]{\label{fig:real-nonrandom-s-wb} 
		\includegraphics[width=0.45\columnwidth]{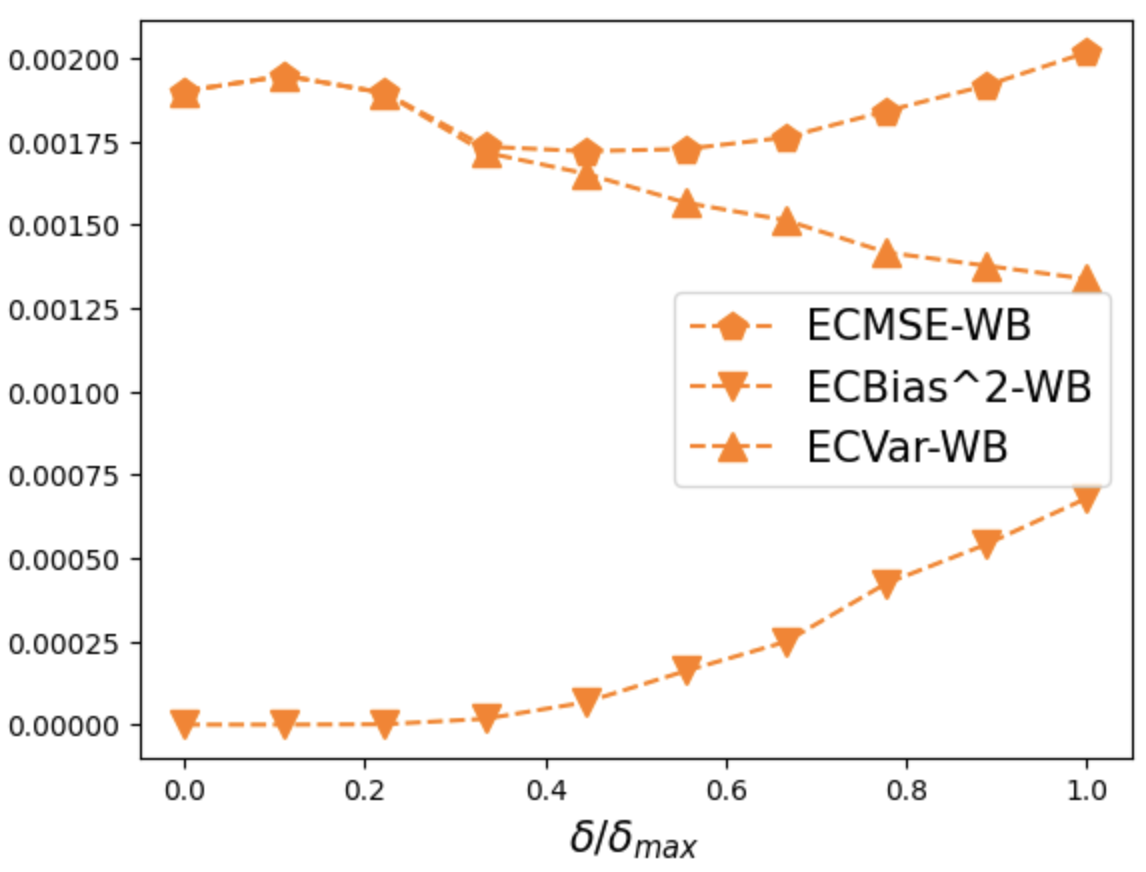}} 
\bnotefig{These figures show the look-ahead-bias measured by the regret ECBias$^2$, the variance ECVar and the expected conditional mean squared error ECMSE for individual stocks with the restricted Wasserstein mechanism for different weights $\delta$, which control the trade-off between look-ahead-bias and variance. The samples are daily returns of 10 representative large capitalization stocks from  01/01/2015 to 01/01/2017.}
    \label{fig:reals-wb}
\end{figure}

\subsection{Imputation of Industry Portfolio Returns}
\label{sec:appnumerics-industry}

\begin{figure}[H]
\tcapfig{Look-ahead-bias and variance with full Wasserstein mechanism for industry portfolios}
    \centering
    \subfigure[Missing completely at random]{\label{fig:real-mcar-17-wbfull} 
		\includegraphics[width=0.45\columnwidth]{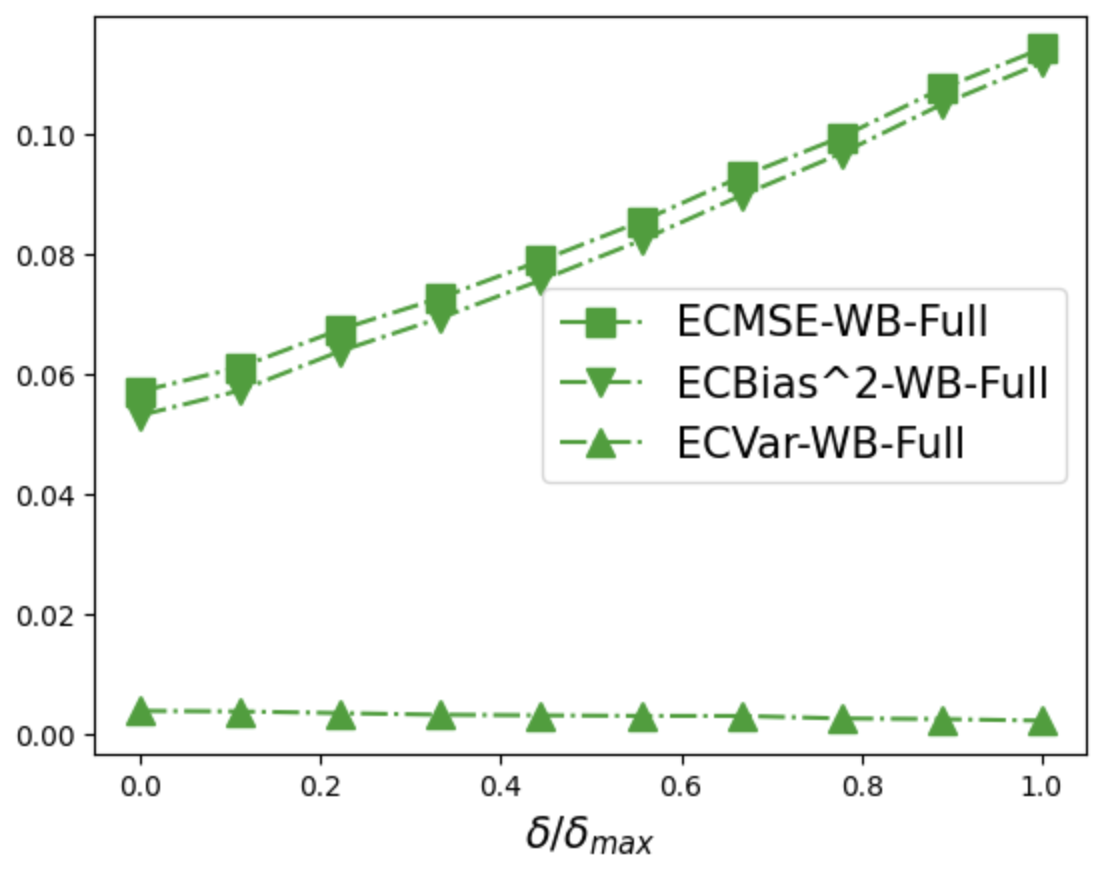}} \hspace{1mm}
	\subfigure[Missing at random]{\label{fig:real-mar-17-wbfull}
	\includegraphics[width=0.45\columnwidth]{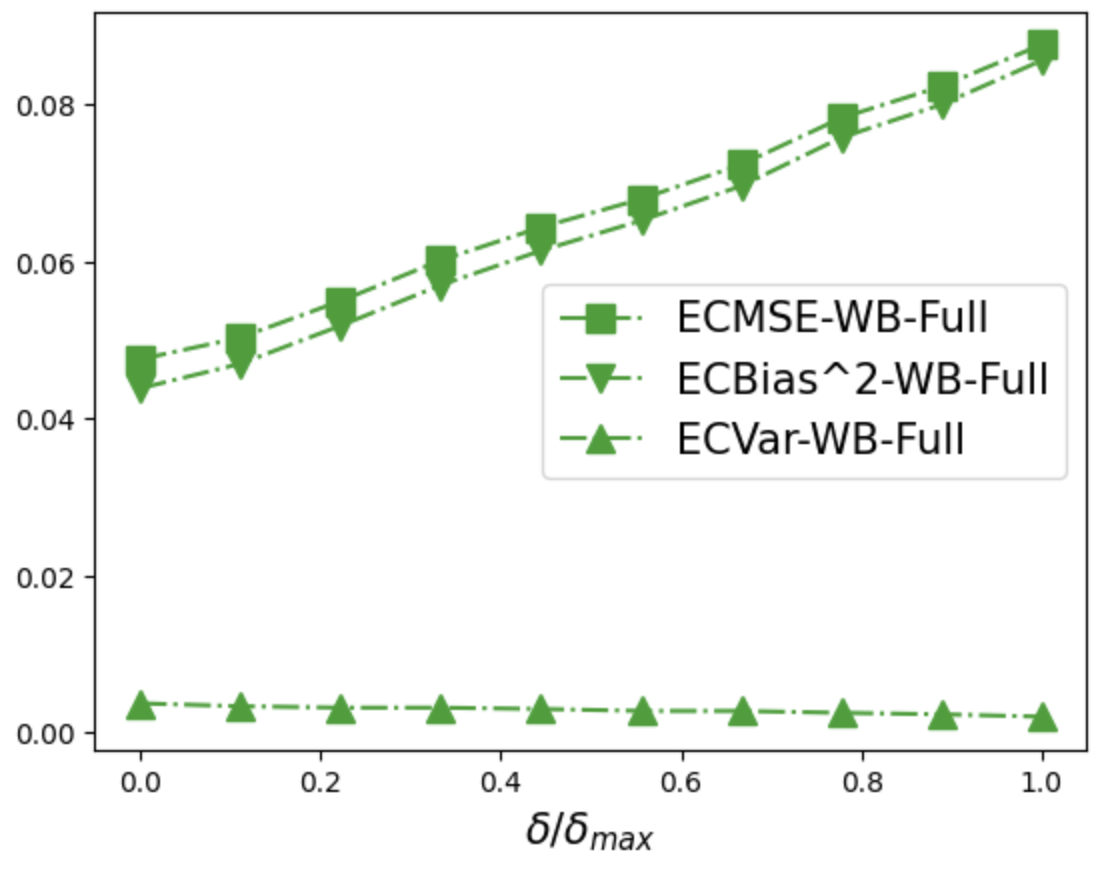}}
	    \hspace{1mm}
	\subfigure[Block missing]{\label{fig:real-block-17-wbfull} 
		\includegraphics[width=0.45\columnwidth]{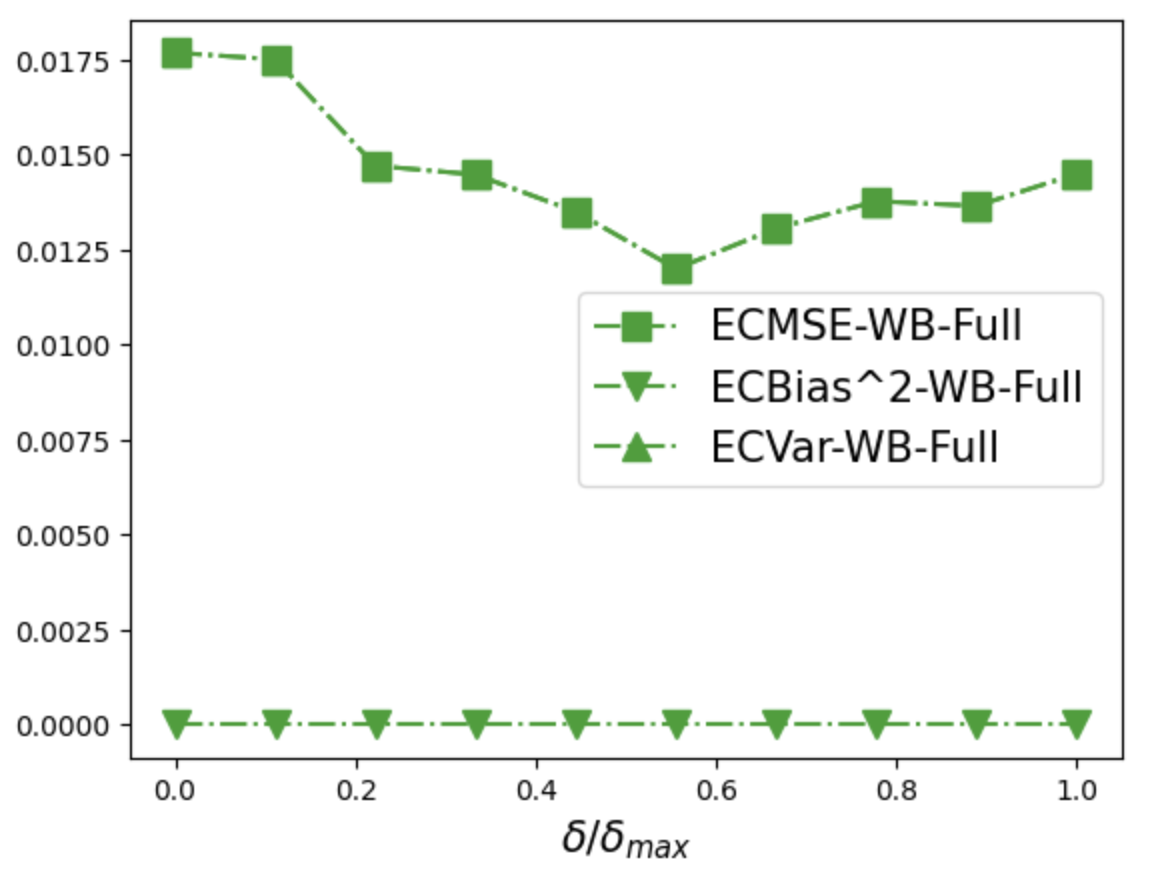}} 
		\subfigure[Missing by value]{\label{fig:real-nonrandom-17-wbfull} 
		\includegraphics[width=0.45\columnwidth]{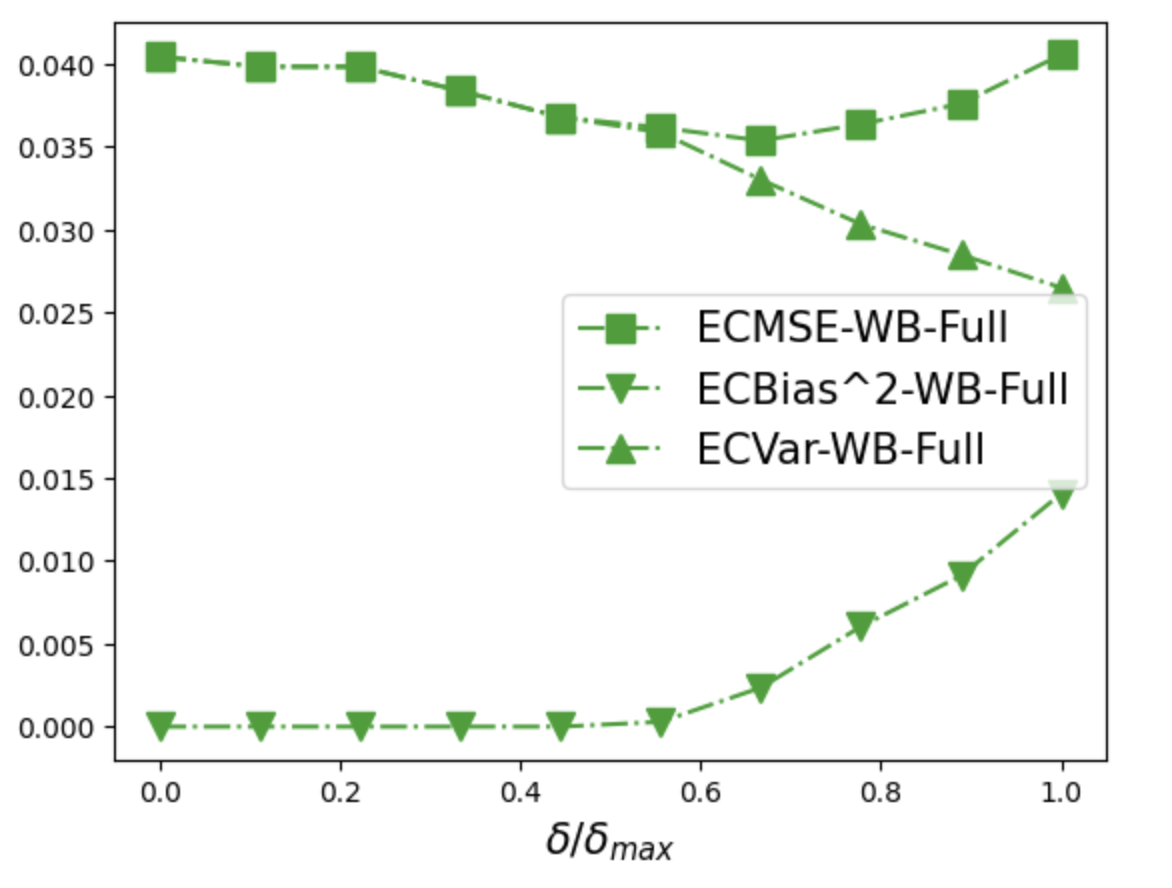}} 
		\bnotefig{These figures show the look-ahead-bias measured by the regret ECBias$^2$, the variance ECVar and the expected conditional mean squared error ECMSE for industry portfolios with the full Wasserstein mechanism for different weights $\delta$, which control the trade-off between look-ahead-bias and variance. The samples are daily returns of 10 industry portfolios from  01/03/2017 to 01/02/2019.}
    \label{fig:real17-wbfull}
\end{figure}

\begin{figure}[H]
\tcapfig{Look-ahead-bias and variance with restricted Wasserstein mechanism for industry portfolios}
    \centering
    \subfigure[Missing completely at random]{\label{fig:real-mcar-17-wb} 
		\includegraphics[width=0.45\columnwidth]{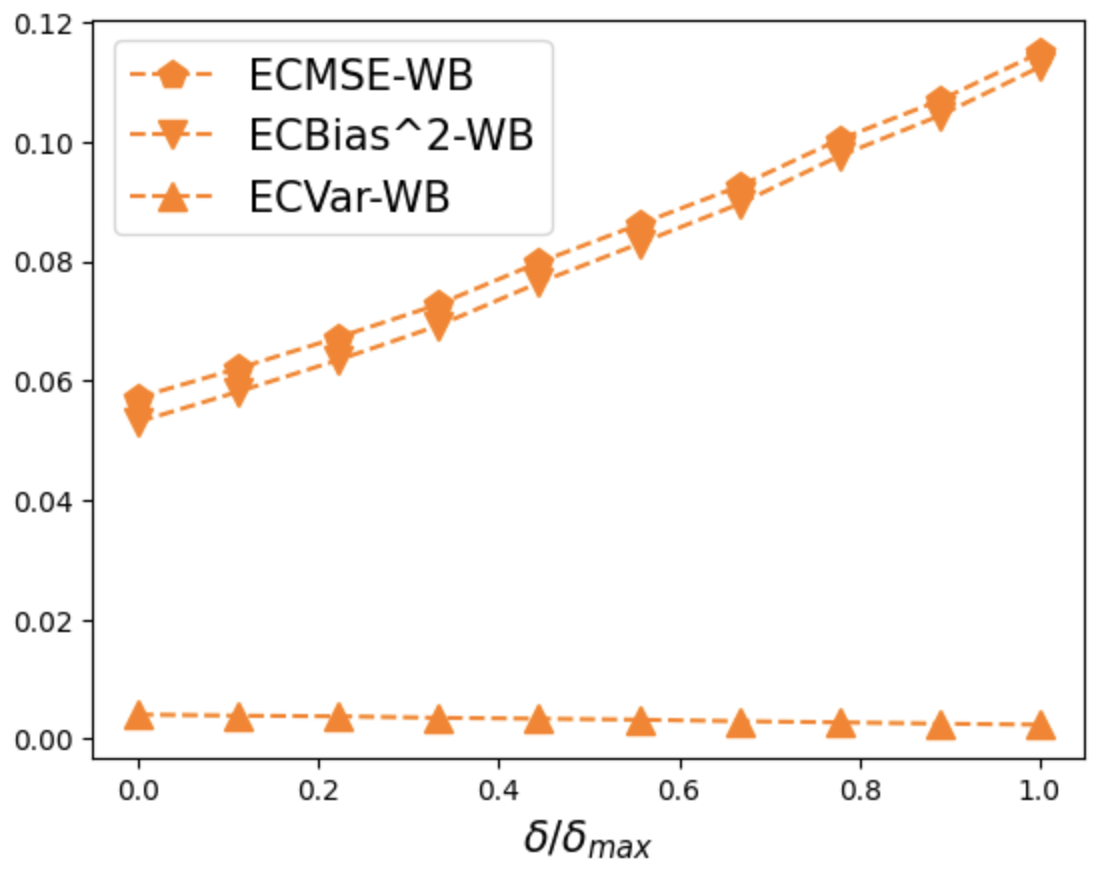}} \hspace{1mm}
	\subfigure[Missing at random]{\label{fig:real-mar-17-wb}
	\includegraphics[width=0.45\columnwidth]{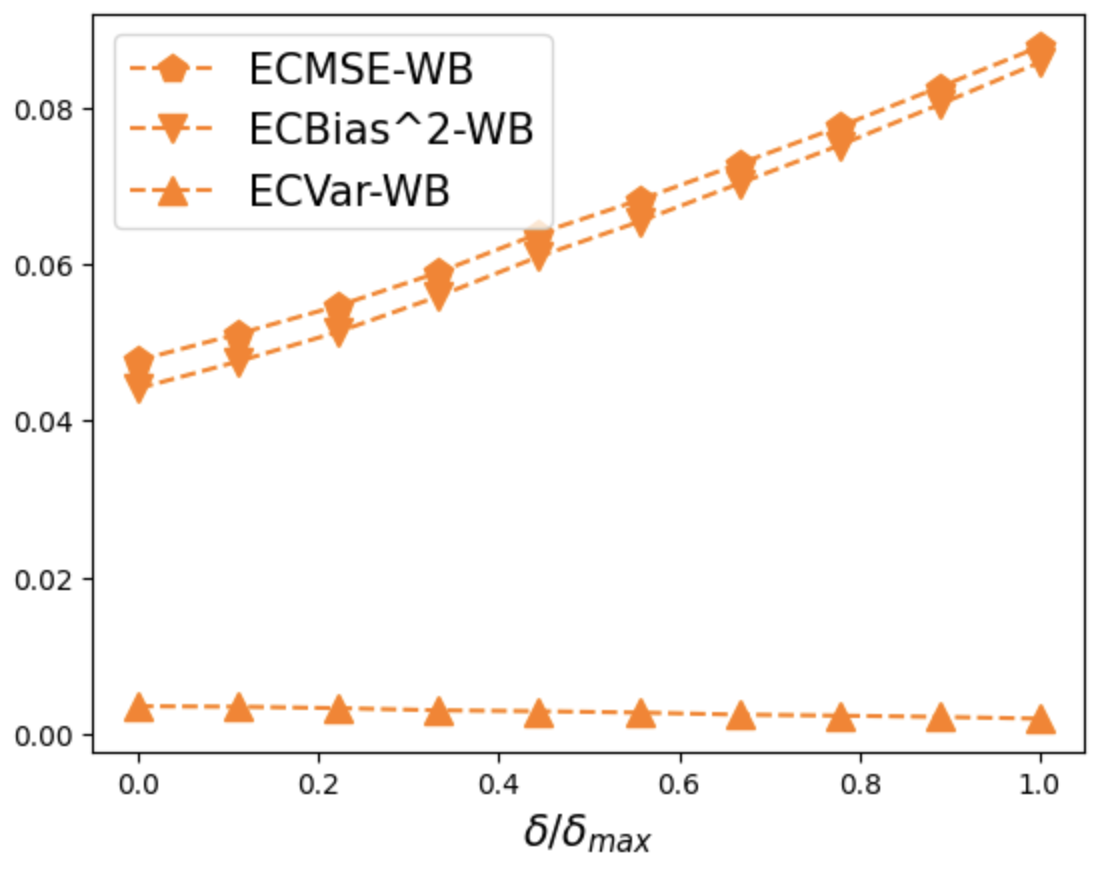}}
	    \hspace{1mm}
	\subfigure[Block missing]{\label{fig:real-block-17-wb} 
		\includegraphics[width=0.45\columnwidth]{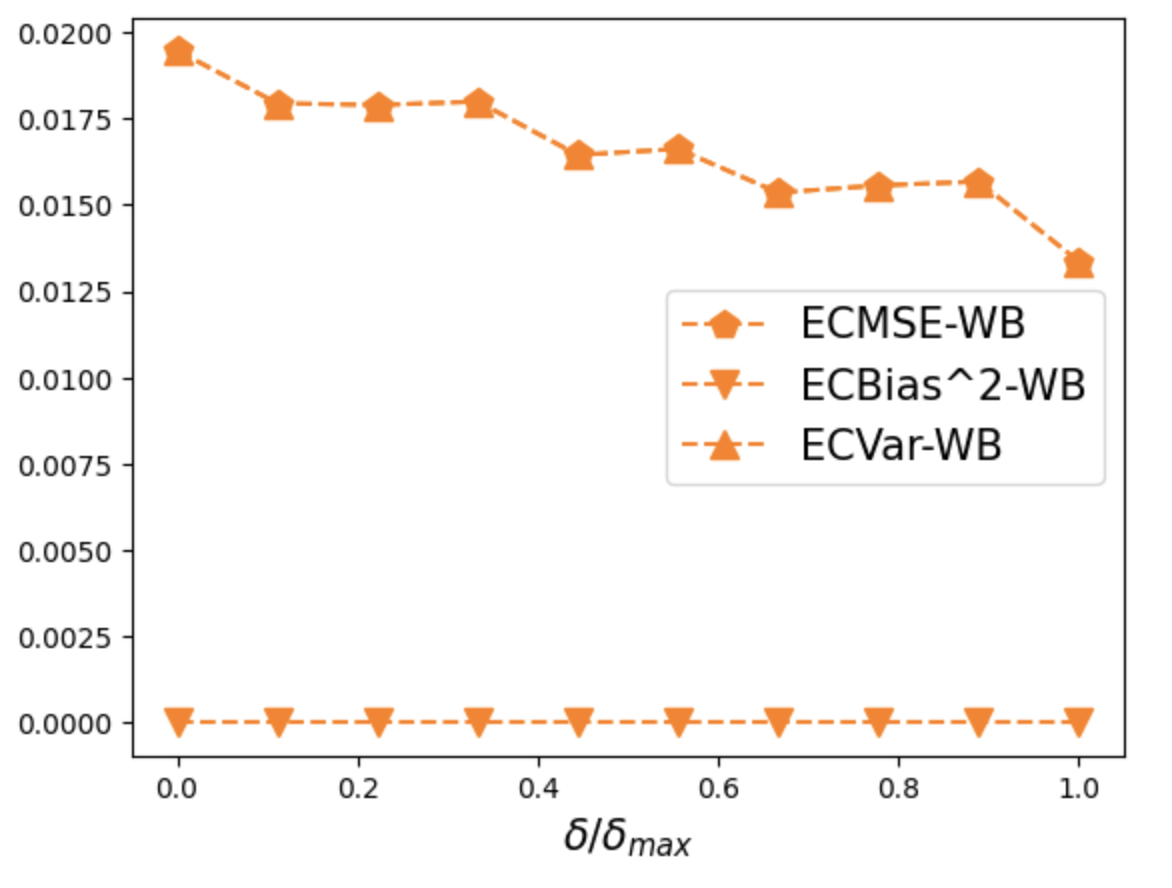}} 
		\subfigure[Missing by value]{\label{fig:real-nonrandom-17-wb} 
		\includegraphics[width=0.45\columnwidth]{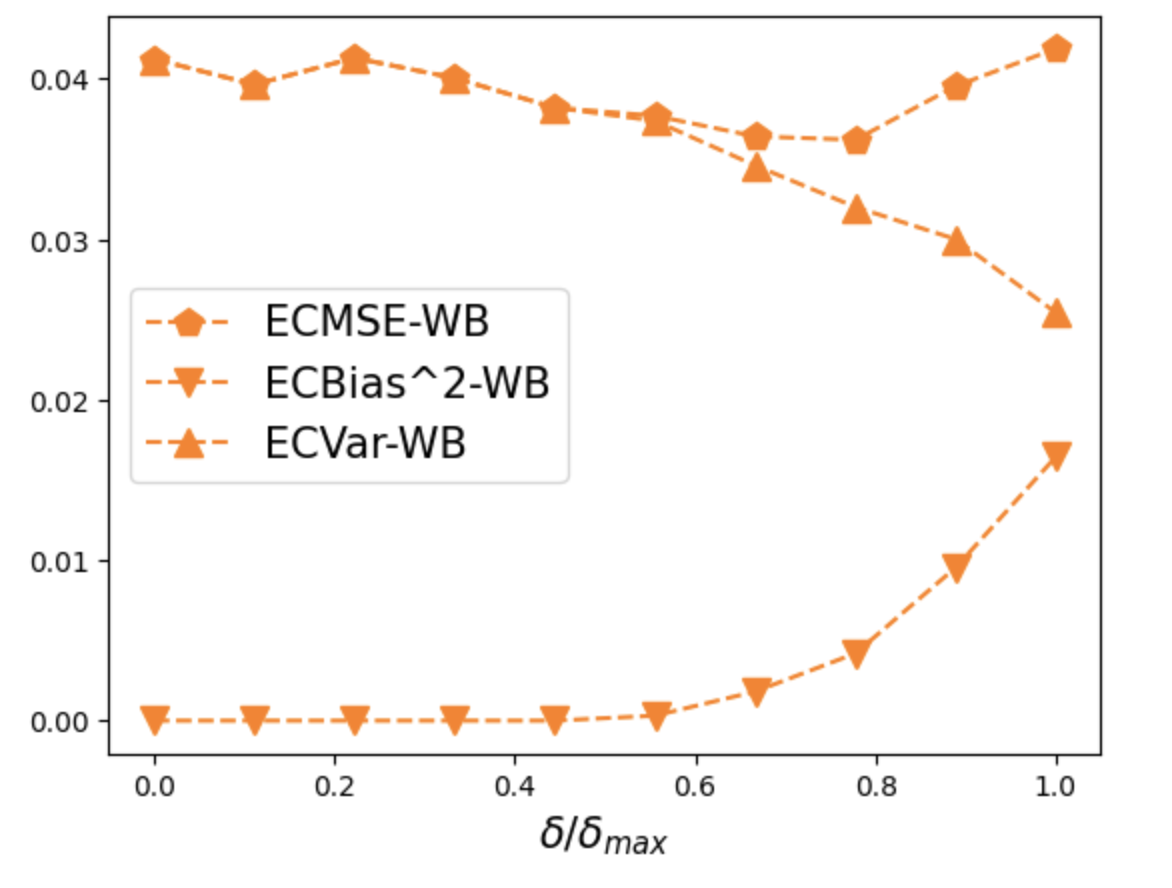}} 
		\bnotefig{These figures show the look-ahead-bias measured by the regret ECBias$^2$, the variance ECVar and the expected conditional mean squared error ECMSE for industry portfolios with the restricted Wasserstein mechanism for different weights $\delta$, which control the trade-off between look-ahead-bias and variance. The samples are daily returns of 10 industry portfolios from  01/03/2017 to 01/02/2019.}
    \label{fig:real17-wb}
\end{figure}

%\bibliographystyle{siam}
%\bibliography{bibliography}
\end{document}